\documentclass[12pt,technote, onecolumn, draft]{IEEEtran}
\usepackage{latexsym}
\usepackage{mathrsfs}
\usepackage{multirow}
\usepackage{amsfonts,amssymb,amsmath,amsthm,bm}
\usepackage{color}
\usepackage{cite}
\usepackage{longtable}
\usepackage{threeparttable}
\usepackage[figuresright]{rotating}
\usepackage{makecell}
\newtheorem{theorem}{Theorem}
\newtheorem{example}{Example}

\newtheorem{remark}{Remark}
\newtheorem{corollary}{Corollary}
\newtheorem{lemma}{Lemma}
\newtheorem{proposition}{Proposition}

\begin{document}

\title{Self-Orthogonal Codes from Vectorial Dual-Bent Functions$^{\dag}$} 
\author{Jiaxin Wang, Yadi Wei, Fang-Wei Fu, Juan Li
\IEEEcompsocitemizethanks{\IEEEcompsocthanksitem Jiaxin Wang, Yadi Wei and Fang-Wei Fu are with Chern Institute of Mathematics and LPMC, Nankai University, Tianjin 300071, China, Emails: wjiaxin@mail.nankai.edu.cn, wydecho@mail.nankai.edu.cn, fwfu@nankai.edu.cn, Juan Li is with School of Mathematics and Statistics, Linyi University, Linyi 276000, China, Email: lijuansl@163.com.
}
\thanks{$^\dag$This research is supported by the National Key Research and Development Program of China (Grant Nos. 2022YFA1005000 and 2018YFA0704703), the National Natural Science Foundation of China (Grant Nos. 12141108, 62371259, 12226336), the Postdoctoral Fellowship Program of CPSF (Grant No. GZC20231177), the Fundamental Research Funds for the Central Universities of China (Nankai University), the Nankai Zhide Foundation, and the Youth Fundamental of Shandong Natural Science Foundation (Grant No. ZR2022QA067).}
\thanks{manuscript submitted  March 19, 2024}
}

\maketitle

\begin{abstract}
  Self-orthogonal codes are a significant class of linear codes in coding theory and have attracted a lot of attention. In \cite{HLL2023Te,LH2023Se}, $p$-ary self-orthogonal codes were constructed by using $p$-ary weakly regular bent functions, where $p$ is an odd prime. In \cite{WH2023Se}, two classes of non-degenerate quadratic forms were used to construct $q$-ary self-orthogonal codes, where $q$ is a power of a prime. In this paper, we construct new families of $q$-ary self-orthogonal codes using vectorial dual-bent functions. Some classes of at least almost optimal linear codes are obtained from the dual codes of the constructed self-orthogonal codes. In some cases, we completely determine the weight distributions of the constructed self-orthogonal codes. From the view of vectorial dual-bent functions, we illustrate that the works on constructing self-orthogonal codes from $p$-ary weakly regular bent functions \cite{HLL2023Te,LH2023Se} and non-degenerate quadratic forms with $q$ being odd  \cite{WH2023Se} can be obtained by our results.  We partially answer an open problem on determining the weight distribution of a class of self-orthogonal codes given in \cite{LH2023Se}. As applications, we construct new infinite families of at least almost optimal $q$-ary linear complementary dual codes (for short, LCD codes) and quantum codes.
\end{abstract}


\section{Introduction}
\label{sec:1}
Let $\mathbb{F}_{q}^{n}$ be the vector space of the $n$-tuples over the finite field $\mathbb{F}_{q}$, where $q$ is a power of a prime $p$. A \emph{$q$-ary $[n, k]$ linear code} is a subspace of $\mathbb{F}_{q}^{n}$ with dimension $k$. Linear codes play an important role in coding theory. Among the known methods to construct linear codes, one of which is based on cryptographic functions, such as $p$-ary bent functions \cite{Ding2015Li,DD2015A,HYL2016Th,Mesnager2017Li,OP2022Tw,TLQZH2016Li,WLZ2020Li2,XCX2017Tw,XQL2022Mi,ZLFH2016Li}, vectorial bent functions \cite{CDY2005Li,FL2007Va,LLQ2009On,WLZ2020Li1,WSWF2023Co,XTD2022Sh,YCD2006Th}, almost perfect nonlinear functions \cite{CCZ1998Co,Ding2016A,XTD2022Sh}, and $p$-ary plateaued functions \cite{MOS2019Li,MS2020Se}.

The \emph{dual code} $C^{\perp}$ of a $q$-ary $[n, k]$ linear code $C$ is defined as $C^{\perp}=\{u\in \mathbb{F}_{q}^{n}:u \cdot c=0\ \text{for all} \ c\in C\}$, where $\cdot$ is the standard inner product on $\mathbb{F}_{q}^{n}$. If $C\subseteq C^{\bot}$, then $C$ is called \emph{self-orthogonal}. Self-orthogonal codes have significant applications in quantum codes \cite{CRSS1997Qu}, linear complementary dual codes (for short, LCD codes), row-self-orthogonal matrices \cite{Massey1998Or}, and even lattices \cite{Wan1998A}.

Very recently, Heng, Li, and Liu in \cite{HLL2023Te}, Li and Heng in \cite{LH2023Se} considered using weakly regular $p$-ary bent functions of $l$-form with $gcd(l-1, p-1)=1$ to construct $p$-ary self-orthogonal codes. In \cite{WH2023Se}, Wang and Heng utilized two classes of non-degenerate quadratic forms to construct $q$-ary self-orthogonal codes. By the results in \cite{CM2018Be, WF2023Ne,WSWF2023Co}, weakly regular $p$-ary bent functions of $l$-form with $gcd(l-1, p-1)=1$ and non-degenerate quadratic forms with $q$ being odd are actually vectorial dual-bent functions introduced in \cite{CMP2018Ve}. Hence, it is interesting to investigate whether there are general results on constructing self-orthogonal codes from vectorial dual-bent functions. In this paper, we construct new families of $q$-ary self-orthogonal codes from vectorial dual-bent functions whose parameters are more abundant and flexible. In some cases, the weight distributions of the constructed self-orthogonal codes are completely determined. Some classes of at least almost optimal linear codes are obtained from the dual codes of the constructed self-orthogonal codes. By using a class of vectorial dual-bent functions, some optimal linear codes or having best parameters up to now are listed. In particular, we explain that the self-orthogonal codes from $p$-ary weakly regular bent functions \cite{HLL2023Te,LH2023Se} and non-degenerate quadratic forms with $q$ being odd  \cite{WH2023Se} can be obtained by our results. We partially answer an open problem on determining the weight distribution of a class of self-orthogonal codes constructed in \cite{LH2023Se}. Moreover, based on the constructed $q$-ary self-orthogonal codes, new infinite families of $q$-ary LCD codes and quantum codes are obtained which are at least almost optimal by the Hamming bound and quantum Hamming bound, respectively.

The rest of the paper is organized as follows. In Section II, the needed preliminaries are introduced. In Sections III-V, we construct new families of $q$-ary self-orthogonal codes using vectorial dual-bent functions with certain conditions. In Section VI, we compare our constructed self-orthogonal codes with the known ones constructed from (vectorial) bent functions. In Section VII, LCD codes and quantum codes are given based on the constructed self-orthogonal codes. In Section VIII, we make a conclusion.

\section{Preliminaries}
\label{sec:2}
In this section, we introduce some notations and results on vectorial dual-bent functions, linear codes and character sums.
\subsection{Notations}\label{Sec: 2.1}
We fix some notations used in the sequel unless otherwise stated.
\begin{itemize}
  \item $q=p^t$, $p$ is a prime.
  \item $\epsilon=1$ if $p\equiv 1 \pmod 4$, $\epsilon=\sqrt{-1}$ if $p\equiv 3 \pmod 4$.
  \item $\zeta_{p}=e^{\frac{2 \pi \sqrt{-1}}{p}}$ is a complex primitive $p$-th root of unity.
  \item $\mathbb{F}_{q}$ is the finite field with $q$ elements.
  \item $\mathbb{F}_{q}^{n}$ is the vector space of the $n$-tuples over $\mathbb{F}_{q}$.
  \item $V_{n}^{(p)}$ is an $n$-dimensional vector space over $\mathbb{F}_{p}$.
  \item $\langle, \rangle_{n}$ denotes a (non-degenerate) inner product of $V_{n}^{(p)}$. In this paper, when $V_{n}^{(p)}=\mathbb{F}_{p}^{n}$, let $\langle a, b\rangle_{n}=a \cdot b=\sum_{i=1}^{n}a_{i}b_{i}$, where $a=(a_{1}, \dots, a_{n}), b=(b_{1}, \dots, b_{n})\in \mathbb{F}_{p}^{n}$; when $V_{n}^{(p)}=\mathbb{F}_{p^n}$, let $\langle a, b\rangle_{n}=Tr_{1}^{n}(ab)$, where $a, b \in \mathbb{F}_{p^n}$, $Tr_{t}^{n}$ denotes the trace function from $\mathbb{F}_{p^n}$ to $\mathbb{F}_{p^t}$, $t \mid n$; when $V_{n}^{(p)}=V_{n_{1}}^{(p)}\times \dots \times V_{n_{s}}^{(p)}$, let $\langle a, b\rangle_{n}=\sum_{i=1}^{s}\langle a_{i}, b_{i}\rangle_{n_{i}}$, where $a=(a_{1}, \dots, a_{s}), b=(b_{1}, \dots, b_{s})\in V_{n}^{(p)}$.
  \item If $V_{n}^{(p)}=\mathbb{F}_{p^{n_{1}}} \times \mathbb{F}_{p^{n_{2}}} \times \dots \times \mathbb{F}_{p^{n_{s}}}$ and $x \in V_{n}^{(p)}$, denote $x=(x_{1}, \dots, x_{s})$, where $x_{j} \in \mathbb{F}_{p^{n_{j}}}, 1\leq j \leq s$.
  \item For $x=0 \in \mathbb{F}_{p^n}$, for convention we denote $x^{-1}=x^{p^n-2}=0$.
  \item For a function $F: V_{n}^{(p)} \rightarrow V_{m}^{(p)}$ and any $A\subseteq V_{m}^{(p)}$, let $D_{F, A}=\{x \in V_{n}^{(p)}: F(x) \in A\}$. When $A=\{a\}$, simply denote $D_{F, \{a\}}$ by $D_{F, a}$.
  \item For any set $A$, $\delta_{A}$ is the indicator function. When $A=\{a\}$, simply denote $\delta_{\{a\}}$ by $\delta_{a}$.
  \item For any set $A \subseteq V_{n}^{(p)}$ and $a \in V_{n}^{(p)}$, let $\chi_{a}(A)=\sum_{x \in A}\chi_{a}(x)$, where $\chi_{a}$ is the character defined by $\chi_{a}(x)=\zeta_{p}^{\langle a, x\rangle_{n}}$.
  \item For $a \in \mathbb{F}_{p^n}$, if $a=0$, then $\eta_{n}(a)=0$; if $a$ is a square in $\mathbb{F}_{p^n}^{*}$, then $\eta_{n}(a)=1$; if $a$ is a non-square in $\mathbb{F}_{p^n}^{*}$, then $\eta_{n}(a)=-1$.
  \item $\textbf{1}$ denotes all one vector, that is, $\textbf{1}=(1, 1, \dots, 1)$.
\end{itemize}

\subsection{Some results on vectorial dual-bent functions} \label{Sec: 2.2}
A function $F: V_{n}^{(p)}\rightarrow V_{m}^{(p)}$ is called a \emph{vectorial $p$-ary function}, or simply \emph{$p$-ary function} when $m=1$. For a vectorial $p$-ary function $F: V_{n}^{(p)}\rightarrow V_{m}^{(p)}$ with $m\leq n$, if $|D_{F, i}|=p^{n-m}$ for any $i \in V_{m}^{(p)}$, then $F$ is called \emph{balanced}.

For a $p$-ary function $f: V_{n}^{(p)}\rightarrow \mathbb{F}_{p}$, the \emph{Walsh transform} of $f$ is defined as
\begin{equation*}
W_{f}(a)=\sum_{x \in V_{n}^{(p)}}\zeta_{p}^{f(x)-\langle a, x\rangle_{n}}, a \in V_{n}^{(p)},
\end{equation*}
and its inverse Walsh transform is given by
\begin{equation*}
\zeta_{p}^{f(x)}=\frac{1}{p^n}\sum_{a \in V_{n}^{(p)}}W_{f}(a)\zeta_{p}^{\langle a, x\rangle_{n}}, x \in V_{n}^{(p)}.
\end{equation*}
If for all $a \in V_{n}^{(p)}$, $|W_{f}(a)|=p^{\frac{n}{2}}$, then $f$ is called a \emph{$p$-ary bent function}. For any $p$-ary bent function $f: V_{n}^{(p)}\rightarrow \mathbb{F}_{p}$, its Walsh transform satisfies that when $p=2$, $W_{f}(a)=2^{\frac{n}{2}}(-1)^{f^{*}(a)}$, and when $p$ is an odd prime,
\begin{equation*}
  W_{f}(a)=\left\{\begin{split}
                     \pm p^{\frac{n}{2}}\zeta_{p}^{f^{*}(a)}, & \ \text{  if } p \equiv 1 \pmod 4 \text{ or } n \text{ is even},\\
                     \pm \sqrt{-1} p^{\frac{n}{2}} \zeta_{p}^{f^{*}(a)}, & \ \text{  if } p \equiv 3 \pmod 4 \text{ and } n \text{ is odd},
                  \end{split}\right.
\end{equation*}
where $f^{*}$ is a $p$-ary function from $V_{n}^{(p)}$ to $\mathbb{F}_{p}$, called the \emph{dual} of $f$. A $p$-ary bent function $f: V_{n}^{(p)}\rightarrow \mathbb{F}_{p}$ is called \emph{weakly regular} if $W_{f}(a)=\varepsilon_{f}p^{\frac{n}{2}}\zeta_{p}^{f^{*}(a)}$, where $\varepsilon_{f} \in \{\pm 1, \pm \sqrt{-1}\}$ is a constant, otherwise $f$ is called \emph{non-weakly regular}. In particular, if $W_{f}(a)=p^{\frac{n}{2}}\zeta_{p}^{f^{*}(a)}$, that is, $\varepsilon_{f}=1$, then $f$ is called \emph{regular}. Any $2$-ary bent function, that is, Boolean bent function, is regular. For a $p$-ary weakly regular bent function $f$, its dual $f^{*}$ is also a weakly regular bent function with
\begin{equation*}
(f^{*})^{*}(x)=f(-x), \varepsilon_{f^{*}}=\varepsilon_{f}^{-1}.
\end{equation*}

For a vectorial $p$-ary function $F: V_{n}^{(p)}\rightarrow V_{m}^{(p)}$, if for any $c \in V_{m}^{(p)} \backslash \{0\}$, the component function $F_{c}$ defined as $F_{c}(x)=\langle c, F(x)\rangle_{m}$ is a $p$-ary bent function, then $F$ is called \emph{vectorial bent}. Every $p$-ary bent function is vectorial bent. A vectorial $p$-ary bent function $F: V_{n}^{(p)}\rightarrow V_{m}^{(p)}$ is called \emph{vectorial dual-bent} if there exists a vectorial bent function $G: V_{n}^{(p)}\rightarrow V_{m}^{(p)}$ such that $(F_{c})^{*}=G_{\sigma(c)}$ for any $c \in V_{m}^{(p)} \backslash \{0\}$, where $(F_{c})^{*}$ is the dual of $F_{c}$, and $\sigma$ is some permutation over $V_{m}^{(p)} \backslash \{0\}$. The vectorial bent function $G$ is called a \emph{vectorial dual} of $F$ and denoted by $F^{*}$.

A $p$-ary function $f: V_{n}^{(p)}\rightarrow \mathbb{F}_{p}$ is called of \emph{$l$-form} if $f(ax)=a^{l}f(x)$ for any $a \in \mathbb{F}_{p}^{*}$ and $x \in V_{n}^{(p)}$, where $l$ is an integer. In a number of papers, a weakly regular bent function $f: V_{n}^{(p)} \rightarrow \mathbb{F}_{p}$ of $l$-form with $f(0)=0$ and $gcd(l-1, p-1)=1$ is called a \emph{bent function belonging to $\mathcal{RF}$}.

\begin{proposition} [\cite{CM2018Be}]\label{Proposition 1}
Let $f: V_{n}^{(p)}\rightarrow \mathbb{F}_{p}$ be a bent function belonging to $\mathcal{RF}$, that is, $f$ is a weakly regular bent function of $l$-form with $f(0)=0$ and $gcd(l-1, p-1)=1$ for some $l$. Then $f$ (seen as a vectorial bent function from $V_{n}^{(p)}$ to $V_{1}^{(p)}$) is a vectorial dual-bent function with $(cf)^{*}=c^{1-d}f^{*}, c \in \mathbb{F}_{p}^{*}$, and $\varepsilon_{cf}=\varepsilon_{f}$ if $n$ is even,  $\varepsilon_{cf}=\varepsilon_{f}\eta_{1}(c)$ if $n$ is odd, where $(l-1)(d-1)\equiv 1 \mod (p-1)$.
\end{proposition}

\subsection{Some results on linear codes}
\label{2.3}
For a vector $a=(a_{1}, \dots, a_{n})\in \mathbb{F}_{q}^{n}$, the \emph{Hamming weight} of $a$, denoted by $wt(a)$, is the size of its support $supp(a)=\{1\leq i\leq n: a_{i}\neq 0\}$. For two vectors $a, b \in \mathbb{F}_{q}^{n}$, the \emph{Hamming distance} $d(a, b)$ between $a$ and $b$ is defined as $d(a, b)=wt(a-b)$. For a $q$-ary $[n, k]$ linear code $C$, the \emph{minimum Hamming distance} $d$ of $C$ is defined as $d=\min\{d(a, b): a, b \in C, a \neq b\}=\min\{wt(c): c\in C, c \neq 0\}$, and $C$ is denoted as an $[n, k, d]_{q}$ linear code. For any $1\leq i \leq n$, let $A_{i}$ denote the number of codewords in the linear code $C$ whose Hamming weight is $i$. The sequence $(1, A_{1}, \dots, A_{n})$ is called the \emph{weight distribution} of $C$. The code $C$ is called \emph{$t$-weight} if $|\{1\leq i\leq n: A_{i}\neq 0\}|=t$. For an $[n, k, d]_{q}$ linear code $C$, it is called \emph{(distance) optimal} if there is no $[n, k, d+1]_{q}$ linear code, and is called \emph{almost optimal} if there is an $[n, k, d+1]_{q}$ optimal code.

We recall the well-known Hamming bound on linear codes (see e.g. \cite{HP2003Fu}).

\begin{proposition}[Hamming Bound]\label{Proposition 2}
Let $C$ be an $[n, k, d]_{q}$ linear code. Then
\begin{equation*}
q^{n-k}\geq \sum_{i=0}^{\lfloor\frac{d-1}{2}\rfloor}\binom{n}{i}(q-1)^{i}.
\end{equation*}
\end{proposition}

When $q$ is odd, the following proposition gives a relatively simple way to show the self-orthogonality of linear codes.

\begin{proposition}[\cite{Wan1998A}]\label{Proposition 3}
Let $q$ be a power of an odd prime and $C$ be a $q$-ary linear code. Then $C$ is self-orthogonal if and only if $c\cdot c=0$ for all $c \in C$.
\end{proposition}

Let $C$ be an $[n, k, d]_{q}$ linear code and $C^{\perp}$ be its dual code. If $C\cap C^{\perp}=\{\textbf{0}\}$, then $C$ is called a \emph{linear complementary dual code} (for short, \emph{LCD code}). The dual code $C^{\perp}$ of an LCD code $C$ is also an LCD code. There is a method to construct LCD codes by using self-orthogonal codes.

\begin{proposition}[\cite{Massey1998Or}]\label{Proposition 4}
If $C$ is a self-orthogonal $[n, k]_{q}$ linear code with generator matrix $G$, then the linear code $C'$ with generator matrix $G'=[I_{k}, G]$ is an $[n+k, k]_{q}$ LCD code, where $I_{k}$ is the identity matrix of size $k \times k$.
\end{proposition}

Quantum codes are used to detect and correct errors caused by quantum noise in quantum communication. An \emph{$[[n, k, d]]_{q}$ quantum error-correcting code} (for short, \emph{quantum code}) of length $n$ and minimum distance $d$ is a $K$-dimensional subspace of the Hilbert space $\mathbb{C}^{q^n}$, where $K=q^k$. For the details, please refer to \cite{Feng2002Qu}. An $[[n, k, d]]_{q}$ quantum code $C$ is called \emph{(distance) optimal} if there is no $[[n, k, d+1]]_{q}$ quantum code, and is called \emph{almost optimal} if there is an $[[n, k, d+1]]_{q}$ optimal quantum code.

We recall the well-known quantum Hamming bound on pure quantum codes.

\begin{proposition}[Quantum Hamming Bound \cite{KKK2006No}]\label{Proposition 5}
If $C$ is an $[[n, k, d]]_{q}$ pure quantum code, then
\begin{equation*}
q^{n-k}\geq \sum_{i=0}^{\lfloor\frac{d-1}{2}\rfloor}\binom{n}{i}(q^{2}-1)^{i}.
\end{equation*}
\end{proposition}

We recall the well-known Steane construction on quantum codes.

\begin{proposition}[Steane construction \cite{Hamada2008Co}]\label{Propsition 6}
Let $C_{1}$ and $C_{2}$ be $[n, k_{1}, d_{1}]_{q}$ and $[n, k_{2}, d_{2}]_{q}$ linear codes, respectively. If $C_{1}^{\perp} \subseteq C_{1} \subseteq C_{2}$ and $k_{1}+2\leq k_{2}$, then there exists an $[[n, k_{1}+k_{2}-n, \min\{d_{1}, \lceil \frac{q+1}{q}d_{2}\rceil\}]]_{q}$ pure quantum code.
\end{proposition}

Let $t, n, n_{j}, 1\leq j \leq s$, be positive integers with $n=\sum_{j=1}^{s}n_{j}, t \mid n_{j}, 1\leq j \leq s$, and let $V_{n}^{(p)}=\mathbb{F}_{p^{n_{1}}} \times \mathbb{F}_{p^{n_{2}}} \times \dots \times \mathbb{F}_{p^{n_{s}}}$. For a function $F: V_{n}^{(p)}\rightarrow V_{m}^{(p)}$ and a nonempty set $I\subset V_{m}^{(p)}$, define
\begin{equation}\label{1}
C_{D_{F, I}}=\{c_{\alpha, \beta}=(\sum_{j=1}^{s}Tr_{t}^{n_{j}}(\alpha_{j}x_{j}))_{x\in D_{F, I}}+\beta \textbf{1}: \alpha \in V_{n}^{(p)}, \beta \in \mathbb{F}_{p^t}\}.
\end{equation}
In this paper, we will construct $p^{t}$-ary self-orthogonal codes from vectorial dual-bent functions with certain conditions.

\subsection{Some results on character sums}
\label{sec:2.4}

In this subsection, we recall some well-known results on character sums.

\begin{proposition}[\cite{LN1997Fi}]\label{Proposition 7}
Let $q=p^m$, where $p$ is an odd prime. For any $a \in \mathbb{F}_{q}$,
\begin{equation*}
  \sum_{x \in \mathbb{F}_{q}^{*}}\eta_{m}(x)\zeta_{p}^{Tr_{1}^{m}(ax)}=(-1)^{m-1}\epsilon^{m}\eta_{m}(a)\sqrt{q}.
\end{equation*}
\end{proposition}

\begin{proposition}[\cite{LN1997Fi}]\label{Proposition 8}
Let $f(x)=a_{2}x^{2}+a_{1}x+a_{0} \in \mathbb{F}_{p^m}[x]$ with $p$ being odd and $a_{2} \neq 0$. Then
\begin{equation*}
\sum_{x \in \mathbb{F}_{p^m}}\eta_{m}(f(x))=\left\{
\begin{split}
-\eta_{m}(a_{2}), & \ \text{ if } a_{1}^{2}-4a_{0}a_{2}\neq 0, \\
(p^m-1)\eta_{m}(a_{2}), & \ \text{ if } a_{1}^{2}-4a_{0}a_{2}=0.
\end{split}
\right.
\end{equation*}
\end{proposition}

\begin{proposition} [\cite{Ding2015Co}]\label{Proposition 9}
Let $m, b$ be positive integers which satisfy that $m=2jj'$ for some positive integers $j, j'$, $b\geq 2$ and $b \mid (p^{j}+1)$, where $j$ is the smallest such positive integer. Let $H_{b}=\{x^{b}: x \in \mathbb{F}_{p^m}^{*}\}$ and $w$ be any fixed primitive element of $\mathbb{F}_{p^m}$. For any $a \in \mathbb{F}_{p^m}^{*}$,
\begin{equation*}
  \sum_{x \in H_{b}}\zeta_{p}^{Tr_{1}^{m}(a x)}=\left\{
  \begin{split}
  \delta_{w^{\frac{b}{2}}H_{b}}(a)p^{\frac{m}{2}}-\frac{p^\frac{m}{2}+1}{b}, & \ \text{ if } \ p, j', \frac{p^j+1}{b}  \text{ are all odd},\\
  \delta_{H_{b}}(a)(-1)^{j'+1}p^{\frac{m}{2}}+\frac{(-1)^{j'}p^{\frac{m}{2}}-1}{b}, & \ \text{otherwise}.
  \end{split}
  \right.
\end{equation*}
\end{proposition}

\begin{proposition}[\cite{Ding2015Co}] \label{Propposition 10}
Let $q$ be a power of an odd prime, $\mathcal{S}=\{x^{2}: x \in \mathbb{F}_{q}^{*}\}$ and $\mathcal{N}=\mathbb{F}_{q}^{*} \backslash \mathcal{S}$. Then when $q\equiv 1 \pmod 4$, $|(\mathcal{S}+1)\cap \mathcal{S}|=\frac{q-5}{4}$, $|(\mathcal{S}+1)\cap \mathcal{N}|=|(\mathcal{N}+1)\cap \mathcal{S}|=|(\mathcal{N}+1)\cap \mathcal{N}|=\frac{q-1}{4}$; when $q\equiv 3 \pmod 4$, $|(\mathcal{S}+1)\cap \mathcal{N}|=\frac{q+1}{4}$, $|(\mathcal{S}+1)\cap \mathcal{S}|=|(\mathcal{N}+1)\cap \mathcal{S}|=|(\mathcal{N}+1)\cap \mathcal{N}|=\frac{q-3}{4}$.
\end{proposition}

\section{Self-orthogonal codes from vectorial dual-bent functions with Condition I}
\label{sec:3}

In this section, we construct self-orthogonal codes from vectorial dual-bent functions with the following condition:

Condition I: Let $n, n_{j}, 1\leq j \leq s, m, t$ be positive integers for which $n=\sum_{j=1}^{s}n_{j}, 2 \mid n, t \mid n_{j}, 1\leq j \leq s, t\leq \frac{n}{2}, m<\frac{n}{2}$, and $m\geq 2$ when $p=2$, and let $V_{n}^{(p)}=\mathbb{F}_{p^{n_{1}}} \times \mathbb{F}_{p^{n_{2}}} \times \dots \times \mathbb{F}_{p^{n_{s}}}$. Let $F: V_{n}^{(p)}\rightarrow V_{m}^{(p)}$ be a vectorial dual-bent function satisfying
\begin{itemize}
  \item There is a vectorial dual $F^{*}$ such that $(F_{c})^{*}=(F^{*})_{c}, c \in V_{m}^{(p)}\backslash \{0\}$;
  \item $F(ax)=F(x), a \in \mathbb{F}_{p^t}^{*}, x \in V_{n}^{(p)}$;
  \item All component functions $F_{c}, c \in V_{m}^{(p)}\backslash \{0\}$, are weakly regular with $\varepsilon_{F_{c}}=\varepsilon, c \in V_{m}^{(p)}\backslash \{0\}$, where $\varepsilon \in \{\pm1\}$ is a constant.
\end{itemize}

\subsection{Some lemmas}\label{3.1}
In this subsection, we give some useful lemmas.

\begin{lemma}\label{Lemma 1}
Let $F$ be a vectorial dual-bent function with Condition I. Then the vectorial dual $F^{*}$ with $(F_{c})^{*}=(F^{*})_{c}, c \in V_{m}^{(p)}\backslash \{0\}$, is a vectorial dual-bent function with Condition I.
\end{lemma}

\begin{proof}
For any $c \in V_{m}^{(p)} \backslash \{0\}$, since $F_{c}$ is weakly regular bent with $\varepsilon_{F_{c}}=\varepsilon \in \{\pm 1\}$, $(F^{*})_{c}=(F_{c})^{*}$ is weakly regular bent with $((F^{*})_{c})^{*}(x)=((F_{c})^{*})^{*}(x)=F_{c}(-x)=F_{c}(x)$ and $\varepsilon_{(F^{*})_{c}}=\varepsilon$. For any $c \in V_{m}^{(p)} \backslash \{0\}$ and $a \in \mathbb{F}_{p^t}^{*}, x \in V_{n}^{(p)}$, we have
\begin{small}
\begin{equation*}
\begin{split}
p^{n}\zeta_{p}^{(F^{*})_{c}(ax)}&=\sum_{y \in V_{n}^{(p)}}W_{(F^{*})_{c}}(y)\zeta_{p}^{\sum_{j=1}^{s}Tr_{1}^{n_{j}}(ax_{j}y_{j})}=\varepsilon p^{\frac{n}{2}}\sum_{y \in V_{n}^{(p)}}\zeta_{p}^{F_{c}(y)+\sum_{j=1}^{s}Tr_{1}^{n_{j}}(ax_{j}y_{j})}\\
&=\varepsilon p^{\frac{n}{2}}\sum_{y \in V_{n}^{(p)}}\zeta_{p}^{F_{c}(ay)+\sum_{j=1}^{s}Tr_{1}^{n_{j}}(ax_{j}y_{j})}=\varepsilon p^{\frac{n}{2}}\sum_{y \in V_{n}^{(p)}}\zeta_{p}^{F_{c}(y)+\sum_{j=1}^{s}Tr_{1}^{n_{j}}(x_{j}y_{j})}\\
&=\sum_{y \in V_{n}^{(p)}}W_{(F^{*})_{c}}(y)\zeta_{p}^{\sum_{j=1}^{s}Tr_{1}^{n_{j}}(x_{j}y_{j})}=p^{n}\zeta_{p}^{(F^{*})_{c}(x)},
\end{split}
\end{equation*}
\end{small}where in the third equation we use $F(ax)=F(x)$. Therefore, $(F^{*})_{c}(ax)=(F^{*})_{c}(x), a \in \mathbb{F}_{p^t}^{*}, x \in V_{n}^{(p)}$, for all $c \in V_{m}^{(p)} \backslash \{0\}$, which implies that $F^{*}(ax)=F^{*}(x), a \in \mathbb{F}_{p^t}^{*}, x \in V_{n}^{(p)}$. Thus $F^{*}$ is a vectorial dual-bent function with Condition I.
\end{proof}

\begin{lemma} \label{Lemma 2}
Let $F$ be a vectorial dual-bent function with Condition I. Then the value distributions of $F$ and $F^{*}$ are given by
\begin{equation*}
|D_{F, i}|=|D_{F^{*}, i}|=p^{n-m}+\varepsilon p^{\frac{n}{2}-m}(p^m\delta_{F(0)}(i)-1), i \in V_{m}^{(p)}.
\end{equation*}
\end{lemma}

\begin{proof}
By Proposition 4 of \cite{WFW2023Be} and its proof, $F(x)-F(0)$ is a vectorial dual-bent function with Condition I, and the corresponding vectorial dual is $F^{*}(x)-F(0)$. By Corollary 2 and Proposition 5 of \cite{CMP2021Ve}, $F^{*}(0)=F(0)$. Then the result follows from Proposition 4 of \cite{WFW2023Be} and Lemma 1.
\end{proof}

\begin{lemma} \label{Lemma 3}
Let $\psi: V_{n}^{(p)}\rightarrow V_{m}^{(p)}$ be a vectorial dual-bent function with Condition I, $m'$ be a positive integer with $m'\leq m$, and $m'\neq 1$ when $p=2$. Then for any balanced function $B: V_{m}^{(p)}\rightarrow V_{m'}^{(p)}$, $F(x)=B(\psi(x))$ is a vectorial dual-bent function with Condition I.
\end{lemma}

\begin{proof}
Since $\psi(ax)=\psi(x), a \in \mathbb{F}_{p^t}^{*}, x \in V_{n}^{(p)}$, $F(ax)=F(x), a \in \mathbb{F}_{p^t}^{*}, x \in V_{n}^{(p)}$. For any $c \in V_{m'}^{(p)} \backslash \{0\}$ and $a \in V_{n}^{(p)}$, we have
\begin{small}
\begin{equation*}
  \begin{split}
  W_{F_{c}}(-a)&=\sum_{x \in V_{n}^{(p)}}\zeta_{p}^{F_{c}(x)+\langle a, x\rangle_{n}}=\sum_{x \in V_{n}^{(p)}}\zeta_{p}^{\langle c, B(\psi(x))\rangle_{m'}+\langle a, x\rangle_{n}}\\
  &=\sum_{i \in V_{m'}^{(p)}}\zeta_{p}^{\langle c, i\rangle_{m'}}\sum_{x \in D_{\psi, D_{B, i}}}\zeta_{p}^{\langle a, x\rangle_{n}}=\sum_{i \in V_{m'}^{(p)}}\zeta_{p}^{\langle c, i\rangle_{m'}}\sum_{j \in D_{B, i}}\chi_{a}(D_{\psi, j}).
  \end{split}
\end{equation*}
\end{small}By the proof of Lemma 1 of \cite{WFW2023Be}, we have $\chi_{a}(D_{\psi, j})=p^{n-m}\delta_{0}(a)+\varepsilon p^{\frac{n}{2}}\delta_{\psi^{*}(-a)}(j)-\varepsilon p^{\frac{n}{2}-m}$. By Lemma 1, $\psi^{*}(-a)=\psi^{*}(a)$. Therefore, for any $c \in V_{m'}^{(p)} \backslash \{0\}$, we have
\begin{small}
\begin{equation*}
\begin{split}
W_{F_{c}}(-a)&=\sum_{i \in V_{m'}^{(p)}}\zeta_{p}^{\langle c, i\rangle_{m'}}\sum_{j \in D_{B, i}}(p^{n-m}\delta_{0}(a)-\varepsilon p^{\frac{n}{2}-m})+\sum_{i \in V_{m'}^{(p)}}\zeta_{p}^{\langle c, i\rangle_{m'}}\sum_{j \in D_{B, i}}\varepsilon p^{\frac{n}{2}}\delta_{\psi^{*}(a)}(j)\\
&=p^{m-m'}(p^{n-m}\delta_{0}(a)-\varepsilon p^{\frac{n}{2}-m})\sum_{i \in V_{m'}^{(p)}}\zeta_{p}^{\langle c, i\rangle_{m'}}+\varepsilon p^{\frac{n}{2}}\sum_{i \in V_{m'}^{(p)}}\zeta_{p}^{\langle c, i\rangle_{m'}}\delta_{D_{B, i}}(\psi^{*}(a))\\
&=\varepsilon p^{\frac{n}{2}}\zeta_{p}^{\langle c, B(\psi^{*}(a))\rangle_{m'}}.
\end{split}
\end{equation*}
\end{small}Hence, $F$ is vectorial bent with $\varepsilon_{F_{c}}=\varepsilon, (F_{c})^{*}=(B(\psi^{*}))_{c}, c \in V_{m'}^{(p)} \backslash \{0\}$. Since $F_{c}$ is weakly regular bent, $(B(\psi^{*}))_{c}=(F_{c})^{*}$ is weakly regular bent, and $B(\psi^{*})$ is vectorial bent. Therefore, $F$ is a vectorial dual-bent function with Condition I and $F^{*}(x)=B(\psi^{*}(x))$.
\end{proof}

\begin{lemma} \label{Lemma 4}
Let $F: V_{n}^{(p)}\rightarrow V_{m}^{(p)}$ be a vectorial dual-bent function with Condition I.

$\mathrm{(i)}$ For any nonempty set $I \subset V_{m}^{(p)}$ and $\alpha \in V_{n}^{(p)} \backslash \{0\}, \beta \in \mathbb{F}_{p^t}$, define
\begin{equation*}
N_{I, \alpha, \beta}=|\{x \in V_{n}^{(p)}: F(x) \in I, \sum_{j=1}^{s}Tr_{t}^{n_{j}}(\alpha_{j}x_{j})+\beta=0\}|.
\end{equation*}
Then
\begin{equation*}
\begin{split}
N_{I, \alpha, \beta}&=\varepsilon p^{\frac{n}{2}-t}\delta_{I}(F^{*}(\alpha))(p^t\delta_{0}(\beta)-1)+\varepsilon p^{\frac{n}{2}-t}\delta_{I}(F(0))
-\varepsilon p^{\frac{n}{2}-m}|I|\delta_{0}(\beta)+p^{n-m-t}|I|.
\end{split}
\end{equation*}

$\mathrm{(ii)}$ When $p=2$, for nonempty set $I \subset V_{m}^{(2)}$ and $\alpha, \alpha' \in V_{n}^{(2)} \backslash \{0\}, i, i' \in \mathbb{F}_{2^t}^{*}$, let
\begin{equation*}
T=\sum_{u \in I}\sum_{z, w \in \mathbb{F}_{2^t}}(-1)^{Tr_{1}^{t}(iz+i'w)}\delta_{u}(F^{*}(z\alpha+w\alpha')).
\end{equation*}
Then
\begin{equation*}
\begin{split}
T&=2^{t}\delta_{I}(F^{*}(\alpha+ii'^{-1}\alpha'))-\sum_{w \in \mathbb{F}_{2^t}}\delta_{I}(F^{*}(\alpha+w\alpha'))-\delta_{I}(F^{*}(\alpha'))+\delta_{I}(F(0)).
\end{split}
\end{equation*}
\end{lemma}

\begin{proof}
The proof of Lemma 4 is given in Appendix-Section IX.
\end{proof}

\subsection{Self-orthogonal codes constructed from vectorial dual-bent functions with Condition I}
\label{sec:3.2}

In this subsection, we show that if $F$ is a vectorial dual-bent function with Condition I, then for any nonempty set $I\subset V_{m}^{(p)}$, $C_{D_{F, I}}$ defined by Eq. (1) is a at most five-weight self-orthogonal code and its weight distribution can be completely determined.

\begin{theorem}\label{Theorem 1}
Let $F: V_{n}^{(p)}\rightarrow V_{m}^{(p)}$ be a vectorial dual-bent function with Condition I, and for any nonempty set $I\subset V_{m}^{(p)}$, let $C_{D_{F, I}}$ be defined by Eq. (1). Then $C_{D_{F, I}}$ is a at most five-weight $[(p^{n-m}-\varepsilon p^{\frac{n}{2}-m})|I|+\varepsilon p^{\frac{n}{2}}\delta_{I}(F(0)), \frac{n}{t}+1]_{p^t}$ self-orthogonal linear code whose weight distribution is given in Table 1. Besides, except $p=2, t=1, m=\frac{n}{2}-1, I\subseteq V_{m}^{(2)} \backslash \{F(0)\}$ with $|I|=1$, the dual code $C_{D_{F, I}}^{\bot}$ is at least almost optimal according to Hamming bound.
\end{theorem}

\renewcommand{\thetable}{1}
\begin{table}\label{1}\scriptsize
  \centering
  \caption{The weight distribution of $C_{D_{F, I}}$ constructed in Theorem 1}
  \begin{tabular}{|c|c|}
    \hline
    Hamming weight & Multiplicity \\ \hline
    $0$ & $1$ \\
    $(p^{n-m}-\varepsilon p^{\frac{n}{2}-m})|I|+\varepsilon p^{\frac{n}{2}}\delta_{I}(F(0))$ & $p^t-1$ \\
    $(p^{n-m-t}|I|+\varepsilon p^{\frac{n}{2}-t}\delta_{I}(F(0))-\varepsilon p^{\frac{n}{2}-t})(p^t-1)$
      & $(p^{n-m}-\varepsilon p^{\frac{n}{2}-m})|I|+(\varepsilon p^{\frac{n}{2}}-1)\delta_{I}(F(0))$ \\
    $(p^{n-m-t}|I|+\varepsilon p^{\frac{n}{2}-t}\delta_{I}(F(0)))(p^t-1)+\varepsilon p^{\frac{n}{2}-t}-\varepsilon p^{\frac{n}{2}-m}|I|$
     & $(p^t-1)((p^{n-m}-\varepsilon p^{\frac{n}{2}-m})|I|+(\varepsilon p^{\frac{n}{2}}-1)\delta_{I}(F(0)))$ \\
    $(p^{n-m-t}|I|+\varepsilon p^{\frac{n}{2}-t}\delta_{I}(F(0)))(p^t-1)$
     & $(p^n-1)-(p^{n-m}-\varepsilon p^{\frac{n}{2}-m})|I|-(\varepsilon p^{\frac{n}{2}}-1)\delta_{I}(F(0))$ \\
    $(p^{n-m-t}|I|+\varepsilon p^{\frac{n}{2}-t}\delta_{I}(F(0)))(p^t-1)-\varepsilon p^{\frac{n}{2}-m}|I|$
     & $(p^t-1)(p^n-1-(p^{n-m}-\varepsilon p^{\frac{n}{2}-m})|I|-(\varepsilon p^{\frac{n}{2}}-1)\delta_{I}(F(0)))$ \\
    \hline
  \end{tabular}
\end{table}

\begin{proof}
By Lemma 2, the length of $C_{D_{F, I}}$ is $|D_{F, I}|=(p^{n-m}-\varepsilon p^{\frac{n}{2}-m})|I|+\varepsilon p^{\frac{n}{2}}\delta_{I}(F(0))$. When $\alpha=0, \beta=0$, $wt(c_{\alpha, \beta})=0$. When $\alpha=0, \beta \in \mathbb{F}_{p^t}^{*}$, $wt(c_{\alpha, \beta})=|D_{F, I}|$. When $\alpha \in V_{n}^{(p)} \backslash \{0\}, \beta \in \mathbb{F}_{p^t}$, $wt(c_{\alpha, \beta})=|D_{F, I}|-N_{I, \alpha, \beta}$, where $N_{I, \alpha, \beta}=|\{x \in V_{n}^{(p)}: F(x) \in I, \sum_{j=1}^{s}Tr_{t}^{n_{j}}(\alpha_{j}x_{j})+\beta=0\}|$.

By Lemma 4,
\begin{itemize}
  \item when $\alpha\neq 0$ with $F^{*}(\alpha) \in I$ and $\beta=0$, $wt(c_{\alpha, \beta})=(p^{n-m-t}|I|+\varepsilon p^{\frac{n}{2}-t}\delta_{I}(F(0))-\varepsilon p^{\frac{n}{2}-t})(p^t-1)$;
  \item when $\alpha \neq 0$ with $F^{*}(\alpha) \in I$ and $\beta \neq 0$, $wt(c_{\alpha, \beta})=(p^{n-m-t}|I|+\varepsilon p^{\frac{n}{2}-t}\delta_{I}(F(0)))(p^t-1)+\varepsilon p^{\frac{n}{2}-t}-\varepsilon p^{\frac{n}{2}-m}|I|$;
  \item when $\alpha\neq 0$ with $F^{*}(\alpha) \notin I$ and $\beta=0$, $wt(c_{\alpha, \beta})=(p^{n-m-t}|I|+\varepsilon p^{\frac{n}{2}-t}\delta_{I}(F(0)))(p^t-1)$;
  \item when $\alpha\neq 0$ with $F^{*}(\alpha) \notin I$ and $\beta\neq 0$, $wt(c_{\alpha, \beta})=(p^{n-m-t}|I|+\varepsilon p^{\frac{n}{2}-t}\delta_{I}(F(0)))(p^t-1)-\varepsilon p^{\frac{n}{2}-m}|I|$.
\end{itemize}
We can see that $wt(c_{\alpha, \beta})=0$ if and only if $\alpha=0, \beta=0$. Thus, the dimension of $C_{D_{F, I}}$ is $\frac{n}{t}+1$. The weight distribution of $C_{D_{F, I}}$ follows from the above arguments and Lemma 2.

When $p$ is odd, $c_{\alpha, \beta}\cdot c_{\alpha, \beta}=\sum_{x \in D_{F, I}}(\sum_{j=1}^{s}Tr_{t}^{n_{j}}(\alpha_{j}x_{j}))^{2}+2\beta \sum_{x \in D_{F, I}}\sum_{j=1}^{s}Tr_{t}^{n_{j}}(\alpha_{j}x_{j})+\beta^{2}|D_{F, I}|$. Note that $p \mid |D_{F, I}|$ since $m<\frac{n}{2}$. When $\alpha=0$, $c_{\alpha, \beta}\cdot c_{\alpha, \beta}=\beta^{2}|D_{F, I}|=0$. When $\alpha \neq 0$, by Lemma 4, we have
\begin{small}
\begin{equation*}
\begin{split}
c_{\alpha, \beta} \cdot c_{\alpha, \beta}&=\sum_{i \in \mathbb{F}_{p^t}^{*}}N_{I, \alpha, -i}i^{2}+2\beta \sum_{i \in \mathbb{F}_{p^t}^{*}}N_{I, \alpha, -i}i+\beta^{2}|D_{F, I}|\\
&=(-\varepsilon p^{\frac{n}{2}-t}\delta_{I}(F^{*}(\alpha))+\varepsilon p^{\frac{n}{2}-t}\delta_{I}(F(0))+p^{n-m-t}|I|)\sum_{i \in \mathbb{F}_{p^t}^{*}}i^{2}\\
& \ \ \ +2\beta (-\varepsilon p^{\frac{n}{2}-t}\delta_{I}(F^{*}(\alpha))+\varepsilon p^{\frac{n}{2}-t}\delta_{I}(F(0))+p^{n-m-t}|I|)\sum_{i \in \mathbb{F}_{p^t}^{*}} i.
\end{split}
\end{equation*}
\end{small}Since $\sum_{i \in \mathbb{F}_{p^t}^{*}}i=0$, $\sum_{i \in \mathbb{F}_{p^t}^{*}}i^{2}=0$ if $p^t>3$, and $n>2$ if $p^t=3$, we have $c_{\alpha, \beta} \cdot c_{\alpha, \beta}=0$. By Proposition 3, when $p$ is odd, $C_{D_{F, I}}$ is self-orthogonal.

When $p=2$, for any $\alpha, \alpha' \in V_{n}^{(2)} \backslash \{0\}, i, i' \in \mathbb{F}_{2^t}^{*}$, we have
\begin{small}
\begin{equation*}
\begin{split}
M_{I, \alpha, \alpha', i, i'}&\triangleq|\{x \in V_{n}^{(2)}: F(x) \in I, \sum_{j=1}^{s}Tr_{t}^{n_{j}}(\alpha_{j}x_{j})=i, \sum_{j=1}^{s}Tr_{t}^{n_{j}}(\alpha_{j}'x_{j})=i'\}|\\
&=2^{-m-2t}\sum_{x \in V_{n}^{(2)}}\sum_{u \in I}\sum_{y \in V_{m}^{(2)}}(-1)^{\langle y, F(x)+u\rangle_{m}}\sum_{z \in \mathbb{F}_{2^t}}(-1)^{Tr_{1}^{t}((\sum_{j=1}^{s}Tr_{t}^{n_{j}}(\alpha_{j}x_{j})+i)z)}\\
& \ \ \ \times \sum_{w \in \mathbb{F}_{2^t}}(-1)^{Tr_{1}^{t}((\sum_{j=1}^{s}Tr_{t}^{n_{j}}(\alpha_{j}'x_{j})+i')w)}\\
&=2^{-m-2t}\sum_{u \in I}\sum_{z, w \in \mathbb{F}_{2^t}}(-1)^{Tr_{1}^{t}(iz+i'w)}\sum_{y \in V_{m}^{(2)}}(-1)^{\langle y, u\rangle_{m}}\sum_{x \in V_{n}^{(2)}}(-1)^{\langle y, F(x)\rangle_{m}+\sum_{j=1}^{s}Tr_{1}^{n_{j}}((\alpha_{j}z+\alpha_{j}'w)x_{j})}\\
&=2^{-m-2t}\sum_{u \in I}\sum_{z, w \in \mathbb{F}_{2^t}}(-1)^{Tr_{1}^{t}(iz+i'w)}\sum_{y \in V_{m}^{(2)}\backslash \{0\}}(-1)^{\langle y, u\rangle_{m}}W_{F_{y}}(z\alpha+w\alpha')\\
& \ \ \ +2^{n-m-2t}|I|\sum_{z, w \in \mathbb{F}_{2^t}}(-1)^{Tr_{1}^{t}(iz+i'w)}\delta_{0}(z\alpha+w\alpha')\\
&=2^{\frac{n}{2}-m-2t}\sum_{u \in I}\sum_{z, w \in \mathbb{F}_{2^t}}(-1)^{Tr_{1}^{t}(iz+i'w)}\sum_{y \in V_{m}^{(2)} \backslash \{0\}}(-1)^{\langle y, F^{*}(z\alpha+w\alpha')+u\rangle_{m}}\\
& \ \ \ +2^{n-m-2t}|I|(\sum_{z \in \mathbb{F}_{2^t}^{*}, w \in \mathbb{F}_{2^t}}(-1)^{Tr_{1}^{t}(z(i+i'z^{-1}w))}\delta_{0}(\alpha+z^{-1}w\alpha')+\sum_{w \in \mathbb{F}_{2^t}^{*}}(-1)^{Tr_{1}^{t}(i'w)}\delta_{0}(w\alpha')+1)\\
&=2^{\frac{n}{2}-m-2t}\sum_{u \in I}\sum_{z, w \in \mathbb{F}_{2^t}}(-1)^{Tr_{1}^{t}(iz+i'w)}(2^m\delta_{u}(F^{*}(z\alpha+w\alpha'))-1)\\
& \ \ \ +2^{n-m-2t}|I|(1+\sum_{w \in \mathbb{F}_{2^t}}\delta_{0}(\alpha+w\alpha')\sum_{z \in \mathbb{F}_{2^t}^{*}}(-1)^{Tr_{1}^{t}(z(i+i'w))})\\
&=2^{\frac{n}{2}-2t}\sum_{u \in I}\sum_{z, w \in \mathbb{F}_{2^t}}(-1)^{Tr_{1}^{t}(iz+i'w)}\delta_{u}(F^{*}(z\alpha+w\alpha'))\\
& \ \ \ +2^{n-m-2t}|I|(1+2^{t}\delta_{0}(\alpha+ii'^{-1}\alpha')-\sum_{w \in \mathbb{F}_{2^t}}\delta_{0}(\alpha+w\alpha')).
\end{split}
\end{equation*}
\end{small}By Lemma 4, we have
\begin{equation}\label{2}
\begin{split}
M_{I, \alpha, \alpha', i, i'}&=2^{\frac{n}{2}-t}\delta_{I}(F^{*}(\alpha+ii'^{-1}\alpha'))+2^{n-m-t}|I|\delta_{0}(\alpha+ii'^{-1}\alpha')+A,
\end{split}
\end{equation}
where $A=2^{n-m-2t}|I|(1-\sum_{w \in \mathbb{F}_{2^t}}\delta_{0}(\alpha+w\alpha'))+2^{\frac{n}{2}-2t}(-\sum_{w \in \mathbb{F}_{2^t}}\delta_{I}(F^{*}(\alpha+w\alpha'))-\delta_{I}(F^{*}(\alpha'))+\delta_{I}(F(0)))$. Note that $A$ is an integer since $M_{I, \alpha, \alpha', i, i'}$ is an integer, $t\leq \frac{n}{2}, m<\frac{n}{2}$, and $2 \mid A$ if $t=1$ as $n>4$ when $p=2$. Since $\sum_{i \in \mathbb{F}_{2^t}^{*}}i=0$ if $t\geq 2$, and $n>4$ if $t=1$, by Lemma 4 we have $\sum_{i \in \mathbb{F}_{2^t}^{*}}N_{I, \alpha, i}i=(-2^{\frac{n}{2}-t}\delta_{I}(F^{*}(\alpha))+2^{\frac{n}{2}-t}\delta_{I}(F(0))+2^{n-m-t}|I|)\sum_{i \in \mathbb{F}_{2^t}^{*}}i=0$ for $\alpha\neq 0$. Note that $2 \mid |D_{F, I}|$. For $\alpha, \alpha' \in V_{n}^{(2)}$ and $\beta, \beta' \in \mathbb{F}_{2^t}$,
\begin{small}
\begin{equation*}
\begin{split}
&c_{\alpha, \beta} \cdot c_{\alpha', \beta'}\\
&=\sum_{x \in D_{F, I}}(\sum_{j=1}^{s}Tr_{t}^{n_{j}}(\alpha_{j}x_{j}))(\sum_{j=1}^{s}Tr_{t}^{n_{j}}(\alpha_{j}'x_{j}))+\beta\sum_{x \in D_{F, I}}\sum_{j=1}^{s}Tr_{t}^{n_{j}}(\alpha_{j}'x_{j})+\beta'\sum_{x \in D_{F, I}}\sum_{j=1}^{s}Tr_{t}^{n_{j}}(\alpha_{j}x_{j})+\beta\beta'|D_{F, I}|.
\end{split}
\end{equation*}
\end{small}When $\alpha=\alpha'=0$, $c_{\alpha, \beta} \cdot c_{\alpha', \beta'}=\beta\beta'|D_{F, I}|=0$. When $\alpha=0, \alpha'\neq 0$, or $\alpha\neq 0, \alpha'=0$, w.l.o.g., $\alpha=0, \alpha' \neq 0$, $c_{\alpha, \beta}\cdot c_{\alpha', \beta'}=\beta \sum_{i \in \mathbb{F}_{2^t}^{*}}N_{I, \alpha', i}i+\beta\beta'|D_{F, I}|=0$. When $\alpha, \alpha' \neq 0$, by (2) we have
\begin{small}
\begin{equation*}
\begin{split}
c_{\alpha, \beta}\cdot c_{\alpha', \beta'}&=\sum_{k \in \mathbb{F}_{2^t}^{*}}k\sum_{i \in \mathbb{F}_{2^t}^{*}}M_{I, \alpha, \alpha', i, i^{-1}k}+\beta \sum_{i \in \mathbb{F}_{2^t}^{*}}N_{I, \alpha', i}i+\beta' \sum_{i \in \mathbb{F}_{2^t}^{*}}N_{I, \alpha, i}i+\beta\beta'|D_{F, I}|\\
&=A\sum_{i, k \in \mathbb{F}_{2^t}^{*}}k+\sum_{i, k \in \mathbb{F}_{2^t}^{*}}(2^{\frac{n}{2}-t}\delta_{I}(F^{*}(\alpha+i^{2}k^{-1}\alpha'))+2^{n-m-t}|I|\delta_{0}(\alpha+i^{2}k^{-1}\alpha'))k\\
&=2^{\frac{n}{2}-t}\sum_{i, k \in \mathbb{F}_{2^t}^{*}}\delta_{I}(F^{*}(\alpha+i^{2}k^{-1}\alpha'))k\\
&=2^{\frac{n}{2}-t}\sum_{i, k \in \mathbb{F}_{2^t}^{*}}\delta_{I}(F^{*}(\alpha+k\alpha'))i^{2}k^{-1}\\
&=2^{\frac{n}{2}-t}\sum_{k \in \mathbb{F}_{2^t}^{*}}\delta_{I}(F^{*}(\alpha+k\alpha'))k^{-1}\sum_{i \in \mathbb{F}_{2^t}^{*}}i^{2}=0,
\end{split}
\end{equation*}
\end{small}where in the third equation we use $\sum_{k \in \mathbb{F}_{2^t}^{*}}k=0$ if $t\geq 2$, and $2 \mid A$ if $t=1$, in the last equation we use $\sum_{i \in \mathbb{F}_{2^t}^{*}}i^{2}=0$ if $t\geq 2$, and $n>4$ if $t=1$. Therefore, when $p=2$, $C_{D_{F, I}}$ is self-orthogonal.

Obviously, the length of $C_{D_{F, I}}^{\bot}$ is $(p^{n-m}-\varepsilon p^{\frac{n}{2}-m})|I|+\varepsilon p^{\frac{n}{2}}\delta_{I}(F(0))$, and the dimension of $C_{D_{F, I}}^{\bot}$ is $(p^{n-m}-\varepsilon p^{\frac{n}{2}-m})|I|+\varepsilon p^{\frac{n}{2}}\delta_{I}(F(0))-\frac{n}{t}-1$. We show that the minimum distance $d(C_{D_{F, I}}^{\bot})\geq 3$. If $d(C_{D_{F, I}}^{\bot})=1$, then there is $x=(x_{1}, \dots, x_{s}) \in D_{F, I}$ such that $\sum_{j=1}^{s}Tr_{t}^{n_{j}}(\alpha_{j}x_{j})$ $+\beta=0$ for all $\alpha=(\alpha_{1}, \dots, \alpha_{s}) \in V_{n}^{(p)}, \beta \in \mathbb{F}_{p^t}$, which is obviously impossible. If $d(C_{D_{F, I}}^{\bot})=2$, then there are $z, z' \in \mathbb{F}_{p^t}^{*}$ and distinct $x=(x_{1}, \dots, x_{s}), x'=(x'_{1}, \dots, x'_{s}) \in D_{F, I}$ such that $z(\sum_{j=1}^{s}Tr_{t}^{n_{j}}(\alpha_{j}x_{j})+\beta)+z'(\sum_{j=1}^{s}Tr_{t}^{n_{j}}(\alpha_{j}x'_{j})+\beta)=0$ for all $\alpha \in V_{n}^{(p)}, \beta \in \mathbb{F}_{p^t}$. Then $\sum_{j=1}^{s}Tr_{t}^{n_{j}}(\alpha_{j}(zx_{j}+z'x'_{j}))+(z+z')\beta=0$ for all $\alpha=(\alpha_{1}, \dots, \alpha_{s}) \in V_{n}^{(p)}, \beta \in \mathbb{F}_{p^t}$. Let $\alpha=0, \beta \in \mathbb{F}_{p^t}^{*}$, we obtain $z+z'=0$. Let $\beta=0$ and $\alpha_{j}=0, j\neq i$, for any fixed $1\leq i\leq s$, we have $Tr_{t}^{n_{i}}(\alpha_{i}z(x_{i}-x'_{i}))=0$ for all $\alpha_{i} \in \mathbb{F}_{p^{n_{i}}}$, which implies that $x_{i}=x'_{i}$ for any $1\leq i\leq s$, and then $x=x'$, which contradicts $x\neq x'$. Thus, $d(C_{D_{F, I}}^{\bot})\geq 3$. By Proposition 2, except $p=2, t=1, m=\frac{n}{2}-1, I\subseteq V_{m}^{(2)} \backslash \{F(0)\}$ with $|I|=1$, $C_{D_{F, I}}^{\bot}$ is at least almost optimal according to Hamming bound.
\end{proof}

In the following, by the results in \cite{CMP2018Ve,WFW2023Be} and Lemma 3, we list some explicit classes of vectorial dual-bent functions with Condition I.
\begin{itemize}
  \item Let $m, n', t$ be positive integers with $m<n', t \mid n'$, and $m\geq 2$ if $p=2$. Let $\alpha \in \mathbb{F}_{p^{n'}}^{*}$, $B: \mathbb{F}_{p^{n'}}\rightarrow V_{m}^{(p)}$ be a balanced function. Define $F: \mathbb{F}_{p^{n'}} \times \mathbb{F}_{p^{n'}}\rightarrow V_{m}^{(p)}$ as
  \begin{equation} \label{3}
  F(x_{1}, x_{2})=B(\alpha x_{1}x_{2}^{-1}).
  \end{equation}
  Then $F$ is a vectorial dual-bent function satisfying Condition I with $\varepsilon=1$.
  \item Let $m, n', r, u, t$ be positive integers with $r \mid n', m\leq r, m\neq n', t\mid n'$, $gcd(u, p^{n'}-1)=1, u\equiv 1 \mod (p^t-1)$ and $u\equiv p^{u_{0}} \mod (p^r-1)$ for some nonnegative integer $u_{0}$, and $m\geq 2$ if $p=2$. Let $\alpha \in \mathbb{F}_{p^{n'}}^{*}$, $B: \mathbb{F}_{p^r}\rightarrow V_{m}^{(p)}$ be a balanced function. Define $F: \mathbb{F}_{p^{n'}} \times \mathbb{F}_{p^{n'}}\rightarrow V_{m}^{(p)}$ as
  \begin{equation}\label{4}
  F(x_{1}, x_{2})=B(Tr_{r}^{n'}(\alpha x_{1}x_{2}^{-u})).
  \end{equation}
  Then $F$ is a vectorial dual-bent function satisfying Condition I with $\varepsilon=1$.
  \item Let $m, n', n'', r, u, t$ be positive integers with $m\leq n', m\leq r, r \mid n'', t \mid n', t \mid n''$, $gcd(u, p^{n''}-1)=1, u\equiv 1 \mod (p^t-1)$ and $u\equiv p^{u_{0}} \mod (p^r-1)$ for some nonnegative integer $u_{0}$, and $m\geq 2$ if $p=2$. Let $\alpha \in \mathbb{F}_{p^{n'}}^{*}, \beta \in \mathbb{F}_{p^{n''}}^{*}$, $B_{1}: \mathbb{F}_{p^{n'}}\rightarrow V_{m}^{(p)}$, $B_{2}: \mathbb{F}_{p^r}\rightarrow V_{m}^{(p)}$ be balanced functions. Define $F: \mathbb{F}_{p^{n'}} \times \mathbb{F}_{p^{n'}} \times \mathbb{F}_{p^{n''}} \times \mathbb{F}_{p^{n''}} \rightarrow V_{m}^{(p)}$ as
  \begin{equation}\label{5}
  F(x_{1}, x_{2}, x_{3}, x_{4})=B_{1}(\alpha x_{1}x_{2}^{-1})+B_{2}(Tr_{r}^{n''}(\beta x_{3}x_{4}^{-u})).
  \end{equation}
  Then $F$ is a vectorial dual-bent function satisfying Condition I with $\varepsilon=1$.
\end{itemize}

By Theorem 1 and vectorial dual-bent functions defined by Eq. (3), in Table 2, we list some linear codes which are optimal or have the best parameters up to now according to the Code Tables at http://www.codetables.de/. Note that some parameters can also be attained by vectorial dual-bent functions defined by Eq. (4), (5).

\renewcommand{\thetable}{2}
\begin{table}\label{2}\scriptsize
\centering
  \caption{Some linear codes produced by Theorem 1 which are optimal or have the best parameters up to now}
  \begin{threeparttable}
  \begin{tabular}{|c|c|c|c|}\hline
    Parameter & Code & Condition & Optimality  \\ \hline
    $[14, 7, 4]_{2}$ & $C_{D_{F, I}}$ & $F$ is given by (3) with $p=2, t=1, m=2, n'=3$, $I\subseteq V_{2}^{(2)} \backslash \{B(0)\}$ with $|I|=1$  & optimal \\ \hline
    $[28, 7, 12]_{2}$ & $C_{D_{F, I}}$ & $F$ is given by (3) with $p=2, t=1, m=2, n'=3$, $I\subseteq V_{2}^{(2)} \backslash \{B(0)\}$ with $|I|=2$ & optimal \\ \hline
    $[28, 21, 4]_{2}$ & $C_{D_{F, I}}^{\bot}$ & $F$ is given by (3) with $p=2, t=1, m=2, n'=3$, $I\subseteq V_{2}^{(2)} \backslash \{B(0)\}$ with $|I|=2$ & optimal \\ \hline
    $[30, 21, 4]_{2}$ & $C_{D_{F, I}}^{\bot}$ & $F$ is given by (3) with $p=2, t=1, m=3, n'=4$, $I\subseteq V_{3}^{(2)}\backslash \{B(0)\}$ with $|I|=1$ & optimal \\ \hline
    $[42, 35, 4]_{2}$ & $C_{D_{F, I}}^{\bot}$ & $F$ is given by (3) with $p=2, t=1, m=2, n'=3$, $I=V_{2}^{(2)} \backslash \{B(0)\}$ & optimal \\ \hline
    $[60, 51, 4]_{2}$ & $C_{D_{F, I}}^{\bot}$ & $F$ is given by (3) with $p=2, t=1, m=2, n'=4$, $I\subseteq V_{2}^{(2)} \backslash \{B(0)\}$ with $|I|=1$ & optimal \\ \hline
    $[62, 51, 4]_{2}$ & $C_{D_{F, I}}^{\bot}$ & $F$ is given by (3) with $p=2, t=1, m=4, n'=5$, $I\subseteq V_{4}^{(2)} \backslash \{B(0)\}$ with $|I|=1$ & optimal \\ \hline
    $[90, 9, 40]_{2}$ & $C_{D_{F, I}}$ & $F$ is given by (3) with $p=2, t=1, m=3, n'=4$, $I\subseteq V_{3}^{(2)}\backslash \{B(0)\}$ with $|I|=3$ & best parameter up to now \\ \hline
    $[90, 81, 4]_{2}$ & $C_{D_{F, I}}^{\bot}$ & $F$ is given by (3) with $p=2, t=1, m=3, n'=4$, $I\subseteq V_{3}^{(2)}\backslash \{B(0)\}$ with $|I|=3$ & optimal \\ \hline
    $[120, 9, 56]_{2}$ & $C_{D_{F, I}}$ & $F$ is given by (3) with $p=2, t=1, m=2, n'=4$, $I\subseteq V_{2}^{(2)} \backslash \{B(0)\}$ with $|I|=2$ & optimal \\ \hline
    $[120, 111, 4]_{2}$ & $C_{D_{F, I}}^{\bot}$ & $F$ is given by (3) with $p=2, t=1, m=2, n'=4$, $I\subseteq V_{2}^{(2)}\backslash \{B(0)\}$ with $|I|=2$ & optimal \\ \hline
    $[124, 113, 4]_{2}$ & $C_{D_{F, I}}^{\bot}$ & $F$ is given by (3) with $p=2, t=1, m=3, n'=5$, $I\subseteq V_{3}^{(2)} \backslash \{B(0)\}$ with $|I|=1$ & optimal \\ \hline
    $[126, 113, 4]_{2}$ & $C_{D_{F, I}}^{\bot}$ & $F$ is given by (3) with $p=2, t=1, m=5, n'=6$, $I\subseteq V_{5}^{(2)} \backslash \{B(0)\}$ with $|I|=1$ & optimal \\ \hline
    $[150, 141, 4]_{2}$ & $C_{D_{F, I}}^{\bot}$ & $F$ is given by (3) with $p=2, t=1, m=3, n'=4$, $I\subseteq V_{3}^{(2)}\backslash \{B(0)\}$ with $|I|=5$ & optimal \\ \hline
    $[180, 171, 4]_{2}$ & $C_{D_{F, I}}^{\bot}$ & $F$ is given by (3) with $p=2, t=1, m=2, n'=4$, $I=V_{2}^{(2)}\backslash \{B(0)\}$ & optimal \\ \hline
    $[186, 175, 4]_{2}$ & $C_{D_{F, I}}^{\bot}$ & $F$ is given by (3) with $p=2, t=1, m=4, n'=5$, $I\subseteq V_{4}^{(2)}\backslash \{B(0)\}$ with $|I|=3$ & optimal \\ \hline
    $[210, 201, 4]_{2}$ & $C_{D_{F, I}}^{\bot}$ & $F$ is given by (3) with $p=2, t=1, m=3, n'=4$, $I=V_{3}^{(2)}\backslash \{B(0)\}$ & optimal \\ \hline
    $[248, 237, 4]_{2}$ & $C_{D_{F, I}}^{\bot}$ & $F$ is given by (3) with $p=2, t=1, m=2, n'=5$, $I\subseteq V_{2}^{(2)}\backslash \{B(0)\}$ with $|I|=1$ & optimal \\ \hline
    $[252, 239, 4]_{2}$ & $C_{D_{F, I}}^{\bot}$ & $F$ is given by (3) with $p=2, t=1, m=4, n'=6$, $I\subseteq V_{4}^{(2)}\backslash \{B(0)\}$ with $|I|=1$ & optimal \\ \hline
    $[254, 239, 4]_{2}$ & $C_{D_{F, I}}^{\bot}$ & $F$ is given by (3) with $p=2, t=1, m=6, n'=7$, $I\subseteq V_{6}^{(2)}\backslash \{B(0)\}$ with $|I|=1$ & optimal \\ \hline
    $[156, 149, 3]_{3}$ & $C_{D_{F, I}}^{\bot}$ & $F$ is given by (3) with $p=3, t=1, m=2, n'=3$, $I\subseteq V_{2}^{(3)}\backslash \{B(0)\}$ with $|I|=2$ & optimal \\ \hline
    $[234, 227, 3]_{3}$ & $C_{D_{F, I}}^{\bot}$ & $F$ is given by (3) with $p=3, t=1, m=2, n'=3$, $I\subseteq V_{2}^{(3)}\backslash \{B(0)\}$ with $|I|=3$ & optimal \\ \hline
    $[60, 55, 3]_{4}$ & $C_{D_{F, I}}^{\bot}$ &  $F$ is given by (3) with $p=2, t=2, m=3, n'=4$, $I\subseteq V_{3}^{(2)}\backslash \{B(0)\}$ with $|I|=2$ & optimal \\ \hline
    $[90, 85, 3]_{4}$ & $C_{D_{F, I}}^{\bot}$ & $F$ is given by (3) with $p=2, t=2, m=3, n'=4$, $I\subseteq V_{3}^{(2)}\backslash \{B(0)\}$ with $|I|=3$ & optimal \\ \hline
    $[120, 115, 3]_{4}$ & $C_{D_{F, I}}^{\bot}$ & $F$ is given by (3) with $p=2, t=2, m=3, n'=4$, $I\subseteq V_{3}^{(2)}\backslash \{B(0)\}$ with $|I|=4$ & optimal \\ \hline
    $[150, 145, 3]_{4}$ & $C_{D_{F, I}}^{\bot}$ & $F$ is given by (3) with $p=2, t=2, m=3, n'=4$, $I\subseteq V_{3}^{(2)}\backslash \{B(0)\}$ with $|I|=5$ & optimal \\ \hline
    $[180, 5, 132]_{4}$ & $C_{D_{F, I}}$ & $F$ is given by (3) with $p=2, t=2, m=3, n'=4$, $I\subseteq V_{3}^{(2)}\backslash \{B(0)\}$ with $|I|=6$ & best parameter up to now \\ \hline
    $[180, 175, 3]_{4}$ & $C_{D_{F, I}}^{\bot}$ & $F$ is given by (3) with $p=2, t=2, m=3, n'=4$, $I\subseteq V_{3}^{(2)}\backslash \{B(0)\}$ with $|I|=6$ & optimal \\ \hline
    $[210, 205, 3]_{4}$ & $C_{D_{F, I}}^{\bot}$ & $F$ is given by (3) with $p=2, t=2, m=3, n'=4$, $I=V_{3}^{(2)}\backslash \{B(0)\}$ & optimal \\ \hline
    $[14, 11, 3]_{8}$ & $C_{D_{F, I}}^{\bot}$ & $F$ is given by (3) with $p=2, t=3, m=2, n'=3$, $I\subseteq V_{2}^{(2)}\backslash \{B(0)\}$ with $|I|=1$ & optimal \\ \hline
    $[28, 25, 3]_{8}$ & $C_{D_{F, I}}^{\bot}$ & $F$ is given by (3) with $p=2, t=3, m=2, n'=3$, $I\subseteq V_{2}^{(2)}\backslash \{B(0)\}$ with $|I|=2$ & optimal \\ \hline
    $[42, 39, 3]_{8}$ & $C_{D_{F, I}}^{\bot}$ & $F$ is given by (3) with $p=2, t=3, m=2, n'=3$, $I=V_{2}^{(2)}\backslash \{B(0)\}$ & optimal \\ \hline
    $[24, 21, 3]_{9}$ & $C_{D_{F, I}}^{\bot}$ & $F$ is given by (3) with $p=3, t=2, m=1, n'=2$, $I\subseteq \mathbb{F}_{3}\backslash \{B(0)\}$ with $|I|=1$ & optimal \\ \hline
    $[48, 45, 3]_{9}$ & $C_{D_{F, I}}^{\bot}$ & $F$ is given by (3) with $p=3, t=2, m=1, n'=2$, $I=\mathbb{F}_{3}\backslash \{B(0)\}$ & optimal \\ \hline
  \end{tabular}
 \end{threeparttable}
\end{table}
\section{Self-orthogonal codes from vectorial dual-bent functions with Condition II}
\label{sec:4}

In this section, we construct self-orthogonal codes from vectorial dual-bent functions with the following condition:

Condition II: Let $n, n_{j}, 1\leq j \leq s, m, t$ be positive integers for which $n=\sum_{j=1}^{s}n_{j}, 2 \mid n, t \mid n_{j}, 1\leq j \leq s, t \mid m, m<\frac{n}{2}$, and when $p=2$, $m\geq 2$ and $m+t<\frac{n}{2}$, and let $V_{n}^{(p)}=\mathbb{F}_{p^{n_{1}}} \times \mathbb{F}_{p^{n_{2}}} \times \dots \times \mathbb{F}_{p^{n_{s}}}$. Let $F: V_{n}^{(p)}\rightarrow \mathbb{F}_{p^m}$ be a vectorial dual-bent function satisfying
\begin{itemize}
  \item There is a vectorial dual $F^{*}$ such that $(F_{c})^{*}=(F^{*})_{c^{1-d}}, c \in \mathbb{F}_{p^m}^{*}$, where $gcd(d-1, p^m-1)=1$;
  \item $F(ax)=a^{l}F(x), a \in \mathbb{F}_{p^t}^{*}, x \in V_{n}^{(p)}$, and $F(0)=0$, where $(l-1)(d-1)\equiv 1 \mod (p^m-1)$;
  \item All component functions $F_{c}, c \in \mathbb{F}_{p^m}^{*}$, are weakly regular with $\varepsilon_{F_{c}}=\varepsilon, c \in \mathbb{F}_{p^m}^{*}$, where $\varepsilon \in \{\pm1\}$ is a constant.
\end{itemize}

\subsection{Some lemmas}\label{4.1}
In this section, we give some useful lemmas.

\begin{lemma} \label{Lemma 5}
Let $F$ be a vectorial dual-bent function with Condition II. Then the vectorial dual $F^{*}$ with $(F_{c})^{*}=(F^{*})_{c^{1-d}}, c \in \mathbb{F}_{p^m}^{*}$, is a vectorial dual-bent function for which $((F^{*})_{c})^{*}=F_{c^{1-l}}, c \in \mathbb{F}_{p^m}^{*}$, $F^{*}(ax)=a^{d}F^{*}(x), a \in \mathbb{F}_{p^t}^{*}$, $F^{*}(0)=0$, and all component functions $(F^{*})_{c}, c \in \mathbb{F}_{p^m}^{*}$, are weakly regular with $\varepsilon_{(F^{*})_{c}}=\varepsilon$.
\end{lemma}

\begin{proof}
Since $F_{c}, c \in \mathbb{F}_{p^m}^{*}$, are all weakly regular bent with $\varepsilon_{F_{c}}=\varepsilon \in \{\pm 1\}$, $(F^{*})_{c}=(F_{c^{1-l}})^{*}$ is weakly regular bent with $\varepsilon_{(F^{*})_{c}}=\varepsilon$ for any $c \in \mathbb{F}_{p^m}^{*}$. For any $c \in \mathbb{F}_{p^m}^{*}$, $((F^{*})_{c})^{*}(x)=((F_{c^{1-l}})^{*})^{*}(x)=F_{c^{1-l}}(-x)=F_{c^{1-l}}(x)$, and thus $F^{*}$ is vectorial dual-bent. By Corollary 2 and Proposition 5 of \cite{CMP2021Ve}, $F^{*}(0)=0$. For any $c \in \mathbb{F}_{p^m}^{*}, a \in \mathbb{F}_{p^t}^{*}, x \in V_{n}^{(p)}$, we have\\
\begin{small}
\begin{equation*}
\begin{split}
p^{n}\zeta_{p}^{(F^{*})_{c}(ax)}&=\sum_{y \in V_{n}^{(p)}}W_{(F^{*})_{c}}(y)\zeta_{p}^{\sum_{j=1}^{s}Tr_{1}^{n_{j}}(ax_{j}y_{j})}=\varepsilon p^{\frac{n}{2}}\sum_{y \in V_{n}^{(p)}}\zeta_{p}^{F_{c^{1-l}}(y)+\sum_{j=1}^{s}Tr_{1}^{n_{j}}(ax_{j}y_{j})}\\
&=\varepsilon p^{\frac{n}{2}}\sum_{y \in V_{n}^{(p)}}\zeta_{p}^{F_{c^{1-l}}(a^{d-1}y)+\sum_{j=1}^{s}Tr_{1}^{n_{j}}(a^{d}x_{j}y_{j})}=\varepsilon p^{\frac{n}{2}}\sum_{y \in V_{n}^{(p)}}\zeta_{p}^{F_{c^{1-l}a^{d}}(y)+\sum_{j=1}^{s}Tr_{1}^{n_{j}}(a^{d}x_{j}y_{j})}\\
&=\sum_{y \in V_{n}^{(p)}}W_{(F^{*})_{ca^{d(1-d)}}}(y)\zeta_{p}^{\sum_{j=1}^{s}Tr_{1}^{n_{j}}(a^{d}x_{j}y_{j})}=p^{n}\zeta_{p}^{(F^{*})_{ca^{d(1-d)}}(a^{d}x)},
\end{split}
\end{equation*}
\end{small}where in the fourth equation we use $F(ax)=a^{l}F(x)$ and $(l-1)(d-1)\equiv 1 \mod (p^m-1), t \mid m$. Thus we have $Tr_{1}^{m}(cF^{*}(ax))=Tr_{1}^{m}(ca^{d(1-d)}F^{*}(a^{d}x)), c \in \mathbb{F}_{p^m}^{*}$, and then $F^{*}(ax)=a^{d(1-d)}F^{*}(a^{d}x)$ for any $a \in \mathbb{F}_{p^t}^{*}, x \in V_{n}^{(p)}$, which implies that $F^{*}(a^{1-d}x)=a^{d(1-d)}F^{*}(x)$ for any $a \in \mathbb{F}_{p^t}^{*}, x \in V_{n}^{(p)}$. Since $(d-1)(l-1)\equiv 1 \mod (p^m-1), t\mid m$, we have $F^{*}(ax)=a^{d}F^{*}(x)$ for any $a \in \mathbb{F}_{p^t}^{*}, x \in V_{n}^{(p)}$.
\end{proof}

\begin{lemma} \label{Lemma 6}
Let $F$ be a vectorial dual-bent function with Condition II. Then the value distributions of $F$ and $F^{*}$ are given by
\begin{equation*}
|D_{F, i}|=|D_{F^{*}, i}|=p^{n-m}+\varepsilon p^{\frac{n}{2}-m}(p^m\delta_{0}(i)-1), i \in \mathbb{F}_{p^m}.
\end{equation*}
\end{lemma}

\begin{proof}
By Lemma 5 and Corollary 1 of \cite{WF2023Ne} (Note that although Corollary 1 of \cite{WF2023Ne} only considers the case of $p$ being odd, the result also holds for $p=2, m\geq 2$), the value distributions of $F$ and $F^{*}$ hold.
\end{proof}

\begin{lemma} \label{Lemma 7}
Let $F$ be a vectorial dual-bent function with Condition II.

$\mathrm{(i)}$ For any $a \in \mathbb{F}_{p^m}$ and $\alpha \in V_{n}^{(p)} \backslash \{0\}, \beta \in \mathbb{F}_{p^t}$, define
\begin{equation*}
N_{a, \alpha, \beta}=|\{x \in V_{n}^{(p)}: F(x)=a, \sum_{j=1}^{s}Tr_{t}^{n_{j}}(\alpha_{j}x_{j})+\beta=0\}|.
\end{equation*}
Then
\begin{equation*}
\begin{split}
N_{a, \alpha, \beta}&=\varepsilon p^{\frac{n}{2}-m}|\{y \in \mathbb{F}_{p^m}^{*}: Tr_{t}^{m}(yF^{*}(\alpha)-ay^{1-l})=\beta\}|+\varepsilon p^{\frac{n}{2}-t}(\delta_{0}(a)-1)+p^{n-m-t}.
\end{split}
\end{equation*}

$\mathrm{(ii)}$ When $p=2$, for $a \in \mathbb{F}_{2^m}$ and $\alpha, \alpha' \in V_{n}^{(2)} \backslash \{0\}, i, i' \in \mathbb{F}_{2^t}^{*}$, define
\begin{equation*}
T=\sum_{z, w \in \mathbb{F}_{2^t}}(-1)^{Tr_{1}^{t}(iz+i'w)}\sum_{y \in \mathbb{F}_{2^m}^{*}}(-1)^{Tr_{1}^{m}(ay)+Tr_{1}^{m}(y^{1-d}F^{*}(z\alpha+w\alpha'))}.
\end{equation*}
Then
\begin{equation*}
\begin{split}
T&=2^{t}\sum_{w \in \mathbb{F}_{2^t}}|\{y \in \mathbb{F}_{2^m}^{*}: Tr_{t}^{m}(F^{*}(\alpha+w\alpha')y+ay^{1-l})=i+wi'\}|\\
& \ \ \ +2^{t}|\{y \in \mathbb{F}_{2^m}^{*}: Tr_{t}^{m}(F^{*}(\alpha')y+ay^{1-l})=i'\}|
-(2^t+1)(2^m-1)+2^{m}\delta_{0}(a)-1.
\end{split}
\end{equation*}
\end{lemma}

\begin{proof}
The proof of Lemma 7 is given in Appendix-Section IX.
\end{proof}

\begin{lemma} \label{Lemma 8}
Let $X=\sum_{z \in \mathbb{F}_{p^m}^{*}}\zeta_{p}^{Tr_{1}^{m}(-z\beta)}\delta_{w^{i}H_{b}}(z^{2})$, where $\beta \in \mathbb{F}_{p^m}$, $H_{b}=\{x^{b}: x \in \mathbb{F}_{p^m}^{*}\}$, $b$ is a positive integer with $b \mid (p^{m}-1)$, $w$ is a primitive element of $\mathbb{F}_{p^m}$, $i$ is an integer.

$\mathrm{(i)}$ When $p$ is an odd prime, $i$ is even, $b$ is odd, $X=\sum_{z \in H_{b}}\zeta_{p}^{Tr_{1}^{m}(-zw^{\frac{i}{2}}\beta)}$.

$\mathrm{(ii)}$ When $p$ is an odd prime, $i$ is odd, $b$ is odd, $X=\sum_{z \in H_{b}}\zeta_{p}^{Tr_{1}^{m}(-zw^{\frac{i+b}{2}}\beta)}$.

$\mathrm{(iii)}$ When $p$ is an odd prime, $i$ is even, $b$ is even, $X=\sum_{z \in H_{b}}(\zeta_{p}^{Tr_{1}^{m}(-zw^{\frac{i}{2}}\beta)}+\zeta_{p}^{Tr_{1}^{m}(-zw^{\frac{i+b}{2}}\beta)})$.

$\mathrm{(iv)}$ When $p$ is an odd prime, $i$ is odd, $b$ is even, $X=0$.

$\mathrm{(v)}$ When $p=2$, $X=\sum_{z \in H_{b}}(-1)^{Tr_{1}^{m}(zw^{i}\beta^{2})}$.
\end{lemma}

\begin{proof}
The proof of Lemma 8 is given in Appendix-Section IX.
\end{proof}

\begin{lemma} \label{Lemma 9}
Let $F: V_{n}^{(p)}\rightarrow \mathbb{F}_{p^m}$ be a vectorial dual-bent function satisfying Condition II for which $l=2, t=m=2jj'$, and there is an integer $b\geq 2$ with $b \mid (p^j+1)$, where $j$ is the smallest such positive integer. Let $w$ be a primitive element of $\mathbb{F}_{p^m}$, $H_{b}=\{x^{b}: x \in \mathbb{F}_{p^m}^{*}\}$, and when $p$ is odd, $\mathcal{S}=\{x^{2}: x \in \mathbb{F}_{p^m}^{*}\}, \mathcal{N}=\mathbb{F}_{p^m}^{*} \backslash \mathcal{S}$. For $\alpha \in V_{n}^{(p)} \backslash \{0\}, \beta \in \mathbb{F}_{p^m}, \gamma \in \mathbb{F}_{p^m}^{*}$, define
\begin{equation*}
T=\sum_{a \in \gamma H_{b}}|\{y \in \mathbb{F}_{p^m}^{*}: F^{*}(\alpha)y^{2}-\beta y-a=0\}|.
\end{equation*}

$\mathrm{(i)}$ When $F^{*}(\alpha)=0, \beta=0$, or $b$ is even, $F^{*}(\alpha) \neq 0, \beta=0, \gamma F^{*}(\alpha)^{-1} \in \mathcal{N}$, then $T=0$.

$\mathrm{(ii)}$ When $F^{*}(\alpha)=0, \beta\neq 0$, or $b$ is odd, $F^{*}(\alpha)\neq 0, \beta=0$, then $T=\frac{p^m-1}{b}$.

$\mathrm{(iii)}$ When $b$ is even, $F^{*}(\alpha) \neq 0, \beta=0, \gamma F^{*}(\alpha)^{-1} \in \mathcal{S}$, then $T=\frac{2(p^m-1)}{b}$.

$\mathrm{(iv)}$ When $b$ is odd, $F^{*}(\alpha)\neq 0, \beta \neq 0$,
\begin{equation*}
T=\left\{
\begin{split}
(-1)^{j'+1}p^{\frac{m}{2}}+\frac{(-1)^{j'}p^{\frac{m}{2}}+p^{m}-2}{b}, & \ \text{ if } \gamma^{-1}F^{*}(\alpha)^{-1}\in \mathcal{S}, \beta \sqrt{\gamma^{-1}F^{*}(\alpha)^{-1}} \in H_{b}, \\
 & \  \text{ or }\gamma^{-1}F^{*}(\alpha)^{-1}\in \mathcal{N}, \beta \sqrt{\gamma^{-1}F^{*}(\alpha)^{-1}w^{b}} \in H_{b},\\
\frac{(-1)^{j'}p^{\frac{m}{2}}+p^{m}-2}{b}, & \ \text{ if } \gamma^{-1}F^{*}(\alpha)^{-1}\in \mathcal{S}, \beta \sqrt{\gamma^{-1}F^{*}(\alpha)^{-1}} \notin H_{b}, \\
 & \  \text{ or }\gamma^{-1}F^{*}(\alpha)^{-1}\in \mathcal{N}, \beta \sqrt{\gamma^{-1}F^{*}(\alpha)^{-1}w^{b}} \notin H_{b},
\end{split}
\right.
\end{equation*}
for odd $p$, and
\begin{equation*}
T=\left\{
\begin{split}
(-1)^{j'+1}2^{\frac{m}{2}}+\frac{(-1)^{j'}2^{\frac{m}{2}}+2^{m}-2}{b}, & \ \text{ if } \gamma^{-1}F^{*}(\alpha)^{-1} \beta^{2} \in H_{b},\\
\frac{(-1)^{j'}2^{\frac{m}{2}}+2^{m}-2}{b}, & \ \text{ if } \gamma^{-1}F^{*}(\alpha)^{-1} \beta^{2} \notin H_{b},
\end{split}
\right.
\end{equation*}
for $p=2$.

$\mathrm{(v)}$ When $b$ is even, $F^{*}(\alpha)\neq 0, \beta \neq 0$,
\begin{equation*}
T=\left\{
\begin{split}
(-1)^{j'+1}p^{\frac{m}{2}}+\frac{(p^m-1)+2((-1)^{j'}p^{\frac{m}{2}}-1)}{b}, & \ \text{ if } \gamma^{-1}F^{*}(\alpha)^{-1} \in \mathcal{S}, \beta\sqrt{\gamma^{-1}F^{*}(\alpha)^{-1}} \in H_{\frac{b}{2}},\\
\frac{(p^m-1)+2((-1)^{j'}p^{\frac{m}{2}}-1)}{b}, & \ \text{ if } \gamma^{-1}F^{*}(\alpha)^{-1} \in \mathcal{S}, \beta\sqrt{\gamma^{-1}F^{*}(\alpha)^{-1}} \notin H_{\frac{b}{2}},\\
\frac{p^m-1}{b}, & \ \text{ if } \gamma^{-1}F^{*}(\alpha)^{-1} \in \mathcal{N}.
\end{split}
\right.
\end{equation*}
\end{lemma}

\begin{proof}
The proof of Lemma 9 is given in Appendix-Section IX.
\end{proof}

\subsection{Self-orthogonal codes constructed from vectorial dual-bent functions with Condition II} \label{4.2}
In this subsection, we show that if $F$ is a vectorial dual-bent function with Condition II, then for any nonempty set $I\subset \mathbb{F}_{p^m}$, $C_{D_{F, I}}$ defined by Eq. (1) is self-orthogonal. Furthermore, for some sets $I$, we completely determine the weight distribution of $C_{D_{F, I}}$.

\begin{theorem}\label{Theorem 2}
Let $F: V_{n}^{(p)}\rightarrow \mathbb{F}_{p^m}$ be a vectorial dual-bent function with Condition II, and for any nonempty set $I\subset \mathbb{F}_{p^m}$, let $C_{D_{F, I}}$ be defined by Eq. (1). Then $C_{D_{F, I}}$ is a $[(p^{n-m}-\varepsilon p^{\frac{n}{2}-m})|I|+\varepsilon p^{\frac{n}{2}}\delta_{I}(0), \frac{n}{t}+1]_{p^t}$ self-orthogonal linear code. Besides, its dual code $C_{D_{F, I}}^{\bot}$ is at least almost optimal according to Hamming bound.
\end{theorem}

\begin{proof}
First, we prove the case of $I=\{a\}$, where $a \in \mathbb{F}_{p^m}$. By Lemma 6, the length of $C_{D_{F, a}}$ is $|D_{F, a}|=(p^{n-m}-\varepsilon p^{\frac{n}{2}-m})+\varepsilon p^{\frac{n}{2}}\delta_{a}(0)$. When $\alpha=0, \beta=0$, $wt(c_{\alpha, \beta})=0$. When $\alpha=0, \beta \in \mathbb{F}_{p^t}^{*}$, $wt(c_{\alpha, \beta})=|D_{F, a}|$. When $\alpha \in V_{n}^{(p)} \backslash \{0\}, \beta \in \mathbb{F}_{p^t}$, $wt(c_{\alpha, \beta})=|D_{F, a}|-N_{a, \alpha, \beta}$, where
$N_{a, \alpha, \beta}=|\{x \in V_{n}^{(p)}: F(x)=a, \sum_{j=1}^{s}Tr_{t}^{n_{j}}(\alpha_{j}x_{j})+\beta=0\}|$.
By Lemma 7,
\begin{equation*}
N_{a, \alpha, \beta}=\varepsilon p^{\frac{n}{2}-m}|\{y \in \mathbb{F}_{p^m}^{*}: Tr_{t}^{m}(yF^{*}(\alpha)-ay^{1-l})=\beta\}|+\varepsilon p^{\frac{n}{2}-t}(\delta_{0}(a)-1)+p^{n-m-t}.
\end{equation*}
In order to show that the dimension of $C_{D_{F, a}}$ is $\frac{n}{t}+1$, we only need to show that for any $\alpha \in V_{n}^{(p)} \backslash \{0\}, \beta \in \mathbb{F}_{p^t}$, $N_{a, \alpha, \beta}<|D_{F, a}|$.
\begin{itemize}
  \item If $\varepsilon=1, a=0$, we have $N_{a, \alpha, \beta}\leq p^{\frac{n}{2}-m}(p^m-1)+p^{n-m-t}<|D_{F, 0}|$;
  \item If $\varepsilon=1, a \in \mathbb{F}_{p^m}^{*}$, we have $N_{a, \alpha, \beta}\leq p^{\frac{n}{2}-m}(p^m-1)-p^{\frac{n}{2}-t}+p^{n-m-t}<|D_{F, a}|$ since $m<\frac{n}{2}$;
  \item If $\varepsilon=-1, a=0$, we have $N_{a, \alpha, \beta}\leq p^{n-m-t}\leq p^{n-m-1}<|D_{F, 0}|$ since $m<\frac{n}{2}$;
  \item If $\varepsilon=-1, a \in \mathbb{F}_{p^m}^{*}$, we have $N_{a, \alpha, \beta}\leq p^{n-m-t}+p^{\frac{n}{2}-t}\leq p^{n-m-1}+p^{\frac{n}{2}-1}<|D_{F, a}|$ since $m<\frac{n}{2}$.
\end{itemize}

Since $m<\frac{n}{2}, t \mid m$, we have $p \mid |D_{F, a}|$ and $p \mid N_{a, \alpha, \beta}$. When $p$ is odd, $c_{0, \beta}\cdot c_{0, \beta}=\beta^{2}|D_{F, a}|=0$; $c_{\alpha, \beta} \cdot c_{\alpha, \beta}=\sum_{i \in \mathbb{F}_{p^t}^{*}}N_{a, \alpha, -i}i^{2}+2\beta\sum_{i \in \mathbb{F}_{p^t}^{*}}N_{a, \alpha, -i}i+\beta^{2}|D_{F, a}|=0$ for any $\alpha \in V_{n}^{(p)} \backslash \{0\}$. By Proposition 3, when $p$ is odd, $C_{D_{F, a}}$ is self-orthogonal.

When $p=2$, for any $\alpha, \alpha' \in V_{n}^{(2)} \backslash \{0\}, i, i' \in \mathbb{F}_{2^t}^{*}$, with the same computation as in the proof of Theorem 1, we have
\begin{small}
\begin{equation*}
\begin{split}
M_{a, \alpha, \alpha', i, i'}&\triangleq |\{x \in V_{n}^{(2)}: F(x)=a, \sum_{j=1}^{s}Tr_{t}^{n_{j}}(\alpha_{j}x_{j})=i, \sum_{j=1}^{s}Tr_{t}^{n_{j}}(\alpha'_{j}x_{j})=i'\}|\\
&=2^{-m-2t}\sum_{z, w \in \mathbb{F}_{2^t}}(-1)^{Tr_{1}^{t}(iz+i'w)}\sum_{y \in \mathbb{F}_{2^m}^{*}}(-1)^{Tr_{1}^{m}(ay)}W_{F_{y}}(z\alpha+w\alpha')\\
& \ \ \ +2^{n-m-2t}(1+2^t\delta_{0}(\alpha+ii'^{-1}\alpha')-\sum_{w \in \mathbb{F}_{2^t}}\delta_{0}(\alpha+w\alpha')).
\end{split}
\end{equation*}
\end{small}Since $F$ is a vectorial dual-bent function with Condition II, we have
\begin{small}
\begin{equation*}
\begin{split}
M_{a, \alpha, \alpha', i, i'}&=2^{\frac{n}{2}-m-2t}\sum_{z, w \in \mathbb{F}_{2^t}}(-1)^{Tr_{1}^{t}(iz+i'w)}\sum_{y \in \mathbb{F}_{2^m}^{*}}(-1)^{Tr_{1}^{m}(ay)+Tr_{1}^{m}(y^{1-d}F^{*}(z\alpha+w\alpha'))}\\
& \ \ \ +2^{n-m-2t}(1+2^t\delta_{0}(\alpha+ii'^{-1}\alpha')-\sum_{w \in \mathbb{F}_{2^t}}\delta_{0}(\alpha+w\alpha')).
\end{split}
\end{equation*}
\end{small} By Lemma 7, we have
\begin{small}
\begin{equation*}
\begin{split}
M_{a, \alpha, \alpha', i, i'}&=2^{\frac{n}{2}-m-t}(\sum_{w \in \mathbb{F}_{2^t}}|\{y \in \mathbb{F}_{2^m}^{*}: Tr_{t}^{m}(F^{*}(\alpha+w\alpha')y+ay^{1-l})=i+wi'\}|\\
& \ \ \ +|\{y \in \mathbb{F}_{2^m}^{*}: Tr_{t}^{m}(F^{*}(\alpha')y+ay^{1-l})=i'\}|)+2^{n-m-t}\delta_{0}(\alpha+ii'^{-1}\alpha')+A,\\
\end{split}
\end{equation*}
\end{small}where $A=2^{n-m-2t}(1-\sum_{w \in \mathbb{F}_{2^t}}\delta_{0}(\alpha+w\alpha'))-2^{\frac{n}{2}-t}+2^{\frac{n}{2}-m-t}+2^{\frac{n}{2}-2t}(\delta_{0}(a)-1)$. Note that $A$ is an integer with $2 \mid A$ since $t \mid m$, and $m+t<\frac{n}{2}$ when $p=2$. Then $2 \mid M_{a, \alpha, \alpha', i, i'}$.

If $\alpha=\alpha'=0$, $c_{\alpha, \beta}\cdot c_{\alpha', \beta'}=\beta\beta'|D_{F, a}|=0$ since $2 \mid |D_{F, a}|$; if $\alpha=0, \alpha'\neq 0$, or $\alpha\neq 0, \alpha'=0$, w.l.o.g., $\alpha=0, \alpha' \neq 0$, then $c_{\alpha, \beta} \cdot c_{\alpha', \beta'}=\beta \sum_{i \in \mathbb{F}_{2^t}^{*}}N_{a, \alpha', i}i+\beta\beta'|D_{F, a}|=0$ since $2 \mid |D_{F, a}|, 2 \mid N_{a, \alpha', i}$; if $\alpha, \alpha' \neq 0$, by $2 \mid |D_{F, a}|$, $2 \mid N_{a, \alpha, i}$ and $2 \mid M_{a, \alpha, \alpha', i, i'}$,
\begin{small}
\begin{equation*}
c_{\alpha, \beta} \cdot c_{\alpha', \beta'}=\sum_{k \in \mathbb{F}_{2^t}^{*}}k\sum_{i \in \mathbb{F}_{2^t}^{*}}M_{a, \alpha, \alpha', i, i^{-1}k}+\beta \sum_{i \in \mathbb{F}_{2^t}^{*}}N_{a, \alpha', i}i+\beta' \sum_{i \in \mathbb{F}_{2^t}^{*}}N_{a, \alpha, i}i+\beta\beta'|D_{F, a}|=0.
\end{equation*}
\end{small}Thus, when $p=2$, $C_{D_{F, a}}$ is self-orthogonal.

For any nonempty set $I\subset \mathbb{F}_{p^m}$, since $D_{F, I}=\cup_{a \in I}D_{F, a}, D_{F, a}\cap D_{F, a'}=\emptyset, a\neq a'$, and $C_{D_{F, a}}, a \in I$, are all self-orthogonal, we have that $C_{D_{F, I}}$ is self-orthogonal.
By Lemma 6, the length of $C_{D_{F, I}}$ is $|D_{F, I}|=(p^{n-m}-\varepsilon p^{\frac{n}{2}-m})|I|+\varepsilon p^{\frac{n}{2}}\delta_{I}(0)$. It is clear that the dimension of $C_{D_{F, I}}$ is $\frac{n}{t}+1$ since the dimension of $C_{D_{F, a}}$ is $\frac{n}{t}+1$ for any $a \in \mathbb{F}_{p^m}$. With the same argument as in the proof of Theorem 1, $d(C_{D_{F, I}}^{\bot})\geq 3$. By Proposition 2, $C_{D_{F, I}}^{\bot}$ is at least almost optimal according to Hamming bound.
\end{proof}

In the following, for some sets $I$, we completely determine the weight distribution of $C_{D_{F, I}}$.

\begin{theorem}\label{Theorem 3}
Let $F: V_{n}^{(p)}\rightarrow \mathbb{F}_{p^m}$ be a vectorial dual-bent function with Condition II, and let $C_{D_{F, I}}$ be defined by Eq. (1) with $I=\{a\}$, where $a \in \mathbb{F}_{p^m}$.

$\mathrm{(i)}$ If $a=0$, then $C_{D_{F, I}}$ is a at most five-weight $[(p^{n-m}+\varepsilon p^{\frac{n}{2}-m}(p^m-1), \frac{n}{t}+1]_{p^t}$ self-orthogonal linear code whose weight distribution is given in Table 3.

$\mathrm{(ii)}$ If $a \in \mathbb{F}_{p^m}^{*}$ and $t=m, l=2$, then $C_{D_{F, I}}$ is a at most four-weight $[p^{n-m}-\varepsilon p^{\frac{n}{2}-m}, \frac{n}{m}+1]_{p^m}$ self-orthogonal linear code whose weight distribution is given in Table 4.
\end{theorem}

\renewcommand{\thetable}{3}
\begin{table}\label{3}\scriptsize
  \centering
  \caption{The weight distribution of $C_{D_{F, I}}$ constructed in Theorem 3 (i)}
  \begin{tabular}{|c|c|}
    \hline
    Hamming weight & Multiplicity \\ \hline
    $0$ & $1$ \\
    $p^{n-m}+\varepsilon p^{\frac{n}{2}-m}(p^m-1)$ & $p^t-1$ \\
    $p^{n-m-t}(p^t-1)$ & $p^{n-m}+\varepsilon p^{\frac{n}{2}-m}(p^m-1)-1$ \\
    $p^{n-m-t}(p^t-1)+\varepsilon p^{\frac{n}{2}-m}(p^m-1)$ & $(p^{n-m}+\varepsilon p^{\frac{n}{2}-m}(p^m-1)-1)(p^t-1)$ \\
    $(p^{n-m-t}+\varepsilon p^{\frac{n}{2}-t})(p^t-1)$ & $(p^{n-m}-\varepsilon p^{\frac{n}{2}-m})(p^m-1)$ \\
    $p^{n-m-t}(p^t-1)+\varepsilon p^{\frac{n}{2}-m}(p^m-p^{m-t}-1)$ & $(p^{n-m}-\varepsilon p^{\frac{n}{2}-m})(p^m-1)(p^t-1)$ \\
    \hline
  \end{tabular}
\end{table}

\renewcommand{\thetable}{4}
\begin{table}\label{4}\scriptsize
  \centering
  \caption{The weight distribution of $C_{D_{F, I}}$ constructed in Theorem 3 (ii)}
  \begin{tabular}{|c|c|}
    \hline
    Hamming weight & Multiplicity \\ \hline
    $0$ & $1$ \\
    $p^{n-m}-\varepsilon p^{\frac{n}{2}-m}$ & $p^m-1$ \\
    $p^{n-2m}(p^m-1)$ & $(p^{n-m}-\varepsilon p^{\frac{n}{2}-m})\frac{p^{2m}-p^m+2}{2}+\varepsilon p^{\frac{n}{2}}-1$ \\
    $p^{n-2m}(p^m-1)-\varepsilon p^{\frac{n}{2}-m}$ & $(2p^{n-m}-2\varepsilon p^{\frac{n}{2}-m}+\varepsilon p^{\frac{n}{2}}-1)(p^m-1)$ \\
    $p^{n-2m}(p^m-1)-2\varepsilon p^{\frac{n}{2}-m}$ & $(p^{n-m}-\varepsilon p^{\frac{n}{2}-m})\frac{(p^m-1)(p^m-2)}{2}$ \\
    \hline
  \end{tabular}
\end{table}

\begin{proof}
By Theorem 2, $C_{D_{F, a}}$ is self-orthogonal.

(i) When $\alpha=0, \beta=0$, $wt(c_{\alpha, \beta})=0$. When $\alpha=0, \beta \in \mathbb{F}_{p^t}^{*}$, $wt(c_{\alpha, \beta})=|D_{F, 0}|=p^{n-m}+\varepsilon p^{\frac{n}{2}-m}(p^m-1)$. By Lemma 7, for any $\alpha \in V_{n}^{(p)} \backslash \{0\}, \beta \in \mathbb{F}_{p^t}$,\\
\begin{small}
\begin{equation*}
\begin{split}
N_{0, \alpha, \beta}&=\varepsilon p^{\frac{n}{2}-m}|\{y \in \mathbb{F}_{p^m}^{*}: Tr_{t}^{m}(F^{*}(\alpha)y)=\beta\}|+p^{n-m-t}\\
&=\left\{\begin{split}
\varepsilon p^{\frac{n}{2}-m}(p^m-1)+p^{n-m-t}, & \ \text{ if } \alpha \neq 0, F^{*}(\alpha)=0, \beta=0,\\
p^{n-m-t}, & \ \text{ if } \alpha \neq 0, F^{*}(\alpha)=0, \beta\neq 0,\\
\varepsilon p^{\frac{n}{2}-m}(p^{m-t}-1)+p^{n-m-t}, & \ \text{ if } \alpha \neq 0, F^{*}(\alpha)\neq 0, \beta=0,\\
\varepsilon p^{\frac{n}{2}-t}+p^{n-m-t}, & \ \text{ if } \alpha \neq 0, F^{*}(\alpha)\neq 0, \beta\neq 0,
\end{split}
\right.
\end{split}
\end{equation*}
\end{small}and then
\begin{small}
\begin{equation*}
\begin{split}
wt(c_{\alpha, \beta})&=|D_{F, 0}|-N_{0, \alpha, \beta}\\
&=\left\{
\begin{split}
p^{n-m-t}(p^t-1), & \ \text{ if } \alpha \neq 0, F^{*}(\alpha)=0, \beta=0,\\
p^{n-m-t}(p^t-1)+\varepsilon p^{\frac{n}{2}-m}(p^m-1), & \ \text{ if } \alpha\neq 0, F^{*}(\alpha)=0, \beta\neq 0,\\
(p^{n-m-t}+\varepsilon p^{\frac{n}{2}-t})(p^t-1), & \ \text{ if } \alpha \neq 0, F^{*}(\alpha)\neq 0, \beta=0,\\
p^{n-m-t}(p^t-1)+\varepsilon p^{\frac{n}{2}-m}(p^m-p^{m-t}-1), & \ \text{ if } \alpha \neq 0, F^{*}(\alpha)\neq 0, \beta\neq 0.
\end{split}
\right.
\end{split}
\end{equation*}
\end{small}The weight distribution of $C_{D_{F, 0}}$ follows from the above equation and Lemma 6.

(ii) When $\alpha=0, \beta=0$, $wt(c_{\alpha, \beta})=0$. When $\alpha=0, \beta \in \mathbb{F}_{p^m}^{*}$, $wt(c_{\alpha, \beta})=|D_{F, a}|=p^{n-m}-\varepsilon p^{\frac{n}{2}-m}$. By Lemma 7, for any $\alpha \in V_{n}^{(p)} \backslash \{0\}, \beta \in \mathbb{F}_{p^m}$,
\begin{small}
\begin{equation*}
\begin{split}
N_{a, \alpha, \beta}&=\varepsilon p^{\frac{n}{2}-m}|\{y \in \mathbb{F}_{p^m}^{*}: F^{*}(\alpha)y-ay^{-1}=\beta\}|
-\varepsilon p^{\frac{n}{2}-m}+p^{n-2m}\\
&=\varepsilon p^{\frac{n}{2}-m}|\{y \in \mathbb{F}_{p^m}^{*}: F^{*}(\alpha)y^{2}-\beta y-a=0\}|-\varepsilon p^{\frac{n}{2}-m}+p^{n-2m}.
\end{split}
\end{equation*}
\end{small}If $p$ is odd, let $\mathcal{S}=\{x^{2}: x \in \mathbb{F}_{p^m}^{*}\}, \mathcal{N}=\mathbb{F}_{p^m}^{*} \backslash \mathcal{S}$, we have
\begin{small}
\begin{equation*}
\begin{split}
N_{a, \alpha, \beta}
&=\left\{
\begin{split}
-\varepsilon p^{\frac{n}{2}-m}+p^{n-2m}, & \ \text{ if } \alpha\neq 0, F^{*}(\alpha)=0, \beta=0, \text{ or } F^{*}(\alpha)\neq 0, \beta^{2}+4aF^{*}(\alpha) \in \mathcal{N},\\
p^{n-2m}, & \ \text{ if } \alpha\neq 0, F^{*}(\alpha)=0, \beta\neq 0, \text{ or } F^{*}(\alpha)\neq 0, \beta^{2}+4aF^{*}(\alpha)=0,\\
\varepsilon p^{\frac{n}{2}-m}+p^{n-2m}, & \ \text{ if } \alpha\neq 0, F^{*}(\alpha)\neq 0, \beta^{2}+4aF^{*}(\alpha) \in \mathcal{S},
\end{split}
\right.
\end{split}
\end{equation*}
\end{small}and then
\begin{small}
\begin{equation*}
\begin{split}
wt(c_{\alpha, \beta})&=|D_{F, a}|-N_{a, \alpha, \beta}\\
&=\left\{
\begin{split}
p^{n-2m}(p^m-1), & \ \text{ if } \alpha\neq 0, F^{*}(\alpha)=0, \beta=0, \text{ or } F^{*}(\alpha)\neq 0, \beta^{2}+4aF^{*}(\alpha) \in \mathcal{N},\\
p^{n-2m}(p^m-1)-\varepsilon p^{\frac{n}{2}-m}, & \ \text{ if } \alpha\neq 0, F^{*}(\alpha)=0, \beta\neq 0, \text{ or } F^{*}(\alpha)\neq 0, \beta^{2}+4aF^{*}(\alpha)=0,\\
p^{n-2m}(p^m-1)-2\varepsilon p^{\frac{n}{2}-m}, & \ \text{ if } \alpha\neq 0, F^{*}(\alpha)\neq 0, \beta^{2}+4aF^{*}(\alpha) \in \mathcal{S}.
\end{split}
\right.
\end{split}
\end{equation*}
\end{small}

If $p=2$, we have
\begin{small}
\begin{equation*}
\begin{split}
N_{a, \alpha, \beta}&=\left\{
\begin{split}
-2^{\frac{n}{2}-m}+2^{n-2m}, & \ \text{ if }\alpha \neq 0, F^{*}(\alpha)=0, \beta=0, \text{ or } F^{*}(\alpha)\neq 0, \beta \neq 0, Tr_{1}^{m}(aF^{*}(\alpha)\beta^{-2})=1,\\
2^{n-2m}, & \ \text{ if } \alpha \neq 0, F^{*}(\alpha)=0, \beta \neq 0, \text{ or }F^{*}(\alpha)\neq 0, \beta=0,\\
2^{\frac{n}{2}-m}+2^{n-2m}, & \ \text{ if }\alpha\neq 0, F^{*}(\alpha)\neq 0, \beta \neq 0, Tr_{1}^{m}(aF^{*}(\alpha)\beta^{-2})=0,
\end{split}
\right.
\end{split}
\end{equation*}
\end{small} and then
\begin{small}
\begin{equation*}
\begin{split}
wt(c_{\alpha, \beta})&=|D_{F, a}|-N_{a, \alpha, \beta}\\
&=\left\{
\begin{split}
2^{n-2m}(2^m-1), & \ \text{ if }\alpha \neq 0, F^{*}(\alpha)=0, \beta=0, \\
& \ \text{ or } F^{*}(\alpha)\neq 0, \beta \neq 0, Tr_{1}^{m}(aF^{*}(\alpha)\beta^{-2})=1,\\
2^{n-2m}(2^m-1)-2^{\frac{n}{2}-m}, & \ \text{ if } \alpha \neq 0, F^{*}(\alpha)=0, \beta \neq 0, \text{ or }F^{*}(\alpha)\neq 0, \beta=0,\\
2^{n-2m}(2^m-1)-2^{\frac{n}{2}-m+1}, & \ \text{ if }\alpha\neq 0, F^{*}(\alpha)\neq 0, \beta \neq 0, Tr_{1}^{m}(aF^{*}(\alpha)\beta^{-2})=0.
\end{split}
\right.
\end{split}
\end{equation*}
\end{small}By the above arguments and Lemma 6, the weight distribution of $C_{D_{F, a}}$ can be easily obtained.
\end{proof}

\begin{theorem}\label{Theorem 4}
Let $F: V_{n}^{(p)}\rightarrow \mathbb{F}_{p^m}$ be a vectorial dual-bent function with Condition II, and let $C_{D_{F, I}}$ be defined by Eq. (1) with $I=\gamma H_{b}$, where $\gamma \in \mathbb{F}_{p^m}^{*}$, $H_{b}=\{x^{b}: x \in \mathbb{F}_{p^m}^{*}\}$, and $b$ is an integer with $b \mid (p^m-1)$ and $b \mid l$. Then $C_{D_{F, I}}$ is a at most five-weight $[(p^{n-m}-\varepsilon p^{\frac{n}{2}-m})\frac{p^m-1}{b}, \frac{n}{t}+1]_{p^t}$ self-orthogonal linear code whose weight distribution is given in Table 5.
\end{theorem}

\renewcommand{\thetable}{5}
\begin{table}\label{5}\scriptsize
  \centering
  \caption{The weight distribution of $C_{D_{F, I}}$ constructed in Theorem 4}
  \begin{tabular}{|c|c|}
    \hline
    Hamming weight & Multiplicity \\ \hline
    $0$ & $1$ \\
    $(p^{n-m}-\varepsilon p^{\frac{n}{2}-m})\frac{p^m-1}{b}$ & $p^t-1$ \\
    $(p^{n-m-t}\frac{p^m-1}{b}-\varepsilon p^{\frac{n}{2}-t})(p^t-1)$ & $(p^{n-m}-\varepsilon p^{\frac{n}{2}-m})\frac{p^m-1}{b}$ \\
    $(p^{n-m}-p^{n-m-t}-\varepsilon p^{\frac{n}{2}-m})\frac{p^m-1}{b}+\varepsilon p^{\frac{n}{2}-t}$ & $(p^t-1)(p^{n-m}-\varepsilon p^{\frac{n}{2}-m})\frac{p^m-1}{b}$ \\
    $p^{n-m-t}(p^t-1)\frac{p^m-1}{b}$ & $p^{n}-1-(p^{n-m}-\varepsilon p^{\frac{n}{2}-m})\frac{p^m-1}{b}$ \\
    $(p^{n-m}-p^{n-m-t}-\varepsilon p^{\frac{n}{2}-m})\frac{p^m-1}{b}$ & $(p^{n}-1)(p^t-1)-(p^t-1)(p^{n-m}-\varepsilon p^{\frac{n}{2}-m})\frac{p^m-1}{b}$ \\
    \hline
  \end{tabular}
\end{table}

\begin{proof}
By Theorem 2, $C_{D_{F, \gamma H_{b}}}$ is self-orthogonal. When $\alpha=0, \beta=0$, $wt(c_{\alpha, \beta})=0$. When $\alpha=0, \beta \in \mathbb{F}_{p^t}^{*}$, $wt(c_{\alpha, \beta})=|D_{F, \gamma H_{b}}|=(p^{n-m}-\varepsilon p^{\frac{n}{2}-m})\frac{p^m-1}{b}$. By Lemma 7, for any $\alpha \in V_{n}^{(p)} \backslash \{0\}, \beta \in \mathbb{F}_{p^t}$, we have
\begin{small}
\begin{equation*}
\begin{split}
N_{\gamma H_{b}, \alpha, \beta}&\triangleq |\{x \in V_{n}^{(p)}: F(x) \in \gamma H_{b}, \sum_{j=1}^{s}Tr_{t}^{n_{j}}(\alpha_{j}x_{j})+\beta=0\}|\\
&=\sum_{a \in \gamma H_{b}}|\{x \in V_{n}^{(p)}: F(x)=a, \sum_{j=1}^{s}Tr_{t}^{n_{j}}(\alpha_{j}x_{j})+\beta=0\}|\\
&=\varepsilon p^{\frac{n}{2}-m}\sum_{a \in \gamma H_{b}}|\{y \in \mathbb{F}_{p^m}^{*}: Tr_{t}^{m}(F^{*}(\alpha)y-ay^{1-l})=\beta\}|-\varepsilon p^{\frac{n}{2}-t}\frac{p^m-1}{b}+p^{n-m-t}\frac{p^m-1}{b}\\
&=\varepsilon p^{\frac{n}{2}-m}T-\varepsilon p^{\frac{n}{2}-t}\frac{p^m-1}{b}+p^{n-m-t}\frac{p^m-1}{b},
\end{split}
\end{equation*}
\end{small}where $T=\sum_{a \in \gamma H_{b}}|\{y \in \mathbb{F}_{p^m}^{*}: Tr_{t}^{m}(F^{*}(\alpha)y-ay^{1-l})=\beta\}|$. For $T$, we have
\begin{small}
\begin{equation*}
\begin{split}
T&=p^{-t}\sum_{a \in \gamma H_{b}}\sum_{y \in \mathbb{F}_{p^m}^{*}}\sum_{z \in \mathbb{F}_{p^t}}\zeta_{p}^{Tr_{1}^{t}(z(Tr_{t}^{m}(F^{*}(\alpha)y-ay^{1-l})-\beta))}\\
&=p^{-t}\sum_{z \in \mathbb{F}_{p^t}}\zeta_{p}^{Tr_{1}^{t}(-z\beta)}\sum_{y \in \mathbb{F}_{p^m}^{*}}\zeta_{p}^{Tr_{1}^{m}(F^{*}(\alpha)yz)}\sum_{a \in \gamma H_{b}}\zeta_{p}^{Tr_{1}^{m}(-ay^{1-l}z)}\\
&=p^{-t}\sum_{z \in \mathbb{F}_{p^t}}\zeta_{p}^{Tr_{1}^{t}(-z\beta)}\sum_{y \in \mathbb{F}_{p^m}^{*}}\zeta_{p}^{Tr_{1}^{m}(F^{*}(\alpha)yz)}\sum_{a \in \gamma H_{b}}\zeta_{p}^{Tr_{1}^{m}(-ayz)}\\
&=p^{-t}\sum_{z \in \mathbb{F}_{p^t}^{*}}\zeta_{p}^{Tr_{1}^{t}(-z\beta)}\sum_{a \in \gamma H_{b}}\sum_{y \in \mathbb{F}_{p^m}^{*}}\zeta_{p}^{Tr_{1}^{m}(y(F^{*}(\alpha)-a))}+p^{-t}(p^m-1)\frac{p^m-1}{b}\\
&=p^{-t}(p^{t}\delta_{0}(\beta)-1)(p^{m}\delta_{\gamma H_{b}}(F^{*}(\alpha))-\frac{p^m-1}{b})+p^{-t}(p^m-1)\frac{p^m-1}{b},
\end{split}
\end{equation*}
\end{small}where in the third equation we use $H_{l}=\{x^{l}: x \in \mathbb{F}_{p^m}^{*}\}\subseteq H_{b}$ since $b \mid l$.

Then
\begin{itemize}
  \item when $\alpha \neq 0, F^{*}(\alpha) \in \gamma H_{b}, \beta=0$, we have $N_{\gamma H_{b}, \alpha, \beta}=(p^{n-m-t}-\varepsilon p^{\frac{n}{2}-m})\frac{p^m-1}{b}+\varepsilon p^{\frac{n}{2}-t}(p^t-1)$ and $wt(c_{\alpha, \beta})=|D_{F, \gamma H_{b}}|-N_{\gamma H_{b}, \alpha, \beta}=(p^{n-m-t}\frac{p^m-1}{b}-\varepsilon p^{\frac{n}{2}-t})(p^t-1)$;
  \item when $\alpha \neq 0, F^{*}(\alpha) \in \gamma H_{b}, \beta \in \mathbb{F}_{p^t}^{*}$, we have $N_{\gamma H_{b}, \alpha, \beta}=p^{n-m-t}\frac{p^m-1}{b}-\varepsilon p^{\frac{n}{2}-t}$ and $wt(c_{\alpha, \beta})=|D_{F, \gamma H_{b}}|-N_{\gamma H_{b}, \alpha, \beta}=(p^{n-m}-p^{n-m-t}-\varepsilon p^{\frac{n}{2}-m})\frac{p^m-1}{b}+\varepsilon p^{\frac{n}{2}-t}$;
  \item when $\alpha \neq 0, F^{*}(\alpha) \notin \gamma H_{b}, \beta=0$, we have $N_{\gamma H_{b}, \alpha, \beta}=(p^{n-m-t}-\varepsilon p^{\frac{n}{2}-m})\frac{p^m-1}{b}$ and $wt(c_{\alpha, \beta})=|D_{F, \gamma H_{b}}|-N_{\gamma H_{b}, \alpha, \beta}=p^{n-m-t}(p^t-1)\frac{p^m-1}{b}$;
  \item when $\alpha\neq 0, F^{*}(\alpha) \notin \gamma H_{b}, \beta \in \mathbb{F}_{p^t}^{*}$, we have $N_{\gamma H_{b}, \alpha, \beta}=p^{n-m-t}\frac{p^m-1}{b}$ and $wt(c_{\alpha, \beta})=|D_{F, \gamma H_{b}}|-N_{\gamma H_{b}, \alpha, \beta}=(p^{n-m}-p^{n-m-t}-\varepsilon p^{\frac{n}{2}-m})\frac{p^m-1}{b}$.
\end{itemize}
The weight distribution of $C_{D_{F, \gamma H_{b}}}$ follows from the above arguments and Lemma 6.
\end{proof}

By Theorem 4, we have the following corollary.

\begin{corollary} \label{Corollary 1}
Let $p$ be an odd prime. Let $F: V_{n}^{(p)}\rightarrow \mathbb{F}_{p^m}$ be a vectorial dual-bent function with Condition II, and $C_{D_{F, I}}$ be defined by Eq. (1) with $I=\mathcal{S}$ or $I=\mathcal{N}$, where $\mathcal{S}=\{x^{2}: x \in \mathbb{F}_{p^m}^{*}\}, \mathcal{N}=\mathbb{F}_{p^m}^{*} \backslash \mathcal{S}$. Then $C_{D_{F, I}}$ is a at most five-weight $[(p^{n-m}-\varepsilon p^{\frac{n}{2}-m})\frac{p^m-1}{2}, \frac{n}{t}+1]_{p^t}$ self-orthogonal linear code whose weight distribution is given in Table 5 with $b=2$.
\end{corollary}

\begin{proof}
When $p$ is odd, by $gcd(p^m-1, l-1)=1$, we have $2 \mid l$. Then the result follows from Theorem 4.
\end{proof}

\begin{theorem}\label{Theorem 5}
Let $F: V_{n}^{(p)}\rightarrow \mathbb{F}_{p^m}$ be a vectorial dual-bent function satisfying Condition II for which $l=2, t=m=2jj'$ for some positive integers $j, j'$, and there is an integer $b\geq 2$ such that $b \mid (p^j+1)$, where $j$ is the smallest such positive integer. Let $C_{D_{F, I}}$ be defined by Eq. (1) with $I=\gamma H_{b}$, where $\gamma \in \mathbb{F}_{p^m}^{*}$ and $H_{b}=\{x^{b}: x \in \mathbb{F}_{p^m}^{*}\}$.

$\mathrm{(i)}$ When $b$ is odd, $C_{D_{F, I}}$ is a at most five-weight $[(p^{n-m}-\varepsilon p^{\frac{n}{2}-m})\frac{p^m-1}{b}, \frac{n}{m}+1]_{p^m}$ self-orthogonal linear code whose weight distribution is given in Table 6.

$\mathrm{(ii)}$ When $b$ is even, $C_{D_{F, I}}$ is a at most six-weight $[(p^{n-m}-\varepsilon p^{\frac{n}{2}-m})\frac{p^m-1}{b}, \frac{n}{m}+1]_{p^m}$ self-orthogonal linear code whose weight distribution is given in Table 7.
\end{theorem}

\renewcommand{\thetable}{6}
\begin{table}\label{6}\scriptsize
  \centering
  \caption{The weight distribution of $C_{D_{F, I}}$ constructed in Theorem 5 (i)}
  \begin{tabular}{|c|c|}
    \hline
    Hamming weight & Multiplicity \\ \hline
    $0$ & $1$ \\
    $(p^{n-m}-\varepsilon p^{\frac{n}{2}-m})\frac{p^m-1}{b}$ & $p^{m}-1$ \\
    $(p^{n-m}-p^{n-2m})\frac{p^m-1}{b}$ & $p^{n-m}+\varepsilon p^{\frac{n}{2}-m}(p^m-1)-1$ \\
    $(p^{n-m}-p^{n-2m}-\varepsilon p^{\frac{n}{2}-m})\frac{p^m-1}{b}$ & $(p^m-1)(2p^{n-m}-2\varepsilon p^{\frac{n}{2}-m}+\varepsilon p^{\frac{n}{2}}-1)$ \\
    $(p^{n-m}-p^{n-2m}-\varepsilon p^{\frac{n}{2}-m})\frac{p^m-1}{b}+\varepsilon (-1)^{j'}p^{\frac{n-m}{2}}-\varepsilon p^{\frac{n}{2}-m}\frac{(-1)^{j'}p^{\frac{m}{2}}-1}{b}$ & $\frac{(p^m-1)^{2}}{b}(p^{n-m}-\varepsilon p^{\frac{n}{2}-m})$ \\
    $(p^{n-m}-p^{n-2m}-\varepsilon p^{\frac{n}{2}-m})\frac{p^m-1}{b}-\varepsilon p^{\frac{n}{2}-m}\frac{(-1)^{j'}p^{\frac{m}{2}}-1}{b}$ & $(p^m-1)(p^m-1-\frac{p^m-1}{b})(p^{n-m}-\varepsilon p^{\frac{n}{2}-m})$ \\
    \hline
  \end{tabular}
\end{table}

\renewcommand{\thetable}{7}
\begin{table}\label{7}\scriptsize
  \centering
  \caption{The weight distribution of $C_{D_{F, I}}$ constructed in Theorem 5 (ii)}
  \begin{tabular}{|c|c|}
    \hline
    Hamming weight & Multiplicity  \\ \hline
    $0$ & $1$ \\
    $(p^{n-m}-\varepsilon p^{\frac{n}{2}-m})\frac{p^m-1}{b}$ & $p^{m}-1$ \\
    $(p^{n-m}-p^{n-2m})\frac{p^m-1}{b}$ & $\frac{p^m+1}{2}(p^{n-m}-\varepsilon p^{\frac{n}{2}-m})+ \varepsilon p^{\frac{n}{2}}-1$ \\
    $(p^{n-m}-p^{n-2m}-\varepsilon p^{\frac{n}{2}-m})\frac{p^m-1}{b}$ & $(p^m-1)(\frac{p^m+1}{2}(p^{n-m}-\varepsilon p^{\frac{n}{2}-m})+ \varepsilon p^{\frac{n}{2}}-1)$ \\
    $(p^{n-m}-p^{n-2m}-2\varepsilon p^{\frac{n}{2}-m})\frac{p^m-1}{b}$ & $\frac{p^m-1}{2}(p^{n-m}-\varepsilon p^{\frac{n}{2}-m})$ \\
    $(p^{n-m}-p^{n-2m}-\varepsilon p^{\frac{n}{2}-m})\frac{p^m-1}{b}+\varepsilon (-1)^{j'}p^{\frac{n-m}{2}}-2\varepsilon p^{\frac{n}{2}-m}\frac{(-1)^{j'}p^{\frac{m}{2}}-1}{b}$ & $\frac{(p^m-1)^{2}}{b}(p^{n-m}-\varepsilon p^{\frac{n}{2}-m})$ \\
    $(p^{n-m}-p^{n-2m}-\varepsilon p^{\frac{n}{2}-m})\frac{p^m-1}{b}-2\varepsilon p^{\frac{n}{2}-m}\frac{(-1)^{j'}p^{\frac{m}{2}}-1}{b}$ &
    $(p^m-1)(\frac{p^m-1}{2}-\frac{p^m-1}{b})(p^{n-m}-\varepsilon p^{\frac{n}{2}-m})$\\
    \hline
  \end{tabular}
\end{table}

\begin{proof}
By Theorem 2, $C_{D_{F, \gamma H_{b}}}$ is self-orthogonal. Let $w$ be a primitive element of $\mathbb{F}_{p^m}$, and if $p$ is odd, let $\mathcal{S}=\{x^{2}: x \in \mathbb{F}_{p^m}^{*}\}, \mathcal{N}=\mathbb{F}_{p^m}^{*} \backslash \mathcal{S}$.

When $\alpha=0, \beta=0$, $wt(c_{\alpha, \beta})=0$. When $\alpha=0, \beta \in \mathbb{F}_{p^m}^{*}$, $wt(c_{\alpha, \beta})=|D_{F, \gamma H_{b}}|=(p^{n-m}-\varepsilon p^{\frac{n}{2}-m})\frac{p^m-1}{b}$. By Lemma 7, for any $\alpha \in V_{n}^{(p)} \backslash \{0\}, \beta \in \mathbb{F}_{p^m}$, we have
\begin{small}
\begin{equation*}
\begin{split}
N_{\gamma H_{b}, \alpha, \beta}&\triangleq |\{x \in V_{n}^{(p)}: F(x) \in \gamma H_{b}, \sum_{i=1}^{s}Tr_{m}^{n_{i}}(\alpha_{i}x_{i})+\beta=0\}|\\
&=\sum_{a \in \gamma H_{b}}|\{x \in V_{n}^{(p)}: F(x)=a, \sum_{i=1}^{s}Tr_{m}^{n_{i}}(\alpha_{i}x_{i})+\beta=0\}|\\
\end{split}
\end{equation*}
\begin{equation}\label{6}
\begin{split}
&=\varepsilon p^{\frac{n}{2}-m}\sum_{a \in \gamma H_{b}}|\{y \in \mathbb{F}_{p^m}^{*}: F^{*}(\alpha)y-ay^{-1}=\beta\}|+(p^{n-2m}-\varepsilon p^{\frac{n}{2}-m})\frac{p^m-1}{b}\\
&=\varepsilon p^{\frac{n}{2}-m}\sum_{a \in \gamma H_{b}}|\{y \in \mathbb{F}_{p^m}^{*}: F^{*}(\alpha)y^{2}-\beta y-a=0\}|+(p^{n-2m}-\varepsilon p^{\frac{n}{2}-m})\frac{p^m-1}{b}\\
&=\varepsilon p^{\frac{n}{2}-m}T+(p^{n-2m}-\varepsilon p^{\frac{n}{2}-m})\frac{p^m-1}{b},
\end{split}
\end{equation}
\end{small}where $T=\sum_{a \in \gamma H_{b}}|\{y \in \mathbb{F}_{p^m}^{*}: F^{*}(\alpha)y^{2}-\beta y-a=0\}|$.

(i) By Eq. (6) and Lemma 9,
\begin{itemize}
  \item when $\alpha \neq 0, F^{*}(\alpha)=0, \beta=0$, $N_{\gamma H_{b}, \alpha, \beta}=(p^{n-2m}-\varepsilon p^{\frac{n}{2}-m})\frac{p^m-1}{b}$ and $wt(c_{\alpha, \beta})=|D_{F, \gamma H_{b}}|-N_{\gamma H_{b}, \alpha, \beta}=(p^{n-m}-p^{n-2m})\frac{p^m-1}{b}$;
  \item when $\alpha \neq 0, F^{*}(\alpha)=0, \beta \in \mathbb{F}_{p^m}^{*}$, or $F^{*}(\alpha)\neq 0, \beta=0$, $N_{\gamma H_{b}, \alpha, \beta}=p^{n-2m}\frac{p^m-1}{b}$ and $wt(c_{\alpha, \beta})=|D_{F, \gamma H_{b}}|-N_{\gamma H_{b}, \alpha, \beta}=(p^{n-m}-p^{n-2m}-\varepsilon p^{\frac{n}{2}-m})\frac{p^m-1}{b}$.
\end{itemize}
If $p$ is odd, by Eq. (6) and Lemma 9,
\begin{itemize}
  \item when $\alpha \neq 0, F^{*}(\alpha)\neq 0, \beta \neq 0$, $\gamma^{-1}F^{*}(\alpha)^{-1}\in \mathcal{S}, \beta \sqrt{\gamma^{-1}F^{*}(\alpha)^{-1}} \in H_{b}$, or $\gamma^{-1}F^{*}(\alpha)^{-1}\in \mathcal{N}, \beta \sqrt{\gamma^{-1}F^{*}(\alpha)^{-1}w^{b}} \in H_{b}$, $N_{\gamma H_{b}, \alpha, \beta}=p^{n-2m}\frac{p^m-1}{b}+\varepsilon (-1)^{j'+1}p^{\frac{n-m}{2}}+\varepsilon p^{\frac{n}{2}-m}\frac{(-1)^{j'}p^{\frac{m}{2}}-1}{b}$ and $wt(c_{\alpha, \beta})=|D_{F, \gamma H_{b}}|-N_{\gamma H_{b}, \alpha, \beta}=(p^{n-m}-p^{n-2m}-\varepsilon p^{\frac{n}{2}-m})\frac{p^m-1}{b}+\varepsilon (-1)^{j'}p^{\frac{n-m}{2}}-\varepsilon p^{\frac{n}{2}-m}\frac{(-1)^{j'}p^{\frac{m}{2}}-1}{b}$;
  \item when $\alpha \neq 0, F^{*}(\alpha)\neq 0, \beta \neq 0$, $\gamma^{-1}F^{*}(\alpha)^{-1}\in \mathcal{S}$, $\beta \sqrt{\gamma^{-1}F^{*}(\alpha)^{-1}} \notin H_{b}$, or $\gamma^{-1}F^{*}(\alpha)^{-1}\in \mathcal{N}, \beta \sqrt{\gamma^{-1}F^{*}(\alpha)^{-1}w^{b}} \notin H_{b}$, $N_{\gamma H_{b}, \alpha, \beta}=p^{n-2m}\frac{p^m-1}{b}+\varepsilon p^{\frac{n}{2}-m}\frac{(-1)^{j'}p^{\frac{m}{2}}-1}{b}$ and $wt(c_{\alpha, \beta})=|D_{F, \gamma H_{b}}|-N_{\gamma H_{b}, \alpha, \beta}=(p^{n-m}-p^{n-2m}-\varepsilon p^{\frac{n}{2}-m})\frac{p^m-1}{b}-\varepsilon p^{\frac{n}{2}-m}\frac{(-1)^{j'}p^{\frac{m}{2}}-1}{b}$.
\end{itemize}

If $p=2$, by Eq. (6) and Lemma 9,
\begin{itemize}
  \item when $\alpha \neq 0, F^{*}(\alpha)\neq 0, \beta \neq 0, \gamma^{-1}F^{*}(\alpha)^{-1} \beta^{2}\in H_{b}$, $N_{\gamma H_{b}, \alpha, \beta}=2^{n-2m}\frac{2^m-1}{b}+(-1)^{j'+1}2^{\frac{n-m}{2}}+2^{\frac{n}{2}-m}\frac{(-1)^{j'}2^{\frac{m}{2}}-1}{b}$ and $wt(c_{\alpha, \beta})=|D_{F, \gamma H_{b}}|-N_{\gamma H_{b}, \alpha, \beta}=(2^{n-m}-2^{n-2m}-2^{\frac{n}{2}-m})\frac{2^m-1}{b}+(-1)^{j'}2^{\frac{n-m}{2}}-2^{\frac{n}{2}-m}\frac{(-1)^{j'}2^{\frac{m}{2}}-1}{b}$;
  \item when $\alpha \neq 0, F^{*}(\alpha)\neq 0, \beta \neq 0, \gamma^{-1}F^{*}(\alpha)^{-1} \beta^{2}\notin H_{b}$, $N_{\gamma H_{b}, \alpha, \beta}=2^{n-2m}\frac{2^m-1}{b}+2^{\frac{n}{2}-m}\frac{(-1)^{j'}2^{\frac{m}{2}}-1}{b}$ and $wt(c_{\alpha, \beta})=|D_{F, \gamma H_{b}}|-N_{\gamma H_{b}, \alpha, \beta}=(2^{n-m}-2^{n-2m}-2^{\frac{n}{2}-m})\frac{p^m-1}{b}-2^{\frac{n}{2}-m}\frac{(-1)^{j'}2^{\frac{m}{2}}-1}{b}$.
\end{itemize}
By the above arguments and Lemma 6, the weight distribution of $C_{D_{F, \gamma H_{b}}}$ can be easily obtained.

(ii) By Eq. (6) and Lemma 9,
\begin{itemize}
  \item when $\alpha \neq 0, F^{*}(\alpha)=0, \beta=0$, or $F^{*}(\alpha)\neq 0, \beta=0, \gamma F^{*}(\alpha)^{-1} \in \mathcal{N}$, $N_{\gamma H_{b}, \alpha, \beta}=(p^{n-2m}-\varepsilon p^{\frac{n}{2}-m})\frac{p^m-1}{b}$ and $wt(c_{\alpha, \beta})=|D_{F, \gamma H_{b}}|-N_{\gamma H_{b}, \alpha, \beta}=(p^{n-m}-p^{n-2m})\frac{p^m-1}{b}$;
  \item when $\alpha \neq 0, F^{*}(\alpha)=0, \beta \neq 0$, or $F^{*}(\alpha) \neq 0, \beta \neq 0, \gamma^{-1}F^{*}(\alpha)^{-1} \in \mathcal{N}$, $N_{\gamma H_{b}, \alpha, \beta}=p^{n-2m}\frac{p^m-1}{b}$ and $wt(c_{\alpha, \beta})=|D_{F, \gamma H_{b}}|-N_{\gamma H_{b}, \alpha, \beta}=(p^{n-m}-p^{n-2m}-\varepsilon p^{\frac{n}{2}-m})\frac{p^m-1}{b}$;
  \item when $\alpha \neq 0, F^{*}(\alpha) \neq 0, \beta=0, \gamma F^{*}(\alpha)^{-1} \in \mathcal{S}$, $N_{\gamma H_{b}, \alpha, \beta}=(p^{n-2m}+\varepsilon p^{\frac{n}{2}-m})\frac{p^m-1}{b}$ and $wt(c_{\alpha, \beta})=|D_{F, \gamma H_{b}}|-N_{\gamma H_{b}, \alpha, \beta}=(p^{n-m}-p^{n-2m}-2\varepsilon p^{\frac{n}{2}-m})\frac{p^m-1}{b}$;
  \item when $\alpha\neq 0, F^{*}(\alpha)\neq 0, \beta \neq 0, \gamma^{-1}F^{*}(\alpha)^{-1} \in \mathcal{S}, \beta \sqrt{\gamma^{-1}F^{*}(\alpha)^{-1}} \in H_{\frac{b}{2}}$, $N_{\gamma H_{b}, \alpha, \beta}=p^{n-2m}\frac{p^m-1}{b}+\varepsilon (-1)^{j'+1}p^{\frac{n-m}{2}}+2\varepsilon p^{\frac{n}{2}-m}\frac{(-1)^{j'}p^{\frac{m}{2}}-1}{b}$, $wt(c_{\alpha, \beta})=(p^{n-m}-p^{n-2m}-\varepsilon p^{\frac{n}{2}-m})\frac{p^m-1}{b}+\varepsilon (-1)^{j'}p^{\frac{n-m}{2}}-2\varepsilon p^{\frac{n}{2}-m}\frac{(-1)^{j'}p^{\frac{m}{2}}-1}{b}$;
  \item when $\alpha\neq 0, F^{*}(\alpha)\neq 0, \beta \neq 0, \gamma^{-1}F^{*}(\alpha)^{-1} \in \mathcal{S}, \beta \sqrt{\gamma^{-1}F^{*}(\alpha)^{-1}} \notin H_{\frac{b}{2}}$, $N_{\gamma H_{b}, \alpha, \beta}=p^{n-2m}\frac{p^m-1}{b}+2\varepsilon p^{\frac{n}{2}-m}\frac{(-1)^{j'}p^{\frac{m}{2}}-1}{b}$, $wt(c_{\alpha, \beta})=(p^{n-m}-p^{n-2m}-\varepsilon p^{\frac{n}{2}-m})\frac{p^m-1}{b}-2\varepsilon p^{\frac{n}{2}-m}\frac{(-1)^{j'}p^{\frac{m}{2}}-1}{b}$.
\end{itemize}
By the above arguments and Lemma 6, the weight distribution of $C_{D_{F, \gamma H_{b}}}$ can be easily obtained.
\end{proof}

In the following, for general $m$, by the results in \cite{CMP2018Ve,WF2023Ne,WSWF2023Co}, we list some explicit classes of vectorial dual-bent functions $F: V_{n}^{(p)}\rightarrow \mathbb{F}_{p^m}$ satisfying the conditions in Theorems 2-5.

\begin{itemize}
  \item Let $m, n', t, u$ be positive integers with $t \mid m, m \mid n', m\neq n', gcd(u, p^{n'}-1)=1$, and when $p=2$, $m\geq 2$ and $m+t<n'$. Let $\alpha \in \mathbb{F}_{p^{n'}}^{*}$. Define $F: \mathbb{F}_{p^{n'}} \times \mathbb{F}_{p^{n'}}\rightarrow \mathbb{F}_{p^m}$ as
     \begin{equation}\label{7}
     F(x_{1}, x_{2})=Tr_{m}^{n'}(\alpha x_{1}x_{2}^{u}).
     \end{equation}
  Then $F$ is a vectorial dual-bent function satisfying Condition II with $l=1+u, d=1+u', \varepsilon=1$, where $uu'\equiv 1 \mod (p^{n'}-1)$. In details, $F$ satisfies the condition in Theorem 2, Theorem 3 (i), Theorem 4, Corollary 1; $F$ satisfies the condition in Theorem 3 (ii) when $u=1, t=m$; $F$ satisfies the condition in Theorem 5 when $u=1, t=m=2jj'$ for some integers $j, j'$, and there is an integer $b\geq 2$ with $b \mid (p^{j}+1)$, where $j$ is the smallest such positive integer.
  \item Let $p$ be an odd prime, $t, m, s$ be positive integers with $t \mid m, 2 \mid s, s\neq 2$. By the results in \cite{WF2023Ne,WSWF2023Co}, all non-degenerate quadratic forms $F$ from $\mathbb{F}_{p^m}^{s}$ ($\mathbb{F}_{p^{ms}}$) to $\mathbb{F}_{p^m}$ are vectorial dual-bent functions satisfying Condition II with $l=d=2$. We list some explicit non-degenerate quadratic forms.
  \begin{itemize}
    \item Let $m, n, t$ be positive integers with $t \mid m, 2m \mid n, 2m\neq n$, $\alpha \in \mathbb{F}_{p^n}^{*}$. Define $F: \mathbb{F}_{p^n}\rightarrow \mathbb{F}_{p^m}$ as
    \begin{equation}\label{8}
    F(x)=Tr_{m}^{n}(\alpha x^{2}).
    \end{equation}
  Then $F$ is a vectorial dual-bent function satisfying Condition II with $l=d=2, \varepsilon=-\epsilon^{n}\eta_{n}(\alpha)$.
  \item Let $m, t, s$ be positive integers with $t \mid m, 2 \mid s, s\neq 2$, $\alpha_{i} \in \mathbb{F}_{p^m}^{*}, 1\leq i \leq s$. Define $F: \mathbb{F}_{p^m}^{s}\rightarrow \mathbb{F}_{p^m}$ as
  \begin{equation}\label{9}
  F(x_{1}, \dots, x_{s})=\sum_{i=1}^{s}\alpha_{i}x_{i}^{2}.
  \end{equation}
  Then $F$ is a vectorial dual-bent function satisfying Condition II with $l=d=2, \varepsilon=\epsilon^{ms}\eta_{m}(\alpha_{1} \cdots \alpha_{s})$.
  \item Let $m, n, t$ be positive integers with $t \mid m, 2m \mid n, 2m\neq n$, $\alpha \in \mathbb{F}_{p^{\frac{n}{2}}}^{*}$. Define $F:
  \mathbb{F}_{p^n}\rightarrow \mathbb{F}_{p^m}$ as
  \begin{equation}\label{10}
  F(x)=Tr_{m}^{\frac{n}{2}}(\alpha x^{p^{\frac{n}{2}}+1}).
  \end{equation}
  Then $F$ is a vectorial dual-bent function satisfying Condition II with $l=d=2, \varepsilon=-1$.
  \end{itemize}
  In details, for $F$ defined by Eq. (8)-(10), $F$ satisfies the condition in Theorem 2, Theorem 3 (i), Corollary 1; $F$ satisfies the condition in Theorem 3 (ii) when $t=m$; $F$ satisfies the condition in Theorem 5 when $t=m=2jj'$ for some integers $j, j'$, and there is an integer $b\geq 2$ with $b \mid (p^{j}+1)$, where $j$ is the smallest such positive integer.
  \item Let $p$ be an odd prime. Let $m, t, n', n''$ be positive integers with $t \mid m, 2m \mid n', m \mid n''$. For $i \in \mathbb{F}_{p^m}$, let $H(i; x): \mathbb{F}_{p^{n'}}\rightarrow \mathbb{F}_{p^m}$ be given by $H(0; x)=Tr_{m}^{n'}(\alpha_{1}x^{2})$, $H(i; x)=Tr_{m}^{n'}(\alpha_{2}x^{2})$ if $i$ is a square in $\mathbb{F}_{p^m}^{*}$, $H(i; x)=Tr_{m}^{n'}(\alpha_{3}x^{2})$ if $i$ is a non-square in $\mathbb{F}_{p^m}^{*}$, where $\alpha_{1}, \alpha_{2}, \alpha_{3}$ are all square elements or all non-square elements in $\mathbb{F}_{p^{n'}}^{*}$. Let $G: \mathbb{F}_{p^{n''}} \times \mathbb{F}_{p^{n''}}\rightarrow \mathbb{F}_{p^m}$ be given by $G(y_{1}, y_{2})=Tr_{m}^{n''}(\beta y_{1}L(y_{2}))$, where $\beta \in \mathbb{F}_{p^{n''}}^{*}$, $L(x)=\sum a_{i}x^{p^{mi}}$ is a $p^{m}$-polynomial over $\mathbb{F}_{p^{n''}}$ inducing a permutation of $\mathbb{F}_{p^{n''}}$. Let $\gamma \in \mathbb{F}_{p^{n''}}^{*}$. Define $F: \mathbb{F}_{p^{n'}} \times \mathbb{F}_{p^{n''}} \times \mathbb{F}_{p^{n''}}\rightarrow \mathbb{F}_{p^m}$ as
      \begin{equation}\label{11}
      F(x, y_{1}, y_{2})=H(Tr_{m}^{n''}(\gamma y_{2}^{2}); x)+G(y_{1}, y_{2}).
      \end{equation}
  Then $F$ is a vectorial dual-bent function satisfying Condition II with $l=d=2, \varepsilon=-\epsilon^{n'}\eta_{n'}(\alpha_{1})$. In details, $F$ satisfies the condition in Theorem 2, Theorem 3 (i), Corollary 1; $F$ satisfies the condition in Theorem 3 (ii) when $t=m$; $F$ satisfies the condition in Theorem 5 when $t=m=2jj'$ for some integers $j, j'$, and there is an integer $b\geq 2$ with $b \mid (p^{j}+1)$, where $j$ is the smallest such positive integer.
\end{itemize}

\begin{remark}\label{Remark 1}
Let $p$ be odd. In Theorem 4.6 of \cite{WH2023Se}, Wang and Heng showed that when $F$ is defined by Eq. (10) with $t=m$, the linear code $C_{D_{F, H_{b}}}$ defined by Eq. (1) is self-orthogonal, where $b\geq 2, b \mid (p^m-1)$, and $H_{b}=\{x^b: x \in \mathbb{F}_{p^m}^{*}\}$. However, the weight distribution is not determined. They conjectured that $C_{D_{F, H_{b}}}$ has five weights. When $b=2$, or $m=2jj', b \mid (p^j+1)$ (where $j, j'$ are positive integers, and $j$ is the smallest such positive integer), we compute the weight distribution in Corollary 1 and Theorem 5, respectively. By Theorem 5, $C_{D_{F, H_{b}}}$ can be a six-weight self-orthogonal linear code.
\end{remark}

We give some examples to illustrate Theorems 2-5.
\begin{example}\label{Example 1}
Let $p=3, t=1, m=2, n=8$, and $I$ be a nonempty proper subset of $\mathbb{F}_{3^2}$. Let $F: \mathbb{F}_{3^8}\rightarrow \mathbb{F}_{3^2}$ be defined by $F(x)=Tr_{2}^{8}(x^{2})$. Then by Eq. (8), $F$ is a vectorial dual-bent function satisfying Condition II with $l=d=2, \varepsilon=-1$. By Theorem 2, $C_{D_{F, I}}$ defined by Eq. (1) is a $[738|I|-81\delta_{I}(0), 9]_{3}$ self-orthogonal code. Furthermore, when $I=\{0\}$, by Theorem 3 (i), $C_{D_{F, 0}}$ defined by Eq. (1) is a five-weight $[657, 9, 414]_{3}$ self-orthogonal code with weight enumerator $1+1312z^{414}+5904z^{432}+11808z^{441}+656z^{486}+2z^{657}$, and its dual code is a $[657, 648, 3]_{3}$ linear code which is at least almost optimal; when $I=H_{2}$, by Theorem 4, $C_{D_{F, H_{2}}}$ defined by Eq. (1) is a five-weight $[2952, 9, 1944]_{3}$ self-orthogonal code with weight enumerator $1+3608z^{1944}+5904z^{1953}+7216z^{1980}+2952z^{1998}+2z^{2952}$, and its dual code is a $[2952, 2943, 3]_{3}$ linear code which is at least almost optimal.
\end{example}

\begin{example}\label{Example 2}
Let $p=3, t=2, m=2, n=8$. Let $F: \mathbb{F}_{3^4} \times \mathbb{F}_{3^2} \times \mathbb{F}_{3^2}$ be defined by $F(x, y_{1}, y_{2})=y_{2}^{8}Tr_{2}^{4}((1-w^{2})x^{2})+Tr_{2}^{4}(w^{2}x^{2})+y_{1}y_{2}$, where $w$ is a primitive element of $\mathbb{F}_{3^4}$. Then by Eq. (11), $F$ is a vectorial dual-bent function satisfying Condition II with $l=d=2, \varepsilon=-1$. By Theorem 3 (ii), $C_{D_{F, 1}}$ defined by Eq. (1) is a four-weight $[738, 5, 648]_{9}$ self-orthogonal linear code with weight enumerator $1+27224z^{648}+11152z^{657}+20664z^{666}+8z^{738}$, and its dual code is a $[738, 733, 3]_{9}$ linear code which is at least almost optimal.
\end{example}

\begin{example} \label{Example 3}
Let $p=5, t=2, m=2, n=8, b=6$. Let $F: \mathbb{F}_{5^8} \rightarrow \mathbb{F}_{5^2}$ be defined by $F(x)=Tr_{2}^{4}(x^{626})$. Then by Eq. (10), $F$ is a vectorial dual-bent function satisfying Condition II with $l=d=2, \varepsilon=-1$. By Theorem 5, $C_{D_{F, H_{6}}}$ defined by Eq. (1) is a six-weight $[62600, 5, 60000]_{25}$ self-orthogonal code with weight enumerator $1+202824z^{60000}+3004800z^{60050}+4867776z^{60100}+1502400z^{60175}+187800z^{60200}+24z^{62600}$. This example shows that the self-orthogonal linear code given in Theorem 4.6 of \cite{WH2023Se} can have six-weights.
\end{example}

\section{Self-orthogonal codes from vectorial dual-bent functions with Condition III}
\label{sec: 5}
Note that in the third point of Condition II, the corresponding $\varepsilon_{F_{c}}$ is independent of $c$. In this section, when $\varepsilon_{F_{c}}$ depends on $c$, we construct self-orthogonal codes from vectorial dual-bent functions with the following condition:

Condition III: Let $p$ be an odd prime. Let $n, n_{j}, 1\leq j \leq s, m, t$ be positive integers for which $n=\sum_{j=1}^{s}n_{j}, t \mid n_{j}, 1\leq j \leq s, t \mid m, 2 \mid (n-m), 3m\leq n, (n, p^t)\neq (3, 3)$, and let $V_{n}^{(p)}=\mathbb{F}_{p^{n_{1}}} \times \mathbb{F}_{p^{n_{2}}} \times \dots \times \mathbb{F}_{p^{n_{s}}}$. Let $F: V_{n}^{(p)}\rightarrow \mathbb{F}_{p^m}$ be a vectorial dual-bent function satisfying
\begin{itemize}
  \item There is a vectorial dual $F^{*}$ such that $(F_{c})^{*}=(F^{*})_{c^{1-d}}, c \in \mathbb{F}_{p^m}^{*}$, where $gcd(d-1, p^m-1)=1$;
  \item $F(ax)=a^{l}F(x), a \in \mathbb{F}_{p^t}^{*}, x \in V_{n}^{(p)}$, and $F(0)=0$, where $(l-1)(d-1)\equiv 1 \mod (p^m-1)$;
  \item All component functions $F_{c}, c \in \mathbb{F}_{p^m}^{*}$, are weakly regular with $\varepsilon_{F_{c}}=\vartheta \eta_{m}(c), c \in \mathbb{F}_{p^m}^{*}$, where $\vartheta \in \{\pm \epsilon^{m}\}$ is a constant.
\end{itemize}

\subsection{Some lemmas}
\label{sec:5.1}

In this subsection, we give some useful lemmas.

\begin{lemma} \label{Lemma 10}
Let $F$ be a vectorial dual-bent function with Condition III. Then the vectorial dual $F^{*}$ with $(F_{c})^{*}=(F^{*})_{c^{1-d}}, c \in \mathbb{F}_{p^m}^{*}$, is a vectorial dual-bent function for which $((F^{*})_{c})^{*}=F_{c^{1-l}}, c \in \mathbb{F}_{p^m}^{*}$, $F^{*}(ax)=a^{d}F^{*}(x), a \in \mathbb{F}_{p^t}^{*}$, $F^{*}(0)=0$, and all component functions $(F^{*})_{c}, c \in \mathbb{F}_{p^m}^{*}$, are weakly regular with $\varepsilon_{(F^{*})_{c}}=\vartheta^{-1}\eta_{m}(c)$.
\end{lemma}

\begin{proof}
Since $F_{c}, c \in \mathbb{F}_{p^m}^{*}$, are all weakly regular bent with $\varepsilon_{F_{c}}=\vartheta\eta_{m}(c)$, for any $c \in \mathbb{F}_{p^m}^{*}$, we have that $(F^{*})_{c}=(F_{c^{1-l}})^{*}$ is weakly regular bent with $\varepsilon_{(F^{*})_{c}}=\vartheta^{-1}\eta_{m}(c^{1-l})=\vartheta^{-1}\eta_{m}(c)$ since $l$ is even. For any $c \in \mathbb{F}_{p^m}^{*}$, $((F^{*})_{c})^{*}(x)=((F_{c^{1-l}})^{*})^{*}(x)=F_{c^{1-l}}(-x)=F_{c^{1-l}}(x)$, and thus $F^{*}$ is vectorial dual-bent. By Proposition II.1 of \cite{OP2022Tw}, $(F^{*})_{c}(0)=(F_{c^{1-l}})^{*}(0)=0, c \in \mathbb{F}_{p^m}^{*}$, and then $F^{*}(0)=0$. With the similar arguments as in the proof of Lemma 5, we have $F^{*}(ax)=a^{d}F^{*}(x)$ for any $a \in \mathbb{F}_{p^t}^{*}, x \in V_{n}^{(p)}$.
\end{proof}

\begin{lemma} \label{Lemma 11}
Let $F$ be a vectorial dual-bent function with Condition III. Then the value distributions of $F$ and $F^{*}$ are give by\\
\begin{equation*}
\begin{split}
&|D_{F, 0}|=|D_{F^{*}, 0}|=p^{n-m}, \\
&|D_{F, i}|=p^{n-m}+\vartheta (-1)^{m-1}\epsilon^{m}\eta_{m}(-i)p^{\frac{n-m}{2}}, i \in \mathbb{F}_{p^m}^{*},\\
&|D_{F^{*}, i}|=p^{n-m}+\vartheta^{-1}(-1)^{m-1}\epsilon^{m}\eta_{m}(-i)p^{\frac{n-m}{2}}, i \in \mathbb{F}_{p^m}^{*}.
\end{split}
\end{equation*}
\end{lemma}

\begin{proof}
By Proposition 1 of \cite{WF2023Ne}, for any $i \in \mathbb{F}_{p^m}$ and any function $G: V_{n}^{(p)}\rightarrow \mathbb{F}_{p^m}$,
\begin{equation*}
|D_{G, i}|=p^{n-m}+p^{-m}\sum_{c \in \mathbb{F}_{p^m}^{*}}W_{G_{c}}(0)\zeta_{p}^{-Tr_{1}^{m}(ci)}.
\end{equation*}
Since $F$ is a vectorial dual-bent function with Condition III, we have
\begin{small}
\begin{equation*}
\begin{split}
|D_{F, i}|&=p^{n-m}+\vartheta p^{\frac{n}{2}-m}\sum_{c \in \mathbb{F}_{p^m}^{*}}\eta_{m}(c)\zeta_{p}^{(F^{*})_{c^{1-d}}(0)-Tr_{1}^{m}(ci)}\\
&=p^{n-m}+\vartheta p^{\frac{n}{2}-m}\sum_{c \in \mathbb{F}_{p^m}^{*}}\eta_{m}(c)\zeta_{p}^{-Tr_{1}^{m}(ci)}\\
&=p^{n-m}+\vartheta (-1)^{m-1}\epsilon^{m}\eta_{m}(-i)p^{\frac{n-m}{2}},
\end{split}
\end{equation*}
\end{small}where in the second equation we use $F^{*}(0)=0$, and in the last equation we use Proposition 7. By Lemma 10, $F^{*}$ is also a vectorial dual-bent function with Condition III. Then with the similar computation as for $F$, we have $|D_{F^{*}, i}|=p^{n-m}+\vartheta^{-1}(-1)^{m-1}\epsilon^{m}\eta_{m}(-i)p^{\frac{n-m}{2}}$.
\end{proof}

\begin{lemma} \label{Lemma 12}
Let $F$ be a vectorial dual-bent function with Condition III. For any $a \in \mathbb{F}_{p^m}$ and $\alpha \in V_{n}^{(p)} \backslash \{0\}, \beta \in \mathbb{F}_{p^t}$, define
\begin{equation*}
N_{a, \alpha, \beta}=|\{x \in V_{n}^{(p)}: F(x)=a, \sum_{j=1}^{s}Tr_{t}^{n_{j}}(\alpha_{j}x_{j})+\beta=0\}|.
\end{equation*}
Then
\begin{equation*}
N_{a, \alpha, \beta}=\left\{
\begin{split}
& \vartheta p^{\frac{n}{2}-m}\sum_{y \in \mathbb{F}_{p^m}^{*}}\eta_{m}(y)\delta_{0}(Tr_{t}^{m}(F^{*}(\alpha)y-ay^{1-l})-\beta)\\
& \  +\vartheta (-1)^{m-1}\epsilon^{m} \eta_{m}(-a)p^{\frac{n-m}{2}-t}+p^{n-m-t},  \ \text{ if } 2 \mid \frac{m}{t},\\
&\vartheta(-1)^{t-1}\epsilon^{t}p^{\frac{n-t}{2}-m}\sum_{y \in \mathbb{F}_{p^m}^{*}}\eta_{m}(y)\eta_{t}(Tr_{t}^{m}(F^{*}(\alpha)y-ay^{1-l})-\beta)\\
& \ +\vartheta (-1)^{m-1}\epsilon^{m} \eta_{m}(-a)p^{\frac{n-m}{2}-t}+p^{n-m-t}, \ \text{ if } 2 \nmid \frac{m}{t},
\end{split}
\right.
\end{equation*}
for $a \in \mathbb{F}_{p^m}^{*}$, and
\begin{equation*}
N_{0, \alpha, \beta}=\vartheta (-1)^{m-1}\epsilon^{m}\eta_{m}(F^{*}(\alpha))p^{\frac{n-m}{2}-t}(p^t \delta_{0}(\beta)-1)+p^{n-m-t}.
\end{equation*}
\end{lemma}

\begin{proof}
The proof of Lemma 12 is given in Appendix-Section IX.
\end{proof}

\subsection{Self-orthogonal codes constructed from vectorial dual-bent functions with Condition III}
\label{sec:5.2}

In this subsection, by using vectorial dual-bent function $F$ with Condition III, we show that for some sets $I$, linear code $C_{D_{F, I}}$ defined by Eq. (1) is self-orthogonal, and the weight distribution of $C_{D_{F, I}}$ can be completely determined, which is at most six-weight.

\begin{theorem}\label{Theorem 6}
Let $F$ be a vectorial dual-bent function with Condition III, and let $C_{D_{F, I}}$ be defined by Eq. (1) with $I=\{a\}$, where $a \in \mathbb{F}_{p^m}$.

$\mathrm{(i)}$ If $a=0$, then $C_{D_{F, I}}$ is a six-weight $[p^{n-m}, \frac{n}{t}+1]_{p^t}$ self-orthogonal linear code whose weight distribution is given by Table 8.

$\mathrm{(ii)}$ If $a \in \mathbb{F}_{p^m}^{*}$ and $t=m, l=2$, then $C_{D_{F, I}}$ is a at most six-weight $[p^{n-m}+\vartheta (-1)^{m-1}\epsilon^{m}\eta_{m}(-a)p^{\frac{n-m}{2}}, $ $\frac{n}{m}+1]_{p^m}$ self-orthogonal linear code whose weight distribution is given by Table 9.
\end{theorem}

\renewcommand{\thetable}{8}
\begin{table}\label{8}\scriptsize
  \centering
  \caption{The weight distribution of $C_{D_{F, I}}$ constructed in Theorem 6 (i)}
  \begin{tabular}{|c|c|}
    \hline
    Hamming weight & Multiplicity \\ \hline
    $0$ & $1$ \\
    $p^{n-m}$ & $p^t-1$ \\
    $p^{n-m-t}(p^t-1)$ & $p^{t}(p^{n-m}-1)$ \\
    $(p^{n-m-t}-\vartheta(-1)^{m-1}\epsilon^{m}p^{\frac{n-m}{2}-t})(p^t-1)$ & $(p^{n-m}+\vartheta^{-1} (-1)^{m-1}\epsilon^{m}\eta_{m}(-1)p^{\frac{n-m}{2}})\frac{p^m-1}{2}$ \\
    $p^{n-m-t}(p^t-1)+\vartheta (-1)^{m-1}\epsilon^{m}p^{\frac{n-m}{2}-t}$ & $(p^t-1)(p^{n-m}+\vartheta^{-1} (-1)^{m-1}\epsilon^{m}\eta_{m}(-1)p^{\frac{n-m}{2}})\frac{p^m-1}{2}$ \\
    $(p^{n-m-t}+\vartheta(-1)^{m-1}\epsilon^{m}p^{\frac{n-m}{2}-t})(p^t-1)$ & $(p^{n-m}-\vartheta^{-1} (-1)^{m-1}\epsilon^{m}\eta_{m}(-1)p^{\frac{n-m}{2}})\frac{p^m-1}{2}$ \\
    $p^{n-m-t}(p^t-1)-\vartheta (-1)^{m-1}\epsilon^{m}p^{\frac{n-m}{2}-t}$ &
    $(p^t-1)(p^{n-m}-\vartheta^{-1} (-1)^{m-1}\epsilon^{m}\eta_{m}(-1)p^{\frac{n-m}{2}})\frac{p^m-1}{2}$\\
    \hline
  \end{tabular}
\end{table}

\renewcommand{\thetable}{9}
\begin{table}\label{9}\scriptsize
  \centering
  \caption{The weight distribution of $C_{D_{F, I}}$ constructed in Theorem 6 (ii)}
  \begin{tabular}{|c|c|}
    \hline
    Hamming weight & Multiplicity  \\ \hline
    $0$ & $1$ \\
    $p^{n-m}+\vartheta (-1)^{m-1}\epsilon^{m}\eta_{m}(-a)p^{\frac{n-m}{2}}$ & $p^m-1$ \\
    $p^{n-2m}(p^m-1)$ & $p^{n-m}-1$ \\
    $p^{n-2m}(p^m-1)+\vartheta (-1)^{m-1}\epsilon^{m}\eta_{m}(-a)p^{\frac{n-m}{2}}$ & $(p^m-1)(p^{n-m}-1)$ \\
    $p^{n-2m}(p^m-1)+\vartheta (-1)^{m-1}\epsilon^{m}p^{\frac{n-3m}{2}}(p^m\eta_{m}(-a)+1)$ & $(p^{n-m}+\vartheta^{-1} (-1)^{m-1}\epsilon^{m}\eta_{m}(-1)p^{\frac{n-m}{2}})\frac{(p^m-1)(p^m-1-\eta_{m}(-a))}{2}$ \\
    $p^{n-2m}(p^m-1)+\vartheta (-1)^{m-1}\epsilon^{m}p^{\frac{n-3m}{2}}(p^m\eta_{m}(-a)-1)$ & $(p^{n-m}-\vartheta^{-1} (-1)^{m-1}\epsilon^{m}\eta_{m}(-1)p^{\frac{n-m}{2}})\frac{(p^m-1)(p^m-1+\eta_{m}(-a))}{2}$ \\
    $p^{n-2m}(p^m-1)+\vartheta (-1)^{m-1}\epsilon^{m}\eta_{m}(-a)p^{\frac{n-3m}{2}}$ &
    $(p^m-1)(p^{n-m}+\vartheta^{-1} (-1)^{m-1}\epsilon^{m}\eta_{m}(a)p^{\frac{n-m}{2}})$\\
    \hline
  \end{tabular}
\end{table}

\begin{proof}
Denote $\mathcal{S}=\{x^{2}: x \in \mathbb{F}_{p^m}^{*}\}, \mathcal{N}=\mathbb{F}_{p^m}^{*} \backslash \mathcal{S}$.

(i) By Lemma 11, then length of $C_{D_{F, 0}}$ is $|D_{F, 0}|=p^{n-m}$. When $\alpha=0, \beta=0$, $wt(c_{\alpha, \beta})=0$. When $\alpha=0, \beta \in \mathbb{F}_{p^t}^{*}$, $wt(c_{\alpha, \beta})=|D_{F, 0}|$. When $\alpha \in V_{n}^{(p)} \backslash \{0\}, \beta \in \mathbb{F}_{p^t}$, $wt(c_{\alpha, \beta})=|D_{F, 0}|-N_{0, \alpha, \beta}$, where $N_{0, \alpha, \beta}=|\{x \in V_{n}^{(p)}: F(x)=0, \sum_{j=1}^{s}Tr_{t}^{n_{j}}(\alpha_{j}x_{j})+\beta=0\}|$.

By Lemma 12,
\begin{itemize}
  \item when $\alpha\neq 0, F^{*}(\alpha)=0, \beta \in \mathbb{F}_{p^t}$, $N_{0, \alpha, \beta}=p^{n-m-t}$ and $wt(c_{\alpha, \beta})=p^{n-m-t}(p^t-1)$;
  \item when $\alpha \neq 0, F^{*}(\alpha)\in \mathcal{S}, \beta=0$, $N_{0, \alpha, \beta}=\vartheta (-1)^{m-1}\epsilon^{m}p^{\frac{n-m}{2}-t}(p^t-1)+p^{n-m-t}$ and $wt(c_{\alpha, \beta})=(p^{n-m-t}-\vartheta(-1)^{m-1}\epsilon^{m}p^{\frac{n-m}{2}-t})(p^t-1)$;
  \item when $\alpha \neq 0, F^{*}(\alpha)\in \mathcal{S}, \beta\neq 0$, $N_{0, \alpha, \beta}=-\vartheta (-1)^{m-1}\epsilon^{m}p^{\frac{n-m}{2}-t}+p^{n-m-t}$ and $wt(c_{\alpha, \beta})=p^{n-m-t}(p^t-1)+\vartheta (-1)^{m-1}\epsilon^{m}p^{\frac{n-m}{2}-t}$;
  \item when $\alpha\neq 0, F^{*}(\alpha)\in \mathcal{N}, \beta=0$, $N_{0, \alpha, \beta}=-\vartheta (-1)^{m-1}\epsilon^{m}p^{\frac{n-m}{2}-t}(p^t-1)+p^{n-m-t}$ and $wt(c_{\alpha, \beta})=(p^{n-m-t}+\vartheta(-1)^{m-1}\epsilon^{m}p^{\frac{n-m}{2}-t})(p^t-1)$;
  \item when $\alpha \neq 0, F^{*}(\alpha)\in \mathcal{N}, \beta\neq 0$, $N_{0, \alpha, \beta}=\vartheta (-1)^{m-1}\epsilon^{m}p^{\frac{n-m}{2}-t}+p^{n-m-t}$ and $wt(c_{\alpha, \beta})=p^{n-m-t}(p^t-1)-\vartheta (-1)^{m-1}\epsilon^{m}p^{\frac{n-m}{2}-t}$.
\end{itemize}
We can see that $wt(c_{\alpha, \beta})=0$ if and only if $\alpha=0, \beta=0$, thus the dimension of $C_{D_{F, 0}}$ is $\frac{n}{t}+1$. The weight distribution of $C_{D_{F, 0}}$ follows from the above arguments and Lemma 11.

When $\alpha=0, \beta \in \mathbb{F}_{p^t}$, $c_{\alpha, \beta} \cdot c_{\alpha, \beta}=\beta^{2}|D_{F, 0}|=0$ since $p \mid |D_{F, 0}|$. When $\alpha\neq 0$, by Lemma 12, for any $i \in \mathbb{F}_{p^t}^{*}$, the value $N_{0, \alpha, -i}$ is independent of $i$, and $p \mid N_{0, \alpha, -i}$ when $p^t=3$ as $(p^t, n)\neq (3, 3)$. Then for any $\alpha \in V_{n}^{(p)} \backslash \{0\}, \beta \in \mathbb{F}_{p^t}$, $c_{\alpha, \beta} \cdot c_{\alpha, \beta}=\sum_{i \in \mathbb{F}_{p^t}^{*}}N_{0, \alpha, -i}i^{2}+2\beta \sum_{i \in \mathbb{F}_{p^t}^{*}}N_{0, \alpha, -i}i+\beta^{2}|D_{F, 0}|=0$. By Proposition 3, $C_{D_{F, 0}}$ is self-orthogonal.

(ii) By Lemma 11, the length of $C_{D_{F, a}}$ is $|D_{F, a}|=p^{n-m}+\vartheta (-1)^{m-1}\epsilon^{m}\eta_{m}(-a)p^{\frac{n-m}{2}}$. When $\alpha=0, \beta=0$, $wt(c_{\alpha, \beta})=0$. When $\alpha=0, \beta \in \mathbb{F}_{p^m}^{*}$, $wt(c_{\alpha, \beta})=|D_{F, a}|$. When $\alpha \in V_{n}^{(p)} \backslash \{0\}, \beta \in \mathbb{F}_{p^m}$, $wt(c_{\alpha, \beta})=|D_{F, a}|-N_{a, \alpha, \beta}$, where $N_{a, \alpha, \beta}=|\{x\in V_{n}^{(p)}: F(x)=a, \sum_{j=1}^{s}Tr_{m}^{n_{j}}(\alpha_{j}x_{j})+\beta=0\}|$. By Lemma 12 and Proposition 8, for any $\alpha \in V_{n}^{(p)} \backslash \{0\}, \beta \in \mathbb{F}_{p^m}$,
\begin{small}
\begin{equation*}
\begin{split}
&N_{a, \alpha, \beta}\\
&=\vartheta (-1)^{m-1}\epsilon^{m}p^{\frac{n-3m}{2}}\sum_{y \in \mathbb{F}_{p^m}^{*}}\eta_{m}(F^{*}(\alpha)y^{2}-\beta y-a)+\vartheta (-1)^{m-1}\epsilon^{m}\eta_{m}(-a)p^{\frac{n-3m}{2}}+p^{n-2m}\\
&=\vartheta (-1)^{m-1}\epsilon^{m}p^{\frac{n-3m}{2}}\sum_{y \in \mathbb{F}_{p^m}}\eta_{m}(F^{*}(\alpha)y^{2}-\beta y-a)+p^{n-2m}\\
&=\left\{
\begin{split}
\vartheta (-1)^{m-1}\epsilon^{m}\eta_{m}(-a)p^{\frac{n-m}{2}}+p^{n-2m}, & \ \text{ if } F^{*}(\alpha)=0, \beta=0, \\
p^{n-2m}, & \ \text{ if } F^{*}(\alpha)=0, \beta\neq 0 , \\
-\vartheta (-1)^{m-1}\epsilon^{m}\eta_{m}(F^{*}(\alpha))p^{\frac{n-3m}{2}}+p^{n-2m}, & \ \text{ if } F^{*}(\alpha)\neq 0, \beta^{2}+4aF^{*}(\alpha)\neq 0, \\
\vartheta (-1)^{m-1}\epsilon^{m}\eta_{m}(F^{*}(\alpha))p^{\frac{n-3m}{2}}(p^m-1)+p^{n-2m}, & \ \text{ if } F^{*}(\alpha)\neq 0, \beta^{2}+4aF^{*}(\alpha)=0,
\end{split}
\right. \\
\end{split}
\end{equation*}
\begin{equation*}
\begin{split}
&=\left\{
\begin{split}
\vartheta (-1)^{m-1}\epsilon^{m}\eta_{m}(-a)p^{\frac{n-m}{2}}+p^{n-2m}, & \ \text{ if } F^{*}(\alpha)=0, \beta=0, \\
p^{n-2m}, & \ \text{ if } F^{*}(\alpha)=0, \beta\neq 0 , \\
-\vartheta (-1)^{m-1}\epsilon^{m}p^{\frac{n-3m}{2}}+p^{n-2m}, & \ \text{ if } F^{*}(\alpha) \in \mathcal{S}, \beta=0, \\
& \ \text{ or } F^{*}(\alpha) \in \mathcal{S}, \beta \neq 0, \beta^{2}+4aF^{*}(\alpha)\neq 0,\\
\vartheta (-1)^{m-1}\epsilon^{m}p^{\frac{n-3m}{2}}+p^{n-2m}, & \ \text{ if } F^{*}(\alpha) \in \mathcal{N}, \beta=0,\\
& \ \text{ or } F^{*}(\alpha) \in \mathcal{N}, \beta \neq 0, \beta^{2}+4aF^{*}(\alpha)\neq 0,\\
\vartheta (-1)^{m-1}\epsilon^{m}\eta_{m}(-a)p^{\frac{n-3m}{2}}(p^m-1)+p^{n-2m}, & \ \text{ if }\beta\neq 0, \beta^{2}+4aF^{*}(\alpha)=0,
\end{split}
\right.
\end{split}
\end{equation*}
\end{small} and then
\begin{small}
\begin{equation*}
\begin{split}
wt(c_{\alpha, \beta})&=|D_{F, a}|-N_{a, \alpha, \beta}\\
&=\left\{
\begin{split}
p^{n-2m}(p^m-1), & \ \text{ if } F^{*}(\alpha)=0, \beta=0, \\
p^{n-2m}(p^m-1)+\vartheta (-1)^{m-1}\epsilon^{m}\eta_{m}(-a)p^{\frac{n-m}{2}}, & \ \text{ if } F^{*}(\alpha)=0, \beta\neq 0 , \\
p^{n-2m}(p^m-1)+\vartheta (-1)^{m-1}\epsilon^{m}p^{\frac{n-3m}{2}}(p^m\eta_{m}(-a)+1), & \ \text{ if } F^{*}(\alpha) \in \mathcal{S}, \beta=0, \\
&  \ \text{ or } F^{*}(\alpha) \in \mathcal{S}, \beta \neq 0, \beta^{2}+4aF^{*}(\alpha)\neq 0,\\
p^{n-2m}(p^m-1)+\vartheta (-1)^{m-1}\epsilon^{m}p^{\frac{n-3m}{2}}(p^m\eta_{m}(-a)-1), & \ \text{ if } F^{*}(\alpha) \in \mathcal{N}, \beta=0, \\
&  \ \text{ or } F^{*}(\alpha) \in \mathcal{N}, \beta \neq 0, \beta^{2}+4aF^{*}(\alpha)\neq 0,\\
p^{n-2m}(p^m-1)+\vartheta (-1)^{m-1}\epsilon^{m}\eta_{m}(-a)p^{\frac{n-3m}{2}}, & \ \text{ if }\beta\neq 0, \beta^{2}+4aF^{*}(\alpha)=0.
\end{split}
\right.
\end{split}
\end{equation*}
\end{small}We can see that $wt(c_{\alpha, \beta})=0$ if and only if $\alpha=0, \beta=0$, thus the dimension of $C_{D_{F, a}}$ is $\frac{n}{m}+1$. The weight distribution of $C_{D_{F, a}}$ follows from the above arguments and Lemma 11.

When $\alpha=0$, $c_{\alpha, \beta} \cdot c_{\alpha, \beta}=\beta^{2}|D_{F, a}|=0$ by $p \mid |D_{F, a}|$. When $\alpha\neq 0, F^{*}(\alpha)=0$, $c_{\alpha, \beta} \cdot c_{\alpha, \beta}=\sum_{i \in \mathbb{F}_{p^m}^{*}}N_{a, \alpha, -i}i^{2}+2\beta \sum_{i \in \mathbb{F}_{p^m}^{*}}N_{a, \alpha, -i}i+\beta^{2}|D_{F, a}|=0$ by $p \mid N_{a, \alpha, -i}, i \in \mathbb{F}_{p^m}^{*}$. Note that $\sum_{i \in \mathbb{F}_{p^m}^{*}}i=0$, $\sum_{i \in \mathbb{F}_{p^m}^{*}}i^{2}=0$ if $p^m>3$, and $p \mid p^{\frac{n-3m}{2}}$ if $p^m=3$ since $(p^m, n)\neq (3, 3)$. When $\alpha\neq 0, F^{*}(\alpha)\neq 0, \eta_{m}(-aF^{*}(\alpha))=1$, then there are two elements $\pm j \in \mathbb{F}_{p^m}^{*}$ such that $(\pm j)^{2}+4aF^{*}(\alpha)=0$, and
\begin{small}
\begin{equation*}
\begin{split}
 \sum_{i \in \mathbb{F}_{p^m}^{*}}N_{a, \alpha, -i}i^{2}&=\sum_{i \in \mathbb{F}_{p^m}^{*} \backslash \{\pm j\}}(-\vartheta (-1)^{m-1}\epsilon^{m}\eta_{m}(F^{*}(\alpha))p^{\frac{n-3m}{2}}+p^{n-2m})i^{2}\\
 & \ \ \ +\sum_{i \in \{\pm j\}}(\vartheta (-1)^{m-1}\epsilon^{m}\eta_{m}(F^{*}(\alpha))p^{\frac{n-3m}{2}}(p^m-1)+p^{n-2m})i^{2}\\
 &=(-\vartheta (-1)^{m-1}\epsilon^{m}\eta_{m}(F^{*}(\alpha))p^{\frac{n-3m}{2}}+p^{n-2m})\sum_{i \in \mathbb{F}_{p^m}^{*}}i^{2}+2 \vartheta (-1)^{m-1}\epsilon^{m}\eta_{m}(F^{*}(\alpha))p^{\frac{n-m}{2}}j^{2}\\
 &=0,\\
 \sum_{i \in \mathbb{F}_{p^m}^{*}}N_{a, \alpha, -i}i&=\sum_{i \in \mathbb{F}_{p^m}^{*} \backslash \{\pm j\}}(-\vartheta (-1)^{m-1}\epsilon^{m}\eta_{m}(F^{*}(\alpha))p^{\frac{n-3m}{2}}+p^{n-2m})i\\
 & \ \ \ +\sum_{i \in \{\pm j\}}(\vartheta (-1)^{m-1}\epsilon^{m}\eta_{m}(F^{*}(\alpha))p^{\frac{n-3m}{2}}(p^m-1)+p^{n-2m})i\\
 &=(-\vartheta (-1)^{m-1}\epsilon^{m}\eta_{m}(F^{*}(\alpha))p^{\frac{n-3m}{2}}+p^{n-2m})\sum_{i \in \mathbb{F}_{p^m}^{*}}i\\
 &=0.
\end{split}
\end{equation*}
\end{small}Thus, in this case, $c_{\alpha, \beta} \cdot c_{\alpha, \beta}=0$. When $\alpha\neq 0, F^{*}(\alpha)\neq 0, \eta_{m}(-aF^{*}(\alpha))=-1$, then there is no $j \in \mathbb{F}_{p^m}^{*}$ such that $j^{2}+4aF^{*}(\alpha)=0$, and
\begin{small}
\begin{equation*}
\begin{split}
\sum_{i \in \mathbb{F}_{p^m}^{*}}N_{a, \alpha, -i}i^{2}=(-\vartheta (-1)^{m-1}\epsilon^{m}\eta_{m}(F^{*}(\alpha))p^{\frac{n-3m}{2}}+p^{n-2m})\sum_{i \in \mathbb{F}_{p^m}^{*}}i^{2}=0, \\
\sum_{i \in \mathbb{F}_{p^m}^{*}}N_{a, \alpha, -i}i=(-\vartheta (-1)^{m-1}\epsilon^{m}\eta_{m}(F^{*}(\alpha))p^{\frac{n-3m}{2}}+p^{n-2m})\sum_{i \in \mathbb{F}_{p^m}^{*}}i=0.
\end{split}
\end{equation*}
\end{small}Thus, $c_{\alpha, \beta} \cdot c_{\alpha, \beta}=0$. By Proposition 3, $C_{D_{F, a}}$ is self-orthogonal.
\end{proof}

\begin{theorem}\label{Theorem 7}
Let $F$ be a vectorial dual-bent function with Condition III, and let $C_{D_{F, I}}$ be defined by Eq. (1) with $I=\mathcal{S}$ or $I=\mathcal{N}$, where $\mathcal{S}=\{x^{2}: x \in \mathbb{F}_{p^m}^{*}\}$, $\mathcal{N}=\mathbb{F}_{p^m}^{*} \backslash \mathcal{S}$.

(i) When $I=\mathcal{S}$, $C_{D_{F, I}}$ is a at most six-weight $[(p^{n-m}+\vartheta (-1)^{m-1}\epsilon^{m}\eta_{m}(-1)p^{\frac{n-m}{2}})\frac{p^m-1}{2}$, $\frac{n}{t}+1]_{p^t}$ self-orthogonal linear code whose weight distribution is given in Table 10.

(ii) When $I=\mathcal{N}$, $C_{D_{F, I}}$ is a at most six-weight $[(p^{n-m}-\vartheta (-1)^{m-1}\epsilon^{m}\eta_{m}(-1)p^{\frac{n-m}{2}})\frac{p^m-1}{2}$, $\frac{n}{t}+1]_{p^t}$ self-orthogonal linear code whose weight distribution is given in Table 11.

Besides, the dual code $C_{D_{F, I}}^{\bot}$ is at least almost optimal according to Hamming bound.
\end{theorem}

\renewcommand{\thetable}{10}
\begin{table}\label{10}\scriptsize
  \centering
  \caption{The weight distribution of $C_{D_{F, I}}$ constructed in Theorem 7 (i)}
  \begin{tabular}{|c|c|}
    \hline
    Hamming weight & Multiplicity  \\ \hline
    $0$ & $1$ \\
    $(p^{n-m}+\vartheta (-1)^{m-1}\epsilon^{m}\eta_{m}(-1)p^{\frac{n-m}{2}})\frac{p^m-1}{2}$ & $p^t-1$ \\
    $(p^{n-m}-p^{n-m-t})\frac{p^m-1}{2}$ & $p^{n-m}-1$ \\
    $(p^{n-m}-p^{n-m-t}+\vartheta (-1)^{m-1}\epsilon^{m}\eta_{m}(-1)p^{\frac{n-m}{2}})\frac{p^m-1}{2}$ & $(p^t-1)(p^{n-m}-1)$ \\
    $(p^{n-m}-p^{n-m-t})\frac{p^m-1}{2}+\vartheta (-1)^{m-1}\epsilon^{m}\eta_{m}(-1)p^{\frac{n-m}{2}-t}\frac{(p^m+1)(p^t-1)}{2}$ & $(p^{n-m}+\vartheta^{-1}(-1)^{m-1}\epsilon^{m}p^{\frac{n-m}{2}})\frac{p^m-1}{2}$ \\
    $(p^{n-m}-p^{n-m-t})\frac{p^m-1}{2}+\vartheta (-1)^{m-1}\epsilon^{m}\eta_{m}(-1)p^{\frac{n-m}{2}-t}\frac{p^t(p^m-1)-(p^m+1)}{2}$ & $(p^t-1)(p^{n-m}+\vartheta^{-1}(-1)^{m-1}\epsilon^{m}p^{\frac{n-m}{2}})\frac{p^m-1}{2}$ \\
    $(p^{n-m}-p^{n-m-t}+\vartheta (-1)^{m-1}\epsilon^{m}\eta_{m}(-1)p^{\frac{n-m}{2}-t}(p^t-1))\frac{p^m-1}{2}$ &
    $p^t(p^{n-m}-\vartheta^{-1}(-1)^{m-1}\epsilon^{m}p^{\frac{n-m}{2}})\frac{p^m-1}{2}$\\
    \hline
  \end{tabular}
\end{table}

\renewcommand{\thetable}{11}
\begin{table}\label{11}\scriptsize
  \centering
  \caption{The weight distribution of $C_{D_{F, I}}$ constructed in Theorem 7 (ii)}
  \begin{tabular}{|c|c|}
    \hline
    Hamming weight  & Multiplicity  \\ \hline
    $0$ & $1$ \\
    $(p^{n-m}-\vartheta (-1)^{m-1}\epsilon^{m}\eta_{m}(-1)p^{\frac{n-m}{2}})\frac{p^m-1}{2}$ & $p^t-1$ \\
    $(p^{n-m}-p^{n-m-t})\frac{p^m-1}{2}$ & $p^{n-m}-1$ \\
    $(p^{n-m}-p^{n-m-t}-\vartheta (-1)^{m-1}\epsilon^{m}\eta_{m}(-1)p^{\frac{n-m}{2}})\frac{p^m-1}{2}$ & $(p^t-1)(p^{n-m}-1)$ \\
    $(p^{n-m}-p^{n-m-t})\frac{p^m-1}{2}-\vartheta (-1)^{m-1}\epsilon^{m}\eta_{m}(-1)p^{\frac{n-m}{2}-t}\frac{(p^m+1)(p^t-1)}{2}$ & $(p^{n-m}-\vartheta^{-1}(-1)^{m-1}\epsilon^{m}p^{\frac{n-m}{2}})\frac{p^m-1}{2}$ \\
    $(p^{n-m}-p^{n-m-t})\frac{p^m-1}{2}-\vartheta (-1)^{m-1}\epsilon^{m}\eta_{m}(-1)p^{\frac{n-m}{2}-t}\frac{p^t(p^m-1)-(p^m+1)}{2}$ & $(p^t-1)(p^{n-m}-\vartheta^{-1}(-1)^{m-1}\epsilon^{m}p^{\frac{n-m}{2}})\frac{p^m-1}{2}$ \\
    $(p^{n-m}-p^{n-m-t}-\vartheta (-1)^{m-1}\epsilon^{m}\eta_{m}(-1)p^{\frac{n-m}{2}-t}(p^t-1))\frac{p^m-1}{2}$ &
    $p^t(p^{n-m}+\vartheta^{-1}(-1)^{m-1}\epsilon^{m}p^{\frac{n-m}{2}})\frac{p^m-1}{2}$\\
    \hline
  \end{tabular}
\end{table}

\begin{proof}
We only prove the case of $I=\mathcal{S}$ since the other case is similar.

By Lemma 11, the length of $C_{D_{F, \mathcal{S}}}$ is $|D_{F, \mathcal{S}}|=(p^{n-m}+\vartheta (-1)^{m-1}\epsilon^{m}\eta_{m}(-1)p^{\frac{n-m}{2}})\frac{p^m-1}{2}$. When $\alpha=0, \beta=0$, $wt(c_{\alpha, \beta})=0$. When $\alpha=0, \beta \in \mathbb{F}_{p^t}^{*}$, $wt(c_{\alpha, \beta})=|D_{F, \mathcal{S}}|$. When $\alpha\neq 0,  \beta \in \mathbb{F}_{p^t}$, $wt(c_{\alpha, \beta})=|D_{F, \mathcal{S}}|-N_{\mathcal{S}, \alpha, \beta}$, where $N_{\mathcal{S}, \alpha, \beta}=|\{x \in V_{n}^{(p)}: F(x) \in \mathcal{S}, \sum_{j=1}^{s}Tr_{t}^{n_{j}}(\alpha_{j}x_{j})+\beta=0\}|$.

When $\frac{m}{t}$ is even, by Lemma 12,
\begin{small}
\begin{equation*}
\begin{split}
N_{\mathcal{S}, \alpha, \beta}&=\vartheta p^{\frac{n}{2}-m}\sum_{a \in \mathcal{S}}\sum_{y \in \mathbb{F}_{p^m}^{*}}\eta_{m}(y)\delta_{0}(Tr_{t}^{m}(F^{*}(\alpha)y-ay^{1-l})-\beta)+p^{n-m-t}\frac{p^m-1}{2}\\
& \ \ +\vartheta (-1)^{m-1}\epsilon^{m}\eta_{m}(-1)p^{\frac{n-m}{2}-t}\frac{p^m-1}{2}.
\end{split}
\end{equation*}
\end{small}Denote $T=\sum_{a \in \mathcal{S}}\sum_{y \in \mathbb{F}_{p^m}^{*}}\eta_{m}(y)\delta_{0}(Tr_{t}^{m}(F^{*}(\alpha)y-ay^{1-l})-\beta)$. Then
\begin{small}
\begin{equation*}
\begin{split}
T&=p^{-t}\sum_{a \in \mathcal{S}}\sum_{y \in \mathbb{F}_{p^m}^{*}}\eta_{m}(y)\sum_{z \in \mathbb{F}_{p^t}}\zeta_{p}^{Tr_{1}^{t}(z(Tr_{t}^{m}(F^{*}(\alpha)y-ay^{1-l})-\beta))}\\
&=p^{-t}\sum_{y \in \mathbb{F}_{p^m}^{*}}\eta_{m}(y)\sum_{z \in \mathbb{F}_{p^t}^{*}}\zeta_{p}^{Tr_{1}^{t}(-z\beta)}\sum_{a \in \mathcal{S}}\zeta_{p}^{Tr_{1}^{m}(F^{*}(\alpha)y z-ay^{1-l}z)}\\
&=p^{-t}\sum_{y \in \mathbb{F}_{p^m}^{*}}\eta_{m}(y)\sum_{z \in \mathbb{F}_{p^t}^{*}}\zeta_{p}^{Tr_{1}^{t}(-z\beta)}\sum_{a \in \mathcal{S}}\zeta_{p}^{Tr_{1}^{m}(F^{*}(\alpha)y z-ayz)}\\
\end{split}
\end{equation*}
\begin{equation*}
\begin{split}
&=p^{-t}\sum_{a \in \mathcal{S}}\sum_{y \in \mathbb{F}_{p^m}^{*}, z \in \mathbb{F}_{p^t}^{*}}\eta_{m}(yz^{-1})\zeta_{p}^{Tr_{1}^{t}(-z\beta)+Tr_{1}^{m}(F^{*}(\alpha)y-ay)}\\
&=p^{-t}\sum_{a \in \mathcal{S}}\sum_{z \in \mathbb{F}_{p^t}^{*}}\zeta_{p}^{Tr_{1}^{t}(-z\beta)}\sum_{y \in \mathbb{F}_{p^m}^{*}}\eta_{m}(y)\zeta_{p}^{Tr_{1}^{m}((F^{*}(\alpha)-a)y)}\\
&=(-1)^{m-1}\epsilon^{m}p^{\frac{m}{2}-t}(p^t \delta_{0}(\beta)-1)\sum_{a \in \mathcal{S}}\eta_{m}(F^{*}(\alpha)-a)\\
&=\left\{
\begin{split}
(-1)^{m-1}\epsilon^{m}\eta_{m}(-1)p^{\frac{m}{2}-t}(p^t \delta_{0}(\beta)-1)\frac{p^m-1}{2}, & \ \text{ if } F^{*}(\alpha)=0, \\
-(-1)^{m-1}\epsilon^{m}\eta_{m}(-1)p^{\frac{m}{2}-t}(p^t \delta_{0}(\beta)-1), & \ \text{ if } -F^{*}(\alpha) \in \mathcal{S},\\
0, & \ \text{ if } -F^{*}(\alpha) \in \mathcal{N},
\end{split}
\right.
\end{split}
\end{equation*}
\end{small}where in the third equation we use that for any $y \in \mathbb{F}_{p^m}^{*}$, $y^{-l} \in \mathcal{S}$ since $l$ is even, and in the last equation we use
\begin{small}
\begin{equation}\label{12}
\begin{split}
&\sum_{a \in S}\eta_{m}(F^{*}(\alpha)-a)\\
&=\left\{
\begin{split}
\eta_{m}(-1)\frac{p^m-1}{2}, & \ \text{ if } F^{*}(\alpha)=0, \\
\eta_{m}(-1)(|(1+\mathcal{S}) \cap \mathcal{S}|-|(1+\mathcal{S}) \cap \mathcal{N}|), & \ \text{ if } -F^{*}(\alpha) \in \mathcal{S},\\
-\eta_{m}(-1)(|(1+\mathcal{N}) \cap \mathcal{S}|-|(1+\mathcal{N}) \cap \mathcal{N}|), & \ \text{ if } -F^{*}(\alpha) \in \mathcal{N}.\\
\end{split}
\right.\\
&=\left\{
\begin{split}
\eta_{m}(-1)\frac{p^m-1}{2}, & \ \text{ if } F^{*}(\alpha)=0, \\
-\eta_{m}(-1), & \ \text{ if } -F^{*}(\alpha) \in \mathcal{S},\\
0, & \ \text{ if } -F^{*}(\alpha) \in \mathcal{N},
\end{split}
\right.
\end{split}
\end{equation}
\end{small}
which is obtained by Proposition 10.

By the above arguments,
\begin{itemize}
  \item when $\alpha \neq 0, F^{*}(\alpha)=0, \beta=0$, $N_{\mathcal{S}, \alpha, \beta}=(\vartheta (-1)^{m-1}\epsilon^{m}\eta_{m}(-1)p^{\frac{n-m}{2}}+p^{n-m-t})\frac{p^m-1}{2}$ and then $wt(c_{\alpha, \beta})=(p^{n-m}-p^{n-m-t})\frac{p^m-1}{2}$;
  \item when $\alpha \neq 0, F^{*}(\alpha)=0, \beta\neq 0$, $N_{\mathcal{S}, \alpha, \beta}=p^{n-m-t}\frac{p^m-1}{2}$ and then $wt(c_{\alpha, \beta})=(p^{n-m}-p^{n-m-t}+\vartheta (-1)^{m-1}\epsilon^{m}\eta_{m}(-1)p^{\frac{n-m}{2}})\frac{p^m-1}{2}$;
  \item when $\alpha\neq 0, -F^{*}(\alpha) \in \mathcal{S}, \beta=0$, $N_{\mathcal{S}, \alpha,\beta}=p^{n-m-t}\frac{p^m-1}{2}+\vartheta (-1)^{m-1}\epsilon^{m}\eta_{m}(-1)p^{\frac{n-m}{2}-t}(\frac{p^m-1}{2}-p^t+1)$ and then $wt(c_{\alpha, \beta})=(p^{n-m}-p^{n-m-t})\frac{p^m-1}{2}+\vartheta (-1)^{m-1}\epsilon^{m}\eta_{m}(-1)p^{\frac{n-m}{2}-t}\frac{(p^m+1)(p^t-1)}{2}$;
  \item when $\alpha\neq 0, -F^{*}(\alpha) \in \mathcal{S}, \beta\neq 0$, $N_{\mathcal{S}, \alpha, \beta}=p^{n-m-t}\frac{p^m-1}{2}+\vartheta (-1)^{m-1}\epsilon^{m}\eta_{m}(-1)p^{\frac{n-m}{2}-t}\frac{p^m+1}{2}$ and then $wt(c_{\alpha,\beta})=(p^{n-m}-p^{n-m-t})\frac{p^m-1}{2}+\vartheta (-1)^{m-1}\epsilon^{m}\eta_{m}(-1)p^{\frac{n-m}{2}-t}\frac{p^t(p^m-1)-(p^m+1)}{2}$;
  \item when $\alpha\neq 0, -F^{*}(\alpha) \in \mathcal{N}, \beta \in \mathbb{F}_{p^t}$, $N_{\mathcal{S}, \alpha, \beta}=(p^{n-m-t}+\vartheta (-1)^{m-1}\epsilon^{m}\eta_{m}(-1)p^{\frac{n-m}{2}-t})\frac{p^m-1}{2}$ and then $wt(c_{\alpha, \beta})=(p^{n-m}-p^{n-m-t}+\vartheta (-1)^{m-1}\epsilon^{m}\eta_{m}(-1)p^{\frac{n-m}{2}-t}(p^t-1))\frac{p^m-1}{2}$.
\end{itemize}

When $\frac{m}{t}$ is odd, by Lemma 12,
\begin{small}
\begin{equation*}
\begin{split}
N_{\mathcal{S}, \alpha, \beta}&=\vartheta (-1)^{t-1}\epsilon^{t}p^{\frac{n-t}{2}-m}\sum_{a \in \mathcal{S}}\sum_{y \in \mathbb{F}_{p^m}^{*}}\eta_{m}(y)\eta_{t}(Tr_{t}^{m}(F^{*}(\alpha)y-ay^{1-l})-\beta)+p^{n-m-t}\frac{p^m-1}{2}\\
& \ \ +\vartheta (-1)^{m-1}\epsilon^{m}\eta_{m}(-1)p^{\frac{n-m}{2}-t}\frac{p^m-1}{2}.
\end{split}
\end{equation*}
\end{small}Denote $R=\sum_{a \in \mathcal{S}}\sum_{y \in \mathbb{F}_{p^m}^{*}}\eta_{m}(y)\eta_{t}(Tr_{t}^{m}(F^{*}(\alpha)y-ay^{1-l})-\beta)$, and let $\mathcal{S}'=\{x^{2}: x \in \mathbb{F}_{p^t}^{*}\}$, $\mathcal{N}'=\mathbb{F}_{p^t}^{*} \backslash \mathcal{S}'$. Then
\begin{small}
\begin{equation*}
\begin{split}
R&=p^{-t}\sum_{a \in \mathcal{S}}\sum_{y \in \mathbb{F}_{p^m}^{*}}\eta_{m}(y)\sum_{w \in \mathcal{S}'}\sum_{z \in \mathbb{F}_{p^t}}\zeta_{p}^{Tr_{1}^{t}(z(Tr_{t}^{m}(F^{*}(\alpha)y-ay^{1-l})-\beta-w))}\\
& \ \ -p^{-t}\sum_{a \in \mathcal{S}}\sum_{y \in \mathbb{F}_{p^m}^{*}}\eta_{m}(y)\sum_{w \in \mathcal{N}'}\sum_{z \in \mathbb{F}_{p^t}}\zeta_{p}^{Tr_{1}^{t}(z(Tr_{t}^{m}(F^{*}(\alpha)y-ay^{1-l})-\beta-w))}\\
&=2p^{-t}\sum_{a \in \mathcal{S}}\sum_{y \in \mathbb{F}_{p^m}^{*}}\eta_{m}(y)\sum_{w \in \mathcal{S}'}\sum_{z \in \mathbb{F}_{p^t}^{*}}\zeta_{p}^{Tr_{1}^{t}(z(Tr_{t}^{m}(F^{*}(\alpha)y-ay^{1-l})-\beta-w))}\\
& \ \ -p^{-t}\sum_{a \in \mathcal{S}}\sum_{y \in \mathbb{F}_{p^m}^{*}}\eta_{m}(y)\sum_{w \in \mathbb{F}_{p^t}^{*}}\sum_{z \in \mathbb{F}_{p^t}^{*}}\zeta_{p}^{Tr_{1}^{t}(z(Tr_{t}^{m}(F^{*}(\alpha)y-ay^{1-l})-\beta-w))}\\
&=2p^{-t}\sum_{y \in \mathbb{F}_{p^m}^{*}, w \in \mathcal{S}', z \in \mathbb{F}_{p^t}^{*}}\eta_{m}(y)\zeta_{p}^{Tr_{1}^{t}(-zw-z\beta)}\sum_{a \in \mathcal{S}}\zeta_{p}^{Tr_{1}^{m}(F^{*}(\alpha)yz-ayz)}\\
& \ \ -p^{-t}\sum_{y \in \mathbb{F}_{p^m}^{*}, w, z \in \mathbb{F}_{p^t}^{*}}\eta_{m}(y)\zeta_{p}^{Tr_{1}^{t}(-zw-z\beta)}\sum_{a \in \mathcal{S}}\zeta_{p}^{Tr_{1}^{m}(F^{*}(\alpha)yz-ayz)}\\
&=2p^{-t}\sum_{a \in \mathcal{S}}\sum_{y \in \mathbb{F}_{p^m}^{*}, w \in \mathcal{S}', z \in \mathbb{F}_{p^t}^{*}}\eta_{m}(yz^{-1})\zeta_{p}^{Tr_{1}^{m}((F^{*}(\alpha)-a)y)+Tr_{1}^{t}(-z(w+\beta))}\\
& \ \ -p^{-t}\sum_{a \in \mathcal{S}}\sum_{y \in \mathbb{F}_{p^m}^{*}, w, z \in \mathbb{F}_{p^t}^{*}}\eta_{m}(yz^{-1})\zeta_{p}^{Tr_{1}^{m}((F^{*}(\alpha)-a)y)+Tr_{1}^{t}(-z(w+\beta))}\\
&=2p^{-t}\sum_{w \in \mathcal{S}'}\sum_{z \in \mathbb{F}_{p^t}^{*}}\eta_{t}(z)\zeta_{p}^{Tr_{1}^{t}(-z(w+\beta))}\sum_{a \in \mathcal{S}}\sum_{y \in \mathbb{F}_{p^m}^{*}}\eta_{m}(y)\zeta_{p}^{Tr_{1}^{m}((F^{*}(\alpha)-a)y)}\\
& \ \ -p^{-t}\sum_{z \in \mathbb{F}_{p^t}^{*}}\eta_{t}(z)\zeta_{p}^{Tr_{1}^{t}(-z\beta)}\sum_{w \in \mathbb{F}_{p^t}^{*}}\zeta_{p}^{Tr_{1}^{t}(-zw)}\sum_{a \in \mathcal{S}}\sum_{y \in \mathbb{F}_{p^m}^{*}}\eta_{m}(y)\zeta_{p}^{Tr_{1}^{m}((F^{*}(\alpha)-a)y)}\\
&=2p^{-t}\sum_{w \in \mathcal{S}'}(-1)^{t-1}\epsilon^{t}\eta_{t}(-(w+\beta))p^{\frac{t}{2}}\sum_{a \in \mathcal{S}}(-1)^{m-1}\epsilon^{m}\eta_{m}(F^{*}(\alpha)-a)p^{\frac{m}{2}}\\
& \ \ +p^{-t}(-1)^{t-1}\epsilon^{t}\eta_{t}(-\beta)p^{\frac{t}{2}}\sum_{a \in \mathcal{S}}(-1)^{m-1}\epsilon^{m}\eta_{m}(F^{*}(\alpha)-a)p^{\frac{m}{2}}\\
&=(-1)^{m+t}\epsilon^{m+t}\eta_{m}(-1)p^{\frac{m-t}{2}}(2\sum_{w \in \mathcal{S}'}\eta_{t}(\beta+w)+\eta_{t}(\beta))\sum_{a \in \mathcal{S}}\eta_{m}(F^{*}(\alpha)-a)\\
&=\left\{
\begin{split}
(-1)^{m+t}\epsilon^{m+t}p^{\frac{m-t}{2}}(p^t\delta_{0}(\beta)-1)\frac{p^m-1}{2}, & \ \text{ if } F^{*}(\alpha)=0, \\
-(-1)^{m+t}\epsilon^{m+t}p^{\frac{m-t}{2}}(p^t\delta_{0}(\beta)-1), & \ \text{ if } -F^{*}(\alpha) \in \mathcal{S},\\
0, & \ \text{ if } -F^{*}(\alpha) \in \mathcal{N},
\end{split}\right.
\end{split}
\end{equation*}
\end{small}where in the last equation we use Eq. (12) and
\begin{small}
\begin{equation*}
\begin{split}
2\sum_{w \in S'}\eta_{t}(\beta+w)+\eta_{t}(\beta)&=\left\{
\begin{split}
p^t-1, & \ \text{ if } \beta=0,\\
2(|(1+\mathcal{S}')\cap \mathcal{S}'|-|(1+\mathcal{S}')\cap \mathcal{N}'|)+1, & \ \text{ if } \beta \in \mathcal{S}',\\
-2(|(1+\mathcal{N}') \cap \mathcal{S}'|-|(1+\mathcal{N}')\cap \mathcal{N}'|)-1, & \ \text{ if } \beta \in \mathcal{N}',\\
\end{split}
\right.\\
&=(p^t\delta_{0}(\beta)-1),
\end{split}
\end{equation*}
\end{small}which is obtained by Proposition 10.

Note that $\epsilon^{2t}=\eta_{t}(-1)=\eta_{m}(-1)$ when $\frac{m}{t}$ is odd. By the above arguments, one can obtain that $N_{\mathcal{S}, \alpha, \beta}$ and $wt(c_{\alpha, \beta})$ are the same as the case of $\frac{m}{t}$ being even.

We can see that $wt(c_{\alpha, \beta})=0$ if and only if $\alpha=0, \beta=0$, thus the dimension of $C_{D_{F, \mathcal{S}}}$ is $\frac{n}{t}+1$. Furthermore, the weight distribution of $C_{D_{F, \mathcal{S}}}$ can be easily obtained by Lemma 11. When $\alpha=0$, since $p \mid |D_{F, \mathcal{S}}|$, $c_{\alpha, \beta}\cdot c_{\alpha, \beta}=\beta^{2}|D_{F, \mathcal{S}}|=0$. When $\alpha\neq 0$, since $p \mid |D_{F, \mathcal{S}}|$ and the values of $N_{\mathcal{S}, \alpha, i}, i \in \mathbb{F}_{p^t}^{*}$, are independent of $i$, and $p \mid N_{\mathcal{S}, \alpha, -i}, i \in \mathbb{F}_{p^t}^{*}$, when $p^t=3$ as $(p^t, n)\neq (3, 3)$, we have $c_{\alpha, \beta}\cdot c_{\alpha, \beta}=\sum_{i \in \mathbb{F}_{p^t}^{*}}N_{\mathcal{S}, \alpha, -i}i^{2}+2\beta\sum_{i \in \mathbb{F}_{p^t}^{*}}N_{\mathcal{S}, \alpha, -i}i+\beta^{2}|D_{F, \mathcal{S}}|=0$. Hence, $C_{D_{F, \mathcal{S}}}$ is self-orthogonal by Proposition 3. With the same argument as in the proof of Theorem 1, $d(C_{D_{F, \mathcal{S}}}^{\bot})\geq 3$. By Proposition 2, $C_{D_{F, \mathcal{S}}}^{\bot}$ is at least almost optimal according to Hamming bound.
\end{proof}

\begin{theorem}\label{Theorem 8}
Let $F$ be a vectorial dual-bent function with Condition III for which $t=m, l=2$, and let $C_{D_{F, I}}$ be defined by Eq. (1) with $I=\gamma H_{b}$, where $\gamma \in \mathbb{F}_{p^m}^{*}$, $H_{b}=\{x^{b}: x \in \mathbb{F}_{p^m}^{*}\}, b \mid (p^m-1)$. Denote $\mathcal{S}=\{x^{2}: x \in \mathbb{F}_{p^m}^{*}\}, \mathcal{N}=\mathbb{F}_{p^m}^{*} \backslash \mathcal{S}$.

$\emph{(i)}$ When $\gamma \in \mathcal{S}$, $b$ is even, $C_{D_{F, I}}$ is a at most six-weight $[(p^{n-m}+\vartheta (-1)^{m-1}\epsilon^{m}\eta_{m}(-1)p^{\frac{n-m}{2}})\frac{p^m-1}{b}$, $\frac{n}{m}+1]_{p^m}$ self-orthogonal linear code whose weight distribution is given in Table 12.

$\emph{(ii)}$ When $\gamma \in \mathcal{N}$, $b$ is even, $C_{D_{F, I}}$ is a at most six-weight $[(p^{n-m}-\vartheta (-1)^{m-1}\epsilon^{m}\eta_{m}(-1)p^{\frac{n-m}{2}})\frac{p^m-1}{b}$, $\frac{n}{m}+1]_{p^m}$ self-orthogonal linear code whose weight distribution is given in Table 13.

$\emph{(iii)}$ When $b$ is odd, $C_{D_{F, I}}$ is a six-weight $[p^{n-m}\frac{p^m-1}{b}, \frac{n}{m}+1]_{p^m}$ self-orthogonal linear code whose weight distribution is given in Table 14.
\end{theorem}

\renewcommand{\thetable}{12}
\begin{table}\label{12}\scriptsize
  \centering
  \caption{The weight distribution of $C_{D_{F, I}}$ constructed in Theorem 8 (i), where $\upsilon=\vartheta (-1)^{m-1}\epsilon^{m}\eta_{m}(-1)$}
  \begin{tabular}{|c|c|}
    \hline
    Hamming weight & Multiplicity \\ \hline
    $0$ & $1$ \\
    $(p^{n-m}+\upsilon p^{\frac{n-m}{2}})\frac{p^m-1}{b}$ & $p^m-1$ \\
    $(p^{n-m}-p^{n-2m})\frac{p^m-1}{b}$ & $p^{n-m}-1$ \\
    $(p^{n-m}-p^{n-2m}+\upsilon p^{\frac{n-m}{2}})\frac{p^m-1}{b}$ & $(p^m-1)(p^{n-m}-1)$ \\
    $(p^{n-m}-p^{n-2m}+\upsilon(p^{\frac{n-m}{2}}+p^{\frac{n-3m}{2}}))\frac{p^m-1}{b}$ & $(p^{n-m}+\vartheta^{-1}(-1)^{m-1}\epsilon^{m}p^{\frac{n-m}{2}})(p^m-\frac{2(p^m-1)}{b})\frac{p^m-1}{2}$ \\
    $(p^{n-m}-p^{n-2m}+\upsilon(p^{\frac{n-m}{2}}-p^{\frac{n-3m}{2}}))\frac{p^m-1}{b}$ & $(p^{n-m}-\vartheta^{-1}(-1)^{m-1}\epsilon^{m}p^{\frac{n-m}{2}})\frac{p^m(p^m-1)}{2}$ \\
    $(p^{n-m}-p^{n-2m}+\upsilon(p^{\frac{n-m}{2}}+p^{\frac{n-3m}{2}}))\frac{p^m-1}{b}-\upsilon p^{\frac{n-m}{2}}$ &
    $(p^{n-m}+\vartheta^{-1}(-1)^{m-1}\epsilon^{m}p^{\frac{n-m}{2}})\frac{(p^m-1)^{2}}{b}$\\
    \hline
  \end{tabular}
\end{table}

\renewcommand{\thetable}{13}
\begin{table}\label{13}\scriptsize
  \centering
  \caption{The weight distribution of $C_{D_{F, I}}$ constructed in Theorem 8 (ii), where $\upsilon=\vartheta (-1)^{m-1}\epsilon^{m}\eta_{m}(-1)$}
  \begin{tabular}{|c|c|}
    \hline
    Hamming weight & Multiplicity  \\ \hline
    $0$ & $1$ \\
    $(p^{n-m}-\upsilon p^{\frac{n-m}{2}})\frac{p^m-1}{b}$ & $p^m-1$ \\
    $(p^{n-m}-p^{n-2m})\frac{p^m-1}{b}$ & $p^{n-m}-1$ \\
    $(p^{n-m}-p^{n-2m}-\upsilon p^{\frac{n-m}{2}})\frac{p^m-1}{b}$ & $(p^m-1)(p^{n-m}-1)$ \\
    $(p^{n-m}-p^{n-2m}-\upsilon(p^{\frac{n-m}{2}}+p^{\frac{n-3m}{2}}))\frac{p^m-1}{b}$ & $(p^{n-m}-\vartheta^{-1}(-1)^{m-1}\epsilon^{m}p^{\frac{n-m}{2}})(p^m-\frac{2(p^m-1)}{b})\frac{p^m-1}{2}$ \\
    $(p^{n-m}-p^{n-2m}-\upsilon(p^{\frac{n-m}{2}}-p^{\frac{n-3m}{2}}))\frac{p^m-1}{b}$ & $(p^{n-m}+\vartheta^{-1}(-1)^{m-1}\epsilon^{m}p^{\frac{n-m}{2}})\frac{p^m(p^m-1)}{2}$ \\
    $(p^{n-m}-p^{n-2m}-\upsilon(p^{\frac{n-m}{2}}+p^{\frac{n-3m}{2}}))\frac{p^m-1}{b}+\upsilon p^{\frac{n-m}{2}}$ &
    $(p^{n-m}-\vartheta^{-1}(-1)^{m-1}\epsilon^{m}p^{\frac{n-m}{2}})\frac{(p^m-1)^{2}}{b}$\\
    \hline
  \end{tabular}
\end{table}

\renewcommand{\thetable}{14}
\begin{table}\label{14}\scriptsize
  \centering
  \caption{The weight distribution of $C_{D_{F, I}}$ constructed in Theorem 8 (iii), where $\upsilon=\vartheta (-1)^{m-1}\epsilon^{m}\eta_{m}(-1)$}
  \begin{tabular}{|c|c|}
    \hline
    Hamming weight & Multiplicity  \\ \hline
    $0$ & $1$ \\
    $p^{n-m}\frac{p^m-1}{b}$ & $p^m-1$ \\
    $(p^{n-m}-p^{n-2m})\frac{p^m-1}{b}$ & $p^m(p^{n-m}-1)$ \\
    $(p^{n-m}-p^{n-2m}+\upsilon p^{\frac{n-3m}{2}})\frac{p^m-1}{b}$ & $(p^{n-m}+\vartheta^{-1} (-1)^{m-1}\epsilon^{m}p^{\frac{n-m}{2}})(p^m-\frac{(p^m-1)}{b})\frac{p^m-1}{2}$ \\
    $(p^{n-m}-p^{n-2m}-\upsilon p^{\frac{n-3m}{2}})\frac{p^m-1}{b}$ & $(p^{n-m}-\vartheta^{-1} (-1)^{m-1}\epsilon^{m}p^{\frac{n-m}{2}})(p^m-\frac{(p^m-1)}{b})\frac{p^m-1}{2}$ \\
    $(p^{n-m}-p^{n-2m}+\upsilon p^{\frac{n-3m}{2}})\frac{p^m-1}{b}-\upsilon p^{\frac{n-m}{2}}$ &
    $(p^{n-m}+\vartheta^{-1} (-1)^{m-1}\epsilon^{m}p^{\frac{n-m}{2}})\frac{(p^m-1)^{2}}{2b}$ \\
    $(p^{n-m}-p^{n-2m}-\upsilon p^{\frac{n-3m}{2}})\frac{p^m-1}{b}+\upsilon p^{\frac{n-m}{2}}$ &
    $(p^{n-m}-\vartheta^{-1} (-1)^{m-1}\epsilon^{m}p^{\frac{n-m}{2}})\frac{(p^m-1)^{2}}{2b}$\\
    \hline
  \end{tabular}
\end{table}

\begin{proof}
Since $D_{F, \gamma H_{b}}=\cup_{a \in \gamma H_{b}}D_{F, a}, D_{F, a}\cap D_{F, a'}=\emptyset, a\neq a'$, and $C_{D_{F, a}}, a \in \gamma H_{b}$, are all self-orthogonal by Theorem 6, we have that $C_{D_{F, \gamma H_{b}}}$ is self-orthogonal. It is clear that the dimension of $C_{D_{F, \gamma H_{b}}}$ is $\frac{n}{m}+1$ since the dimension of $C_{D_{F, a}}$ is $\frac{n}{m}+1$ for any $a \in \gamma H_{b}$.

For the weight distribution of $C_{D_{F, \gamma H_{b}}}$, we only give the proof for cases (i) and (iii) since the proof of case (ii) is similar to case (i).

(i) By Lemma 11, when $\gamma \in \mathcal{S}$ and $b$ is even, the length of $C_{D_{F, \gamma H_{b}}}$ is $|D_{F, \gamma H_{b}}|=p^{n-m}\frac{p^m-1}{b}+\vartheta (-1)^{m-1}\epsilon^{m}p^{\frac{n-m}{2}}\sum_{a \in \gamma H_{b}}\eta_{m}(-a)=(p^{n-m}+\vartheta (-1)^{m-1}\epsilon^{m}\eta_{m}(-1)p^{\frac{n-m}{2}})\frac{p^m-1}{b}$. When $\alpha=0, \beta=0$, $wt(c_{\alpha, \beta})=0$. When $\alpha=0, \beta\neq 0$, $wt(c_{\alpha, \beta})=|D_{F, \gamma H_{b}}|$. When $\alpha \in V_{n}^{(p)} \backslash \{0\}, \beta \in \mathbb{F}_{p^m}$, $wt(c_{\alpha, \beta})=|D_{F, \gamma H_{b}}|-N_{\gamma H_{b}, \alpha, \beta}$, where $N_{\gamma H_{b}, \alpha, \beta}=|\{x \in V_{n}^{(p)}: F(x) \in \gamma H_{b}, \sum_{j=1}^{s}Tr_{m}^{n_{j}}(\alpha_{j}x_{j})+\beta=0\}|$.

By Lemma 12 and Proposition 8, for any $\alpha \in V_{n}^{(p)} \backslash \{0\}, \beta \in \mathbb{F}_{p^m}$, we have
\begin{small}
\begin{equation*}
\begin{split}
&N_{\gamma H_{b}, \alpha, \beta}\\
&=\sum_{a \in \gamma H_{b}}|\{x \in V_{n}^{(p)}: F(x)=a, \sum_{j=1}^{s}Tr_{m}^{n_{j}}(\alpha_{j}x_{j})+\beta=0\}|\\
&=\vartheta (-1)^{m-1}\epsilon^{m}p^{\frac{n-3m}{2}}\sum_{a \in \gamma H_{b}}\sum_{y \in \mathbb{F}_{p^m}^{*}}\eta_{m}(F^{*}(\alpha)y^{2}-\beta y-a)+p^{n-2m}\frac{p^m-1}{b}\\
& \ \ +\vartheta (-1)^{m-1}\epsilon^{m}p^{\frac{n-3m}{2}}\sum_{a \in \gamma H_{b}}\eta_{m}(-a)\\
&=\vartheta (-1)^{m-1}\epsilon^{m}p^{\frac{n-3m}{2}}\sum_{a \in \gamma H_{b}}\sum_{y \in \mathbb{F}_{p^m}}\eta_{m}(F^{*}(\alpha)y^{2}-\beta y-a)+p^{n-2m}\frac{p^m-1}{b}\\
&=\left\{
\begin{split}
(\vartheta (-1)^{m-1}\epsilon^{m}\eta_{m}(-1)p^{\frac{n-m}{2}}+p^{n-2m})\frac{p^m-1}{b}, & \ \text{ if } F^{*}(\alpha)=0, \beta=0, \\
p^{n-2m}\frac{p^m-1}{b}, & \ \text{ if } F^{*}(\alpha)=0, \beta\neq 0, \\
(-\vartheta (-1)^{m-1}\epsilon^{m}\eta_{m}(F^{*}(\alpha))p^{\frac{n-3m}{2}}+p^{n-2m})\frac{p^m-1}{b}, & \ \text{ if } F^{*}(\alpha) \neq 0, \frac{\beta^{2}}{-4F^{*}(\alpha)} \notin \gamma H_{b},\\
\vartheta (-1)^{m-1}\epsilon^{m}\eta_{m}(F^{*}(\alpha))p^{\frac{n-3m}{2}}(p^m-\frac{p^m-1}{b})+p^{n-2m}\frac{p^m-1}{b}, & \ \text{ if } F^{*}(\alpha) \neq 0, \frac{\beta^{2}}{-4F^{*}(\alpha)} \in \gamma H_{b},
\end{split}
\right.\\
&=\left\{
\begin{split}
& (\vartheta (-1)^{m-1}\epsilon^{m}\eta_{m}(-1)p^{\frac{n-m}{2}}+p^{n-2m})\frac{p^m-1}{b}, \ \  \ \text{ if } F^{*}(\alpha)=0, \beta=0, \\
& p^{n-2m}\frac{p^m-1}{b}, \ \ \ \ \ \ \ \ \ \ \ \ \ \ \ \ \ \ \ \ \ \ \ \ \ \ \ \ \ \ \ \ \ \ \ \ \ \ \ \ \ \text{ if } F^{*}(\alpha)=0, \beta\neq 0, \\
& (-\vartheta (-1)^{m-1}\epsilon^{m}\eta_{m}(-1)p^{\frac{n-3m}{2}}+p^{n-2m})\frac{p^m-1}{b}, \\
& \ \ \ \ \ \ \ \ \text{ if } -F^{*}(\alpha) \in \mathcal{S}, \beta=0, \text{ or } -F^{*}(\alpha) \in \mathcal{S}, \beta\neq 0, \frac{\beta^{2}}{-4F^{*}(\alpha)} \notin \gamma H_{b},\\
& (\vartheta (-1)^{m-1}\epsilon^{m}\eta_{m}(-1)p^{\frac{n-3m}{2}}+p^{n-2m})\frac{p^m-1}{b}, \ \ \ \ \ \ \text{ if } -F^{*}(\alpha) \in \mathcal{N}, \\
& \vartheta (-1)^{m-1}\epsilon^{m}\eta_{m}(-1)p^{\frac{n-3m}{2}}(p^m-\frac{p^m-1}{b})+p^{n-2m}\frac{p^m-1}{b}, \\
& \ \ \ \ \ \ \ \ \ \ \ \ \ \ \ \ \ \ \ \ \ \ \ \ \ \ \ \ \ \ \ \ \ \ \ \ \ \text{ if } -F^{*}(\alpha) \in \mathcal{S}, \beta\neq 0, \frac{\beta^{2}}{-4F^{*}(\alpha)} \in \gamma H_{b},
\end{split}
\right.
\end{split}
\end{equation*}
\end{small}and then\\
\begin{small}
\begin{equation*}
\begin{split}
&wt(c_{\alpha, \beta})\\
&=\left\{
\begin{split}
& (p^{n-m}-p^{n-2m})\frac{p^m-1}{b}, \ \ \ \ \ \ \ \ \ \ \ \ \ \ \ \ \ \ \ \ \ \ \ \ \ \ \ \ \ \ \ \ \ \ \ \ \ \ \ \ \ \ \ \ \ \text{ if } F^{*}(\alpha)=0, \beta=0, \\
& (p^{n-m}-p^{n-2m}+\vartheta (-1)^{m-1}\epsilon^{m}\eta_{m}(-1)p^{\frac{n-m}{2}})\frac{p^m-1}{b},  \ \ \ \ \ \ \ \ \ \text{ if } F^{*}(\alpha)=0, \beta\neq 0, \\
& (p^{n-m}-p^{n-2m}+\vartheta (-1)^{m-1}\epsilon^{m}\eta_{m}(-1)(p^{\frac{n-m}{2}}+p^{\frac{n-3m}{2}}))\frac{p^m-1}{b}, \\
& \ \ \ \ \ \ \ \ \ \ \ \ \ \ \ \ \ \ \text{ if } -F^{*}(\alpha) \in \mathcal{S}, \beta=0, \text{ or } -F^{*}(\alpha) \in \mathcal{S}, \beta\neq 0, \frac{\beta^{2}}{-4F^{*}(\alpha)} \notin \gamma H_{b},\\
& (p^{n-m}-p^{n-2m}+\vartheta (-1)^{m-1}\epsilon^{m}\eta_{m}(-1)(p^{\frac{n-m}{2}}-p^{\frac{n-3m}{2}}))\frac{p^m-1}{b},  \ \text{ if } -F^{*}(\alpha) \in \mathcal{N}, \\
& (p^{n-m}-p^{n-2m}+\vartheta (-1)^{m-1}\epsilon^{m}\eta_{m}(-1)(p^{\frac{n-m}{2}}+p^{\frac{n-3m}{2}}))\frac{p^m-1}{b}-\vartheta (-1)^{m-1}\epsilon^{m}\eta_{m}(-1)p^{\frac{n-m}{2}},  \\
& \ \ \ \ \ \ \ \ \ \ \ \ \ \ \ \ \ \ \ \ \ \ \ \ \ \ \ \ \ \ \ \ \ \ \ \ \ \ \ \ \ \ \ \ \ \ \ \  \text{ if } -F^{*}(\alpha) \in \mathcal{S}, \beta\neq 0, \frac{\beta^{2}}{-4F^{*}(\alpha)} \in \gamma H_{b}.
\end{split}
\right.
\end{split}
\end{equation*}
\end{small}

By the above equation and Lemma 11, the weight distribution of $C_{D_{F, \gamma H_{b}}}$ can be easily obtained.

(iii) By Lemma 11, when $b$ is odd, the length of $C_{D_{F, \gamma H_{b}}}$ is $|D_{F, \gamma H_{b}}|=p^{n-m}\frac{p^m-1}{b}+\vartheta (-1)^{m-1}\epsilon^{m}$ $p^{\frac{n-m}{2}}\sum_{a \in \gamma H_{b}}\eta_{m}(-a)=p^{n-m}\frac{p^m-1}{b}$. When $\alpha=0, \beta=0$, $wt(c_{\alpha, \beta})=0$. When $\alpha=0, \beta\neq 0$, $wt(c_{\alpha, \beta})=|D_{F, \gamma H_{b}}|$. With the similar computation as in (i), for any $\alpha \in V_{n}^{(p)} \backslash \{0\}, \beta \in \mathbb{F}_{p^m}$,\\
\begin{small}
\begin{equation*}
\begin{split}
N_{\gamma H_{b}, \alpha, \beta}&=\left\{
\begin{split}
& p^{n-2m}\frac{p^m-1}{b}, \text{ if } F^{*}(\alpha)=0, \\
& (-\vartheta (-1)^{m-1}\epsilon^{m}\eta_{m}(-1)p^{\frac{n-3m}{2}}+p^{n-2m})\frac{p^m-1}{b}, \\
& \ \text{ if } -F^{*}(\alpha) \in \mathcal{S}, \beta=0, \text{ or } -F^{*}(\alpha) \in \mathcal{S}, \beta\neq 0, \frac{\beta^{2}}{-4F^{*}(\alpha)} \notin \gamma H_{b},\\
& (\vartheta (-1)^{m-1}\epsilon^{m}\eta_{m}(-1)p^{\frac{n-3m}{2}}+p^{n-2m})\frac{p^m-1}{b}, \\
& \ \text{ if } -F^{*}(\alpha) \in \mathcal{N}, \beta=0, \text{ or } -F^{*}(\alpha) \in \mathcal{N}, \beta\neq 0, \frac{\beta^{2}}{-4F^{*}(\alpha)} \notin \gamma H_{b},\\
& \vartheta (-1)^{m-1}\epsilon^{m}\eta_{m}(-1)p^{\frac{n-3m}{2}}(p^m-\frac{p^m-1}{b})+p^{n-2m}\frac{p^m-1}{b}, \\
& \ \ \ \ \ \ \ \ \ \ \ \ \ \ \ \ \ \ \ \ \ \ \ \ \ \ \ \ \ \ \ \ \ \text{ if } -F^{*}(\alpha) \in \mathcal{S}, \beta\neq 0, \frac{\beta^{2}}{-4F^{*}(\alpha)} \in \gamma H_{b},\\
& -\vartheta (-1)^{m-1}\epsilon^{m}\eta_{m}(-1)p^{\frac{n-3m}{2}}(p^m-\frac{p^m-1}{b})+p^{n-2m}\frac{p^m-1}{b}, \\
& \ \ \ \ \ \ \ \ \ \ \ \ \ \ \ \ \ \ \ \ \ \ \ \ \ \ \ \ \ \ \ \ \ \text{ if } -F^{*}(\alpha) \in \mathcal{N}, \beta\neq 0, \frac{\beta^{2}}{-4F^{*}(\alpha)} \in \gamma H_{b},
\end{split}
\right.
\end{split}
\end{equation*}
\end{small}and then
\begin{small}
\begin{equation*}
\begin{split}
wt(c_{\alpha, \beta})&=\left\{
\begin{split}
& (p^{n-m}-p^{n-2m})\frac{p^m-1}{b}, \text{ if } F^{*}(\alpha)=0, \\
& (p^{n-m}-p^{n-2m}+\vartheta (-1)^{m-1}\epsilon^{m}\eta_{m}(-1)p^{\frac{n-3m}{2}})\frac{p^m-1}{b}, \\
& \ \text{ if } -F^{*}(\alpha) \in \mathcal{S}, \beta=0, \text{ or } -F^{*}(\alpha) \in \mathcal{S}, \beta\neq 0, \frac{\beta^{2}}{-4F^{*}(\alpha)} \notin \gamma H_{b},\\
& (p^{n-m}-p^{n-2m}-\vartheta (-1)^{m-1}\epsilon^{m}\eta_{m}(-1)p^{\frac{n-3m}{2}})\frac{p^m-1}{b}, \\
& \ \text{ if } -F^{*}(\alpha) \in \mathcal{N}, \beta=0, \text{ or } -F^{*}(\alpha) \in \mathcal{N}, \beta\neq 0, \frac{\beta^{2}}{-4F^{*}(\alpha)} \notin \gamma H_{b},\\
& (p^{n-m}-p^{n-2m}+\vartheta (-1)^{m-1}\epsilon^{m}\eta_{m}(-1)p^{\frac{n-3m}{2}})\frac{p^m-1}{b}-\vartheta (-1)^{m-1}\epsilon^{m}\eta_{m}(-1)p^{\frac{n-m}{2}}, \\
& \ \ \ \ \ \ \ \ \ \ \ \ \ \ \ \ \ \ \ \ \ \ \ \ \ \ \ \ \ \ \ \ \ \text{ if } -F^{*}(\alpha) \in \mathcal{S}, \beta\neq 0, \frac{\beta^{2}}{-4F^{*}(\alpha)} \in \gamma H_{b},\\
& (p^{n-m}-p^{n-2m}-\vartheta (-1)^{m-1}\epsilon^{m}\eta_{m}(-1)p^{\frac{n-3m}{2}})\frac{p^m-1}{b}+\vartheta (-1)^{m-1}\epsilon^{m}\eta_{m}(-1)p^{\frac{n-m}{2}}, \\
& \ \ \ \ \ \ \ \ \ \ \ \ \ \ \ \ \ \ \ \ \ \ \ \ \ \ \ \ \ \ \ \ \ \text{ if } -F^{*}(\alpha) \in \mathcal{N}, \beta\neq 0, \frac{\beta^{2}}{-4F^{*}(\alpha)} \in \gamma H_{b}.
\end{split}
\right.
\end{split}
\end{equation*}
\end{small}

By the above equation and Lemma 11, the weight distribution of $C_{D_{F, \gamma H_{b}}}$ can be easily obtained.
\end{proof}

In the following, for general $m$, by the results in \cite{CMP2018Ve,WF2023Ne,WSWF2023Co}, we list some explicit classes of vectorial dual-bent functions $F: V_{n}^{(p)}\rightarrow \mathbb{F}_{p^m}$ satisfying the conditions in Theorems 6-8. Note that when $\frac{n}{m}$ is odd, $\epsilon^{n} \in \{\pm \epsilon^{m}\}$.

\begin{itemize}
  \item Let $p$ be an odd prime, $t, m, s$ be positive integers for which $t \mid m$, $s\geq 3$ is odd, $(p^t, ms)\neq (3, 3)$. By the result in \cite{WF2023Ne} and the proof of Proposition 8 of \cite{WSWF2023Co}, one can see that all non-degenerate quadratic forms $F$ from $\mathbb{F}_{p^m}^{s}$ ($\mathbb{F}_{p^{ms}}$) to $\mathbb{F}_{p^m}$ are vectorial dual-bent functions satisfying Condition III with $l=d=2$. We list two explicit non-degenerate quadratic forms.
   \begin{itemize}
   \item Let $m, n, t$ be positive integers with $t \mid m, m \mid n$, $\frac{n}{m}\geq3$ is odd, $(p^t, n)\neq (3, 3)$, $\alpha \in \mathbb{F}_{p^n}^{*}$. Define $F: \mathbb{F}_{p^n}\rightarrow \mathbb{F}_{p^m}$ as
    \begin{equation}\label{13}
    F(x)=Tr_{m}^{n}(\alpha x^{2}).
    \end{equation}
  Then $F$ is a vectorial dual-bent function satisfying Condition III with $l=d=2, \vartheta=(-1)^{n-1}\epsilon^{n}\eta_{n}(\alpha)$.
   \item Let $m, t, s$ be positive integers with $t \mid m$, $s\geq3$ is odd, $(p^t, ms)\neq (3, 3)$, $\alpha_{i} \in \mathbb{F}_{p^m}^{*}, 1\leq i \leq s$. Define $F: \mathbb{F}_{p^m}^{s}\rightarrow \mathbb{F}_{p^m}$ as
   \begin{equation}\label{14}
   F(x_{1}, \dots, x_{s})=\sum_{i=1}^{s}\alpha_{i}x_{i}^{2}.
   \end{equation}
   Then $F$ is a vectorial dual-bent function satisfying Condition III with $l=d=2, \vartheta=(-1)^{m-1}\epsilon^{ms}\eta_{m}(\alpha_{1} \cdots \alpha_{s})$.

  In details, for $F$ defined by Eq. (13) and Eq. (14), $F$ satisfies the condition in Theorem 6 (i) and Theorem 7; $F$ satisfies the condition in Theorem 6 (ii) and Theorem 8 when $t=m$.
   \end{itemize}
  \item Let $p$ be an odd prime. Let $m, t, n', n''$ be positive integers with $t \mid m, m \mid n', m \mid n''$, $\frac{n'}{m}$ is odd, $(p^t, n'+2n'')\neq (3, 3)$. For $i \in \mathbb{F}_{p^m}$, let $H(i; x): \mathbb{F}_{p^{n'}}\rightarrow \mathbb{F}_{p^m}$ be given by $H(0; x)=Tr_{m}^{n'}(\alpha_{1}x^{2})$, $H(i; x)=Tr_{m}^{n'}(\alpha_{2}x^{2})$ if $i$ is a square in $\mathbb{F}_{p^m}^{*}$, $H(i; x)=Tr_{m}^{n'}(\alpha_{3}x^{2})$ if $i$ is a non-square in $\mathbb{F}_{p^m}^{*}$, where $\alpha_{1}, \alpha_{2}, \alpha_{3}$ are all square elements or all non-square elements in $\mathbb{F}_{p^{n'}}^{*}$. Let $G: \mathbb{F}_{p^{n''}} \times \mathbb{F}_{p^{n''}}\rightarrow \mathbb{F}_{p^m}$ be given by $G(y_{1}, y_{2})=Tr_{m}^{n''}(\beta y_{1}L(y_{2}))$, where $\beta \in \mathbb{F}_{p^{n''}}^{*}$, $L(x)=\sum a_{i}x^{p^{mi}}$ is a $p^{m}$-polynomial over $\mathbb{F}_{p^{n''}}$ inducing a permutation of $\mathbb{F}_{p^{n''}}$. Let $\gamma \in \mathbb{F}_{p^{n''}}^{*}$. Define $F: \mathbb{F}_{p^{n'}} \times \mathbb{F}_{p^{n''}} \times \mathbb{F}_{p^{n''}}\rightarrow \mathbb{F}_{p^m}$ as
      \begin{equation}\label{15}
      F(x, y_{1}, y_{2})=H(Tr_{m}^{n''}(\gamma y_{2}^{2}); x)+G(y_{1}, y_{2}).
      \end{equation}
  Then $F$ is a vectorial dual-bent function satisfying Condition III with $l=d=2, \vartheta=(-1)^{n'-1}\epsilon^{n'}\eta_{n'}(\alpha_{1})$. In details, $F$ satisfies the condition in Theorem 6 (i) and Theorem 7; $F$ satisfies the condition in Theorem 6 (ii) and Theorem 8 when $t=m$.
\end{itemize}

We give some examples to illustrate Theorems 6-8.
\begin{example}\label{Example 4}
Let $p=3, t=1, m=2, n=6$. Let $F: \mathbb{F}_{3^2}^{3}\rightarrow \mathbb{F}_{3^2}$ be defined by $F(x, y_{1}, y_{2})=(1-w^{2})y_{2}^{8}x^{2}+w^{2}x^{2}+y_{1}y_{2}$, where $w$ is a primitive element of $\mathbb{F}_{3^2}$. Then by Eq. (15), $F$ is a vectorial dual-bent function satisfying Condition III with $l=d=2, \vartheta=1$. When $I=\{0\}$, by Theorem 6 (i), the linear code $C_{D_{F, 0}}$ defined by Eq. (1) is a six-weight $[81, 7, 48]_{3}$ self-orthogonal code with weight enumerator $1+360z^{48}+576z^{51}+240z^{54}+720z^{57}+288z^{60}+2z^{81}$, and its dual code is a $[81, 74, 3]_{3}$ linear code which is almost optimal. When $I=\mathcal{N}=\mathbb{F}_{3^2}^{*} \backslash \{x^{2}: x \in \mathbb{F}_{3^2}^{*}\}$, by Theorem 7, the linear code $C_{D_{F, \mathcal{N}}}$ defined by Eq. (1) is a six-weight $[288, 7, 180]_{3}$ self-orthogonal code with weight enumerator $1+160z^{180}+288z^{186}+1080z^{192}+576z^{195}+80z^{216}+2z^{288}$, and its dual code is a $[288, 281, 3]_{3}$ linear code which is at least almost optimal.
\end{example}

\begin{example}\label{Example 5}
Let $p=3, t=2, m=2, n=6$, $w$ be a primitive element of $\mathbb{F}_{3^2}$, and $F: \mathbb{F}_{3^2}^{3}\rightarrow \mathbb{F}_{3^2}$ be given in Example 4. By Theorem 6 (ii), the linear code $C_{D_{F, w}}$ defined by Eq. (1) is a five-weight $[72, 4, 62]_{9}$ self-orthogonal linear code with weight enumerator $1+2016z^{62}+640z^{63}+3240z^{64}+576z^{71}+88z^{72}$, and its dual code is a $[72, 68, 4]_{9}$ linear code which is optimal.
\end{example}

\begin{example}\label{Example 6}
Let $p=5, t=2, m=2, n=6, b=4$. Let $F: \mathbb{F}_{5^2}^{3}\rightarrow \mathbb{F}_{5^2}$ be defined by $F(x, y_{1}, y_{2})=(w-w^{3})y_{2}^{8}x^{2}+w^{3}x^{2}+y_{1}y_{2}$, where $w$ is a primitive element of $\mathbb{F}_{5^2}$. Then by Eq. (15), $F$ is a vectorial dual-bent function satisfying Condition III with $l=d=2, \vartheta=1$. By Theorem 8, the linear code $C_{D_{F, H_{4}}}$ defined by Eq. (1) is a five-weight $[3600, 4, 3444]_{25}$ self-orthogonal code with weight enumerator $1+93600z^{3444}+14976z^{3450}+195000z^{3456}+86400z^{3469}+648z^{3600}$, and its dual code is a $[3600, 3596, 3]_{25}$ linear code which is at least almost optimal.
\end{example}

\section{Comparison}
\label{sec:6}

To the best of our knowledge, Heng, Li and Liu in \cite{HLL2023Te} for the first time considered using ternary bent functions to construct ternary self-orthogonal linear codes. Very recently, Li and Heng in \cite{LH2023Se} showed that two classes of $p$-ary linear codes constructed in \cite{HLL2023Te} are also self-orthogonal for general odd prime $p$. In \cite{WH2023Se}, Wang and Heng used two classes of non-degenerate quadratic forms to construct $q$-ary self-orthogonal codes. In the following, we compare our results with those in \cite{HLL2023Te,LH2023Se,WH2023Se}. We will show that the works on constructing self-orthogonal codes from $p$-ary bent functions in \cite{HLL2023Te,LH2023Se} and non-degenerate quadratic forms with $q$ being odd in \cite{WH2023Se} can be obtained by our results. Moreover, the parameters of the constructed self-orthogonal in this paper are more abundant and flexible.

\begin{itemize}
  \item Let $p$ be odd and $n\geq 4$ be even. By Proposition 1, bent functions $f: \mathbb{F}_{p^n}\rightarrow \mathbb{F}_{p}$ belonging to $\mathcal{RF}$ are vectorial dual-bent functions satisfying Condition II with $t=m=1$.
      \begin{itemize}
        \item Then the self-orthogonal code defined by Eq. (1) with $I=\{0\}$ from bent function $f: \mathbb{F}_{p^n}\rightarrow \mathbb{F}_{p}$ belonging to $\mathcal{RF}$ given in Theorem 4 of \cite{HLL2023Te} for $p=3$ and Theorem 55 of \cite{LH2023Se} for general odd prime $p$ can be obtained by Theorem 3 (i) with $t=m=1$.
        \item When $p=3, m=1$, $\mathcal{S}=\{x^{2}: x \in \mathbb{F}_{p^m}^{*}\}=\{1\}, \mathcal{N}=\mathbb{F}_{p^m}^{*}\backslash \mathcal{S}=\{2\}$. Then the self-orthogonal ternary linear code defined by Eq. (1) with $I=\{1\}$ or $I=\{2\}$ from bent function $f: \mathbb{F}_{3^n}\rightarrow \mathbb{F}_{3}$ belonging to $\mathcal{RF}$ given in Theorem 1 of \cite{HLL2023Te} can be obtained by Corollary 1 with $t=m=1$.
      \end{itemize}
  \item Let $p$ be odd and $n\geq 5$ be odd. By Proposition 1, bent functions $f: \mathbb{F}_{p^n}\rightarrow \mathbb{F}_{p}$ belonging to $\mathcal{RF}$ are vectorial dual-bent functions satisfying Condition III with $t=m=1$.
      \begin{itemize}
        \item Then the self-orthogonal code defined by Eq. (1) with $I=\{0\}$ from bent function $f: \mathbb{F}_{p^n}\rightarrow \mathbb{F}_{p}$ belonging to $\mathcal{RF}$ given in Theorem 5 of \cite{HLL2023Te} for $p=3$ and Theorem 55 of \cite{LH2023Se} for general odd prime $p$ can be obtained by Theorem 6 (i) with $t=m=1$. Besides, Theorem 6 (i) shows that when $n=3, p>3$, the linear code defined by Eq. (1) with $I=\{0\}$ from bent function $f: \mathbb{F}_{p^n} \rightarrow \mathbb{F}_{p}$ belonging to $\mathcal{RF}$ is also self-orthogonal.
        \item When $p=3, m=1$, $\mathcal{S}=\{x^{2}: x \in \mathbb{F}_{p^m}^{*}\}=\{1\}, \mathcal{N}=\mathbb{F}_{p^m}^{*}\backslash \mathcal{S}=\{2\}$. Then the self-orthogonal ternary linear code defined by Eq. (1) with $I=\{1\}$ or $I=\{2\}$ from bent function $f: \mathbb{F}_{3^n}\rightarrow \mathbb{F}_{3}$ belonging to $\mathcal{RF}$ given in Theorem 2 of \cite{HLL2023Te} can be obtained by Theorem 7 with $t=m=1$.
      \end{itemize}
  \item Let $p$ be odd, $t, m, n$ be positive integers with $t=m, 2m \mid n, 2m\neq n$. By the analysis in Section IV, non-degenerate quadratic form $F: \mathbb{F}_{p^n}\rightarrow \mathbb{F}_{p^m}$ defined as $F(x)=Tr_{m}^{\frac{n}{2}}(x^{p^{\frac{n}{2}+1}})$ is a vectorial dual-bent function satisfying Condition II with $l=d=2, \varepsilon=-1$.
    \begin{itemize}
      \item Then the self-orthogonal code defined by Eq. (1) with $I=\{a\}, a \in \mathbb{F}_{p^m}^{*}$, from $F$ given in Theorem 4.5 of \cite{WH2023Se} can be obtained by Theorem 3 (ii).
      \item In Theorem 4.6 of \cite{WH2023Se}, Wang and Heng showed that the linear code defined by Eq. (1) with $I=H_{b}$ from $F$ is self-orthogonal, where $b \mid (p^m-1), H_{b}=\{x^{b}: x \in \mathbb{F}_{p^m}^{*}\}$. However, the weight distribution is open. For $b=2$, or $m=2jj', b \mid (p^j+1)$ (where $j, j'$ are positive integers, and $j$ is the smallest such positive integer), we compute the weight distribution in Corollary 1 and Theorem 5, respectively.
    \end{itemize}
  \item Let $p$ be odd, $t, m, n$ be positive integers for which $t=m$, $m \mid n$, $\frac{n}{m}\geq 3$ is odd, $(p^m, n)\neq (3, 3)$. By the analysis in Section V, non-degenerate quadratic form $F: \mathbb{F}_{p^n}\rightarrow \mathbb{F}_{p^m}$ defined as $F(x)=Tr_{m}^{n}(x^{2})$ is a vectorial dual-bent function satisfying Condition III with $l=d=2, \vartheta=(-1)^{n-1}\epsilon^{n}$. Then the self-orthogonal code defined by Eq. (1) with $I=\{a\}, a \in \mathbb{F}_{p^m}^{*}$, from $F$ given in Theorem 4.7 of \cite{WH2023Se} can be obtained by Theorem 6 (ii), and the self-orthogonal code defined by Eq. (1) with $I=H_{b}$ from $F$ given in Theorem 4.8 of \cite{WH2023Se} can be obtained by Theorem 8.
\end{itemize}

In Section III, we show that vectorial dual-bent functions with Condition I are quite powerful in constructing self-orthogonal codes as for any nonempty $I\subset V_{m}^{(p)}$, the linear code defined by Eq. (1) is self-orthogonal whose weight distribution is completely determined. Moreover, since there is no division restriction on $t$ and $m$, by using vectorial dual-bent functions defined by Eq. (3), we can obtain self-orthogonal codes whose parameters are abundant and flexible. Furthermore, by Theorems 3 (i), 4, 6 (i), 7, when $t \mid m, t\neq m$, the parameters of the self-orthogonal codes from vectorial dual-bent functions with Conditions II and III are different from those in \cite{HLL2023Te,LH2023Se,WH2023Se}.

\section{Applications in LCD codes and quantum codes}
\label{sec:7}

In this section, by using the obtained self-orthogonal codes, some new families of LCD codes and quantum codes are constructed, some of which are at least almost optimal.

\begin{theorem}\label{Theorem 9}
Let $t, m, n_{j}, 1\leq j\leq s$, be positive integers with $t \mid n_{j}, 1\leq j \leq s$, and let $n=\sum_{j=1}^{s}n_{j}$. For $F: \mathbb{F}_{p^{n_{1}}} \times \dots \times \mathbb{F}_{p^{n_{s}}}\rightarrow V_{m}^{(p)}$ and $I\subset V_{m}^{(p)}$, denote $D_{F, I}=\{\mu^{(1)}, \dots, \mu^{(e)}\}$, where $\mu^{(i)}=(\mu^{(i)}_{1}, \dots, \mu^{(i)}_{s}), 1\leq i \leq e$. For any $1\leq j\leq s$, let $\{\gamma_{j, 1}, \dots, \gamma_{j, \frac{n_{j}}{t}}\}$ be a basis of $\mathbb{F}_{p^{n_{j}}}$ over $\mathbb{F}_{p^t}$. Define
\begin{equation*}
G=[I_{\frac{n}{t}+1}, \overline{G}],
\end{equation*}
where $I_{\frac{n}{t}+1}$ is the identity matrix of size $(\frac{n}{t}+1) \times (\frac{n}{t}+1)$, and\\
\begin{small}
\begin{equation*}
\overline{G}=\left(
  \begin{array}{cccc}
     Tr_{t}^{n_{1}}(\gamma_{1, 1}\mu^{(1)}_{1}) & Tr_{t}^{n_{1}}(\gamma_{1, 1}\mu^{(2)}_{1}) & \dots & Tr_{t}^{n_{1}}(\gamma_{1, 1}\mu^{(e)}_{1}) \\
     \vdots & \vdots & \vdots & \vdots \\
     Tr_{t}^{n_{1}}(\gamma_{1, \frac{n_{1}}{t}}\mu^{(1)}_{1}) & Tr_{t}^{n_{1}}(\gamma_{1, \frac{n_{1}}{t}}\mu^{(2)}_{1}) & \dots & Tr_{t}^{n_{1}}(\gamma_{1, \frac{n_{1}}{t}}\mu^{(e)}_{1}) \\
      \vdots & \vdots & \vdots & \vdots \\
     Tr_{t}^{n_{s}}(\gamma_{s, 1}\mu^{(1)}_{s}) & Tr_{t}^{n_{s}}(\gamma_{s, 1}\mu^{(2)}_{s}) & \dots & Tr_{t}^{n_{s}}(\gamma_{s, 1}\mu^{(e)}_{s}) \\
     \vdots & \vdots & \vdots & \vdots \\
     Tr_{t}^{n_{s}}(\gamma_{s, \frac{n_{s}}{t}}\mu^{(1)}_{s}) & Tr_{t}^{n_{s}}(\gamma_{s,\frac{n_{s}}{t}}\mu^{(2)}_{s}) & \dots & Tr_{t}^{n_{s}}(\gamma_{s,\frac{n_{s}}{t}}\mu^{(e)}_{s}) \\
     1 & 1 & \cdots & 1 \\
  \end{array}
\right).
\end{equation*}
\end{small}

$\mathrm{(i)}$ If $F$ is a vectorial dual-bent function with Condition I (resp., Condition II) and $F(0) \notin I$, then $G$ generates a $[(p^{n-m}-\varepsilon p^{\frac{n}{2}-m})|I|+\frac{n}{t}+1, \frac{n}{t}+1]_{p^t}$ LCD code $C$, and its dual code $C^{\bot}$ is a $[(p^{n-m}-\varepsilon p^{\frac{n}{2}-m})|I|+\frac{n}{t}+1, (p^{n-m}-\varepsilon p^{\frac{n}{2}-m})|I|]_{p^t}$ LCD code which is at least almost optimal according to Hamming bound.

$\mathrm{(ii)}$ If $F$ is a vectorial dual-bent function with Condition III, $I=\mathcal{S}$ or $I=\mathcal{N}$, where $\mathcal{S}=\{x^{2}: x \in \mathbb{F}_{p^m}^{*}\}$ and $\mathcal{N}=\mathbb{F}_{p^m}^{*} \backslash \mathcal{S}$, then $G$ generates a $[(p^{n-m}\pm p^{\frac{n-m}{2}})\frac{p^m-1}{2}+\frac{n}{t}+1, \frac{n}{t}+1]_{p^t}$ LCD code $C$, and its dual code $C^{\bot}$ is a $[(p^{n-m}\pm p^{\frac{n-m}{2}})\frac{p^m-1}{2}+\frac{n}{t}+1, (p^{n-m}\pm p^{\frac{n-m}{2}})\frac{p^m-1}{2}]_{p^t}$ LCD code which is at least almost optimal according to Hamming bound.
\end{theorem}

\begin{proof}
(i) It is easy to see that $\overline{G}$ is a generator matrix of $C_{D_{F, I}}$ defined by Eq. (1). By Theorem 1 (resp., Theorem 2), if $F$ is a vectorial dual-bent function with Condition I (resp., Condition II), then $C_{D_{F, I}}$ is self-orthogonal. Then $G$ generates an LCD code $C$ by Proposition 4, and its dual code $C^{\bot}$ is also an LCD code. The length and dimension of $C$ and $C^{\bot}$ follow from Theorem 1 (resp., Theorem 2). We show that $d(C^{\bot})\geq 3$. By the proof of Theorem 1 (resp., Theorem 2), $d(C_{D_{F, I}}^{\bot})\geq 3$. Note that $d(C_{D_{F, I}}^{\bot})\geq 3$ if and only if any two columns of $\overline{G}$ are linearly independent. As easily seen, in order to prove $d(C^{\bot})\geq 3$, we only need to prove that for every $1\leq i\leq e$, $(0, \dots, 0, \dots, 0, \dots, 0, 1)$ and $(Tr_{t}^{n_{1}}(\gamma_{1, 1}\mu^{(i)}_{1}), \dots, Tr_{t}^{n_{1}}(\gamma_{1, \frac{n_{1}}{t}}\mu^{(i)}_{1}), \dots, Tr_{t}^{n_{s}}(\gamma_{s, 1}\mu^{(i)}_{s}),$ $\dots, Tr_{t}^{n_{s}}(\gamma_{s, \frac{n_{s}}{t}}\mu^{(i)}_{s}), 1)$ are linearly independent. If there are $k_{1}, k_{2} \in \mathbb{F}_{p^t}$ such that $k_{1}(0, \dots, 0, $ $\dots, 0, \dots, 0, 1)+k_{2}(Tr_{t}^{n_{1}}(\gamma_{1, 1}\mu^{(i)}_{1}), \dots, Tr_{t}^{n_{1}}(\gamma_{1, \frac{n_{1}}{t}}\mu^{(i)}_{1}), \dots, Tr_{t}^{n_{s}}(\gamma_{s, 1}\mu^{(i)}_{s}), \dots, Tr_{t}^{n_{s}}(\gamma_{s, \frac{n_{s}}{t}}\mu^{(i)}_{s}), $ $1)=(0, \dots, 0, \dots, 0, \dots, 0, 0)$, then
\begin{equation*}
\left\{
\begin{split}
& k_{2}Tr_{t}^{n_{j}}(\gamma_{j, r}\mu^{(i)}_{j})=0 \ \text{ for all } 1\leq j\leq s, 1\leq r\leq \frac{n_{j}}{t},\\
& k_{1}+k_{2}=0.
\end{split}
\right.
\end{equation*}
If $Tr_{t}^{n_{j}}(\gamma_{j, r}\mu^{(i)}_{j})=0$ for any $1\leq j\leq s, 1\leq r\leq \frac{n_{j}}{t}$, since $\{\gamma_{j, r}, 1\leq r\leq \frac{n_{j}}{t}\}$ is a basis of $\mathbb{F}_{p^{n_{j}}}$ over $\mathbb{F}_{p^t}$, we have $\mu^{(i)}_{j}=0$ for any $1\leq j \leq s$, and then $\mu^{(i)}=0$, which contradicts $0 \notin D_{F, I}$. Thus, there is $(j, r)$ such that $Tr_{t}^{n_{j}}(\gamma_{j, r}\mu^{(i)}_{j})\neq0$, and then $k_{2}=k_{1}=0$. Therefore, $d(C^{\bot})\geq 3$. By Proposition 2, $C^{\bot}$ is at least almost optimal according to Hamming bound.

(ii) The proof is the same as (i), we omit it.
\end{proof}

By Theorem 9, in Table 15, we list some LCD codes by using vectorial dual-bent functions defined by Eq. (3), which are optimal according to the Code Tables at http://www.codetables.de/.

\renewcommand{\thetable}{15}
\begin{table}\label{15}\scriptsize
  \centering
  \caption{Some optimal LCD codes $C^{\perp}$ produced by Theorem 9}
  \begin{threeparttable}
  \begin{tabular}{|c|c|}\hline
    Parameter & Condition \\ \hline
    $[163, 156, 3]_{3}$  & $F$ is defined by Eq. (3) with $p=3, t=1, m=2, n'=3$, $I\subseteq V_{2}^{(3)} \backslash \{B(0)\}$ with $|I|=2$ \\ \hline
    $[241, 234, 3]_{3}$  & $F$ is defined by Eq. (3) with $p=3, t=1, m=2, n'=3$, $I\subseteq V_{2}^{(3)} \backslash \{B(0)\}$ with $|I|=3$ \\ \hline
    $[65, 60, 3]_{4}$  & $F$ is defined by Eq. (3) with $p=2, t=2, m=3, n'=4$, $I\subseteq V_{3}^{(2)} \backslash \{B(0)\}$ with $|I|=2$ \\ \hline
    $[95, 90, 3]_{4}$  & $F$ is defined by Eq. (3) with $p=2, t=2, m=3, n'=4$, $I\subseteq V_{3}^{(2)} \backslash \{B(0)\}$ with $|I|=3$ \\ \hline
    $[125, 120, 3]_{4}$  & $F$ is defined by Eq. (3) with $p=2, t=2, m=3, n'=4$, $I\subseteq V_{3}^{(2)} \backslash \{B(0)\}$ with $|I|=4$ \\ \hline
    $[155, 150, 3]_{4}$  & $F$ is defined by Eq. (3) with $p=2, t=2, m=3, n'=4$, $I\subseteq V_{3}^{(2)} \backslash \{B(0)\}$ with $|I|=5$ \\ \hline
    $[185, 180, 3]_{4}$  & $F$ is defined by Eq. (3) with $p=2, t=2, m=3, n'=4$, $I\subseteq V_{3}^{(2)} \backslash \{B(0)\}$ with $|I|=6$ \\ \hline
    $[215, 210, 3]_{4}$  & $F$ is defined by Eq. (3) with $p=2, t=2, m=3, n'=4$, $I=V_{3}^{(2)} \backslash \{B(0)\}$ \\ \hline
    $[17, 14, 3]_{8}$    & $F$ is defined by Eq. (3) with $p=2, t=3, m=2, n'=3$, $I\subseteq V_{2}^{(2)} \backslash \{B(0)\}$ with $|I|=1$ \\ \hline
    $[31, 28, 3]_{8}$    & $F$ is defined by Eq. (3) with $p=2, t=3, m=2, n'=3$, $I\subseteq V_{2}^{(2)} \backslash \{B(0)\}$ with $|I|=2$ \\ \hline
    $[45, 42, 3]_{8}$    & $F$ is defined by Eq. (3) with $p=2, t=3, m=2, n'=3$, $I=V_{2}^{(2)} \backslash \{B(0)\}$ \\ \hline
    $[27, 24, 3]_{9}$    & $F$ is defined by Eq. (3) with $p=3, t=2, m=1, n'=2$, $I\subseteq \mathbb{F}_{3} \backslash \{B(0)\}$ with $|I|=1$ \\ \hline
    $[51, 48, 3]_{9}$    & $F$ is defined by Eq. (3) with $p=3, t=2, m=1, n'=2$, $I=\mathbb{F}_{3}\backslash \{B(0)\}$  \\ \hline
  \end{tabular}
 \end{threeparttable}
\end{table}

Let $C_{D_{F, I}}$ be a $p^t$-ary self-orthogonal code constructed by Theorem 1 or Theorem 2 or Theorem 7. Let $C_{1}=C_{D_{F, I}}^{\perp}, C_{2}=C^{\perp}$, where $C=\{\beta \textbf{1}: \beta \in \mathbb{F}_{p^t}\}\subseteq C_{D_{F, I}}$. Since $C_{D_{F, I}}$ is self-orthogonal, $C_{1}^{\perp}\subseteq C_{1}\subseteq C_{2}$. It is easy to see that the minimum distance of $C_{2}$ is $2$. By Theorems 1, 2, 7 and their proofs, the minimum distance $d(C_{1})\geq 3$. Then by Theorems 1, 2, 7, Proposition 6, and the known vectorial dual-bent functions defined by Eq. (3), Eq. (8)-(10), Eq. (13)-(14), we list the corresponding parameters of quantum codes in Table 16. By Proposition 5, these quantum codes are at least almost optimal according to the quantum Hamming bound.

\renewcommand{\thetable}{16}
\begin{table}\label{16}\scriptsize
\centering
  \caption{The parameters of some $[[l, k, 3]]_{p^t}$ quantum codes which are at least almost optimal}
  \begin{threeparttable}
  \begin{tabular}{|c|c|c|}\hline
    $l$ & $k$ & Condition\\ \hline
     $(p^{n-m}-p^{\frac{n}{2}-m})\lambda$ & $(p^{n-m}-p^{\frac{n}{2}-m})\lambda-\frac{n}{t}-2$ & \makecell{$n, m, t, \lambda$ are positive integers with $2m<n, 2t \mid n, \lambda<p^m$, \\ and when $p=2$, $m\geq 2$} \\ \hline
    $(p^{n-m}-p^{\frac{n}{2}-m})\lambda+p^{\frac{n}{2}}$ & $(p^{n-m}-p^{\frac{n}{2}-m})\lambda+p^{\frac{n}{2}}-\frac{n}{t}-2$ & \makecell{$n, m, t, \lambda$ are positive integers with $2m<n, 2t \mid n, \lambda<p^m$, \\ and when $p=2$, $m\geq 2$} \\ \hline
   $(p^{n-m}+p^{\frac{n}{2}-m})\lambda$ & $(p^{n-m}+p^{\frac{n}{2}-m})\lambda-\frac{n}{t}-2$ & \makecell{$p$ is odd, $n, m, t, \lambda$ are positive integers with $2m \mid n, 2m\neq n$, \\ $t \mid m, \lambda<p^m$} \\ \hline
   $(p^{n-m}+p^{\frac{n}{2}-m})\lambda-p^{\frac{n}{2}}$ & $(p^{n-m}+p^{\frac{n}{2}-m})\lambda-p^{\frac{n}{2}}-\frac{n}{t}-2$ & \makecell{$p$ is odd, $n, m, t, \lambda$ are positive integers with $2m \mid n, 2m\neq n$, \\ $t \mid m, \lambda<p^m$} \\ \hline
   $(p^{n-m}\pm p^{\frac{n-m}{2}})\frac{p^m-1}{2}$ & $(p^{n-m}\pm p^{\frac{n-m}{2}})\frac{p^m-1}{2}-\frac{n}{t}-2$ & \makecell{$p$ is odd, $n, m, t$ are positive integers with $t \mid m, m \mid n$, \\ $\frac{n}{m}\geq 3$ is odd, $(p^t, n)\neq (3, 3)$} \\ \hline
  \end{tabular}
 \end{threeparttable}
\end{table}

In Table 17, we compare the first class of quantum codes given in Table 16 with the known ones in \cite{EdelTable}. It is shown that our pure quantum codes have better parameters than that of known ones in \cite{EdelTable}.

\renewcommand{\thetable}{17}
\begin{table}\label{17}\scriptsize
  \centering
  \caption{Comparing our pure quantum codes given in Table 16 with that in \cite{EdelTable}}
  \begin{threeparttable}
  \begin{tabular}{|c|c|c|}\hline
    Condition & Our quantum codes & Quantum codes in \cite{EdelTable} \\ \hline
    $p=2, t=2, m=3, n=8, \lambda=5$ & $[[150, 144, 3]]_{4}$ & $[[156, 144, 3]]_{4}$ \\ \hline
    $p=2, t=2, m=3, n=8, \lambda=6$ & $[[180, 174, 3]]_{4}$ & $[[189, 174, 3]]_{4}$ \\ \hline
    $p=2, t=2, m=3, n=8, \lambda=7$ & $[[210, 204, 3]]_{4}$ & $[[217, 204, 3]]_{4}$ \\ \hline
    $p=2, t=2, m=4, n=12, \lambda=2$ & $[[504, 496, 3]]_{4}$ & $[[511, 496, 3]]_{4}$ \\ \hline
    $p=2, t=3, m=2, n=6, \lambda=1$  & $[[14, 10, 3]]_{8}$ & $[[16, 10, 3]]_{8}$ \\ \hline
    $p=2, t=3, m=5, n=12, \lambda=2$ & $[[252, 246, 3]]_{8}$ & $[[256, 246, 3]]_{8}$ \\ \hline
    $p=3, t=2, m=3, n=8, \lambda=1$ & $[[240, 234, 3]]_{9}$ & $[[244, 234, 3]]_{9}$ \\ \hline
    $p=3, t=2, m=3, n=8, \lambda=2$ & $[[480, 474, 3]]_{9}$ & $[[484, 474, 3]]_{9}$ \\ \hline
  \end{tabular}
 \end{threeparttable}
\end{table}

\section{Conclusion}
\label{sec:8}
Self-orthogonal codes are an important class of linear codes which have applications in quantum codes, LCD codes, row-self-orthogonal matrices, and even lattices. In this paper, we constructed new families of self-orthogonal codes by using vectorial dual-bent functions.

(1) By Theorem 1 and vectorial dual-bent functions defined by Eq. (3), one can obtain self-orthogonal codes with parameters $[(p^{n-m}-p^{\frac{n}{2}-m})\lambda, \frac{n}{t}+1]_{p^t}$ and $[(p^{n-m}-p^{\frac{n}{2}-m})\lambda+p^{\frac{n}{2}}, \frac{n}{t}+1]_{p^t}$, and the weight distributions are completely determined, where $n, m, t, \lambda$ are positive integers with $2m<n, 2t \mid n, \lambda< p^m$, and when $p=2$, $m\geq 2$. Some optimal linear codes or having best parameters up to now produced by Theorem 1 were listed in Table 2.

(2) By Theorem 2 and vectorial dual-bent functions defined by Eq. (8)-(10), one can obtain self-orthogonal codes with parameters $[(p^{n-m}+p^{\frac{n}{2}-m})\lambda, \frac{n}{t}+1]_{p^t}$ and $[(p^{n-m}+p^{\frac{n}{2}-m})\lambda-p^{\frac{n}{2}}, \frac{n}{t}+1]_{p^t}$, where $p$ is an odd prime, $n, m, t, \lambda$ are positive integers with $2m \mid n, 2m\neq n, t \mid m, \lambda<p^m$. In some cases, we completely determined the weight distributions of the constructed self-orthogonal codes (Theorems 3, 4, 5).

(3) By Theorems 6 (i), 7 and vectorial dual-bent functions defined by Eq. (13)-(14), one can obtain self-orthogonal codes with parameters $[p^{n-m}, \frac{n}{t}+1]_{p^t}$ and $[(p^{n-m}\pm p^{\frac{n-m}{2}})\frac{p^m-1}{2}, \frac{n}{t}+1]_{p^t}$, and the weight distributions are completely determined, where $p$ is an odd prime, $n, m, t$ are positive integers with $t \mid m, m \mid n$, $\frac{n}{m}\geq 3$ is odd, $(p^t, n)\neq (3, 3)$.

(4) We illustrated that the works on constructing $p$-ary self-orthogonal codes from $p$-ary bent functions given in \cite{HLL2023Te,LH2023Se} can be obtained by Theorems 3 (i), 6 (i), 7, Corollary 1. The works on constructing $q$-ary self-orthogonal codes from two classes of non-degenerate quadratic forms with $q$ being odd given in \cite{WH2023Se} can be obtained by Theorems 3 (ii), 6 (ii), 8. In Corollary 1 and Theorem 5, we partially answered an open problem on determining the weight distribution of a class of self-orthogonal codes given in \cite{WH2023Se}. Moreover, the parameters of self-orthogonal codes obtained in this paper are more abundant than those from (vectorial) bent functions given in \cite{HLL2023Te,LH2023Se,WH2023Se}.

(5) By using the obtained self-orthogonal codes, we constructed several classes of LCD codes and quantum codes which are at least almost optimal. Some optimal LCD codes were listed in Table 15, and some quantum codes with better parameters were listed in Table 17.

\section{Appendix}
\subsection{The proof of Lemma 4}
(i) For any nonempty set $I\subset V_{m}^{(p)}$ and $\alpha \in V_{n}^{(p)} \backslash \{0\}, \beta \in \mathbb{F}_{p^t}$, we have
\begin{small}
\begin{equation*}
\begin{split}
N_{I, \alpha, \beta}&=p^{-m-t}\sum_{x \in V_{n}^{(p)}}\sum_{u \in I}\sum_{y \in V_{m}^{(p)}}\zeta_{p}^{\langle F(x)-u, y\rangle_{m}}\sum_{z \in \mathbb{F}_{p^t}}\zeta_{p}^{-Tr_{1}^{t}(z(\sum_{j=1}^{s}Tr_{t}^{n_{j}}(\alpha_{j}x_{j})+\beta))}\\
\end{split}
\end{equation*}
\begin{equation*}
\begin{split}
&=p^{-m-t}\sum_{z \in \mathbb{F}_{p^t}}\zeta_{p}^{-Tr_{1}^{t}(z\beta)}\sum_{u \in I}\sum_{y \in V_{m}^{(p)}}\zeta_{p}^{-\langle u, y\rangle_{m}}\sum_{x \in V_{n}^{(p)}}\zeta_{p}^{\langle y, F(x)\rangle_{m}-\sum_{j=1}^{s}Tr_{1}^{n_{j}}(z\alpha_{j}x_{j})}\\
&=p^{-m-t}\sum_{z \in \mathbb{F}_{p^t}}\zeta_{p}^{-Tr_{1}^{t}(z\beta)}\sum_{u \in I}\sum_{y \in V_{m}^{(p)}\backslash \{0\}}\zeta_{p}^{-\langle u, y\rangle_{m}}W_{F_{y}}(z\alpha)+p^{n-m-t}|I|.
\end{split}
\end{equation*}
\end{small}Since $F$ is a vectorial dual-bent function with Condition I, we have
\begin{small}
\begin{equation*}
\begin{split}
N_{I, \alpha, \beta}&=\varepsilon p^{\frac{n}{2}-m-t}\sum_{u \in I}\sum_{z \in \mathbb{F}_{p^t}}\zeta_{p}^{-Tr_{1}^{t}(z\beta)}\sum_{y \in V_{m}^{(p)}\backslash \{0\}}\zeta_{p}^{\langle y, F^{*}(z\alpha)-u\rangle_{m}}+p^{n-m-t}|I|\\
&=\varepsilon p^{\frac{n}{2}-m-t}\sum_{u \in I}\sum_{z \in \mathbb{F}_{p^t}}\zeta_{p}^{-Tr_{1}^{t}(z\beta)}(p^m\delta_{u}(F^{*}(z\alpha))-1)+p^{n-m-t}|I|\\
&=\varepsilon p^{\frac{n}{2}-t}\sum_{u \in I}\sum_{z \in \mathbb{F}_{p^t}^{*}}\zeta_{p}^{-Tr_{1}^{t}(z\beta)}\delta_{u}(F^{*}(\alpha))+\varepsilon p^{\frac{n}{2}-t}\delta_{I}(F(0))
-\varepsilon p^{\frac{n}{2}-m}|I|\delta_{0}(\beta)+p^{n-m-t}|I|\\
&=\varepsilon p^{\frac{n}{2}-t}\delta_{I}(F^{*}(\alpha))(p^t\delta_{0}(\beta)-1)+\varepsilon p^{\frac{n}{2}-t}\delta_{I}(F(0))
-\varepsilon p^{\frac{n}{2}-m}|I|\delta_{0}(\beta)+p^{n-m-t}|I|,
\end{split}
\end{equation*}
\end{small}where in the third equation we use Lemma 1 that $F^{*}(ax)=F^{*}(x), a \in \mathbb{F}_{p^t}^{*}, x \in V_{n}^{(p)}$, and $F^{*}(0)=F(0)$ by Corollary 2 and Proposition 5 of \cite{CMP2021Ve}.

(ii) When $p=2$, for any nonempty set $I \subset V_{m}^{(2)}$ and $\alpha, \alpha' \in V_{n}^{(2)} \backslash \{0\}, i, i' \in \mathbb{F}_{2^t}^{*}$, we have
\begin{small}
\begin{equation*}
\begin{split}
T&=\sum_{u \in I}\sum_{z \in \mathbb{F}_{2^t}^{*}}\sum_{w \in \mathbb{F}_{2^t}}\delta_{u}(F^{*}(\alpha+z^{-1}w\alpha'))(-1)^{Tr_{1}^{t}(z(i+i'z^{-1}w))}+\sum_{u \in I}\sum_{w \in \mathbb{F}_{2^t}^{*}}\delta_{u}(F^{*}(\alpha'))(-1)^{Tr_{1}^{t}(i'w)}+\delta_{I}(F(0))\\
&=\sum_{u \in I}\sum_{z \in \mathbb{F}_{2^t}^{*}}\sum_{w \in \mathbb{F}_{2^t}}\delta_{u}(F^{*}(\alpha+w\alpha'))(-1)^{Tr_{1}^{t}(z(i+i'w))}+\sum_{u \in I}\delta_{u}(F^{*}(\alpha'))\sum_{w \in \mathbb{F}_{2^t}^{*}}(-1)^{Tr_{1}^{t}(i'w)}+\delta_{I}(F(0))\\
&=\sum_{u \in I}\sum_{w \in \mathbb{F}_{2^t}}\delta_{u}(F^{*}(\alpha+w\alpha'))\sum_{z \in \mathbb{F}_{2^t}^{*}}(-1)^{Tr_{1}^{t}(z(i+i'w))}-\delta_{I}(F^{*}(\alpha'))+\delta_{I}(F(0))\\
&=2^{t}\delta_{I}(F^{*}(\alpha+ii'^{-1}\alpha'))-\sum_{w \in \mathbb{F}_{2^t}}\delta_{I}(F^{*}(\alpha+w\alpha'))-\delta_{I}(F^{*}(\alpha'))+\delta_{I}(F(0)),
\end{split}
\end{equation*}
\end{small}where in the first equation we use $F^{*}(ax)=F^{*}(x), a \in \mathbb{F}_{p^t}^{*}, x \in V_{n}^{(p)}$, and $F^{*}(0)=F(0)$. $\hfill\blacksquare$
\subsection{The proof of Lemma 7}
(i) For any $a \in \mathbb{F}_{p^m}$, $\alpha \in V_{n}^{(p)} \backslash \{0\}$ and $\beta \in \mathbb{F}_{p^t}$, with the same computation as in the proof of Lemma 4, we have
\begin{equation}\label{16}
N_{a, \alpha, \beta}=p^{-m-t}\sum_{z \in \mathbb{F}_{p^t}}\zeta_{p}^{-Tr_{1}^{t}(z\beta)}\sum_{y \in \mathbb{F}_{p^m}^{*}}\zeta_{p}^{-Tr_{1}^{m}(ay)}W_{F_{y}}(z\alpha)+p^{n-m-t}.
\end{equation}
Since $F$ is a vectorial dual-bent function with Condition II, we have
\begin{small}
\begin{equation*}
\begin{split}
N_{a, \alpha, \beta}&=\varepsilon p^{\frac{n}{2}-m-t}\sum_{y \in \mathbb{F}_{p^m}^{*}}\zeta_{p}^{-Tr_{1}^{m}(ay)}\sum_{z \in \mathbb{F}_{p^t}}\zeta_{p}^{Tr_{1}^{m}(y^{1-d}F^{*}(z\alpha))-Tr_{1}^{t}(z\beta)}+p^{n-m-t}\\
&=\varepsilon p^{\frac{n}{2}-m-t}\sum_{y \in \mathbb{F}_{p^m}^{*}}\zeta_{p}^{-Tr_{1}^{m}(ay)}\sum_{z \in \mathbb{F}_{p^t}^{*}}\zeta_{p}^{Tr_{1}^{m}((\frac{z}{y})^{d-1}zF^{*}(\alpha))-Tr_{1}^{t}(z\beta)}+\varepsilon p^{\frac{n}{2}-m-t}(p^{m}\delta_{0}(a)-1)+p^{n-m-t}\\
\end{split}
\end{equation*}
\begin{equation*}
\begin{split}
&=\varepsilon p^{\frac{n}{2}-m-t}\sum_{z \in \mathbb{F}_{p^t}^{*}}\sum_{y \in \mathbb{F}_{p^m}^{*}}\zeta_{p}^{Tr_{1}^{m}(yzF^{*}(\alpha)-ay^{1-l}z)-Tr_{1}^{t}(z\beta)}+\varepsilon p^{\frac{n}{2}-m-t}(p^{m}\delta_{0}(a)-1)+p^{n-m-t}\\
&=\varepsilon p^{\frac{n}{2}-m-t}\sum_{y \in \mathbb{F}_{p^m}^{*}}\sum_{z \in \mathbb{F}_{p^t}}\zeta_{p}^{Tr_{1}^{t}(z(Tr_{t}^{m}(yF^{*}(\alpha)-ay^{1-l})-\beta))}+\varepsilon p^{\frac{n}{2}-t}(\delta_{0}(a)-1)+p^{n-m-t}\\
&=\varepsilon p^{\frac{n}{2}-m}|\{y \in \mathbb{F}_{p^m}^{*}: Tr_{t}^{m}(yF^{*}(\alpha)-ay^{1-l})=\beta\}|+\varepsilon p^{\frac{n}{2}-t}(\delta_{0}(a)-1)+p^{n-m-t},
\end{split}
\end{equation*}
\end{small}where in the second equation we use Lemma 5 that $F^{*}(zx)=z^{d}F^{*}(x), z \in \mathbb{F}_{p^t}^{*}, x \in V_{n}^{(p)}$, and $F^{*}(0)=0$, in the third equation we use that for $z \in \mathbb{F}_{p^t}^{*}$, $y\mapsto (\frac{z}{y})^{d-1}$ is a permutation of $\mathbb{F}_{p^m}^{*}$.

(ii) When $p=2$, for any $a \in \mathbb{F}_{2^m}$ and $\alpha, \alpha' \in V_{n}^{(2)} \backslash \{0\}, i, i' \in \mathbb{F}_{2^t}^{*}$, we have
\begin{small}
\begin{equation*}
\begin{split}
T&=\sum_{z \in \mathbb{F}_{2^t}^{*}, w \in \mathbb{F}_{2^t}}(-1)^{Tr_{1}^{t}(z(i+z^{-1}wi'))}\sum_{y \in \mathbb{F}_{2^m}^{*}}(-1)^{Tr_{1}^{m}(ay)+Tr_{1}^{m}((\frac{z}{y})^{d-1}zF^{*}(\alpha+z^{-1}w\alpha'))}\\
& \ \ \ +\sum_{w \in \mathbb{F}_{2^t}^{*}}(-1)^{Tr_{1}^{t}(i'w)}\sum_{y \in \mathbb{F}_{2^m}^{*}}(-1)^{Tr_{1}^{m}(ay)+Tr_{1}^{m}((\frac{w}{y})^{d-1}wF^{*}(\alpha'))}+2^{m}\delta_{0}(a)-1\\
&=\sum_{z \in \mathbb{F}_{2^t}^{*}, w \in \mathbb{F}_{2^t}}(-1)^{Tr_{1}^{t}(z(i+wi'))}\sum_{y \in \mathbb{F}_{2^m}^{*}}(-1)^{Tr_{1}^{m}(yzF^{*}(\alpha+w\alpha')+ay^{1-l}z)}\\
& \ \ \ +\sum_{w \in \mathbb{F}_{2^t}^{*}}(-1)^{Tr_{1}^{t}(i'w)}\sum_{y \in \mathbb{F}_{2^m}^{*}}(-1)^{Tr_{1}^{m}(ywF^{*}(\alpha')+ay^{1-l}w)}+2^{m}\delta_{0}(a)-1\\
&=\sum_{w \in \mathbb{F}_{2^t}}\sum_{y \in \mathbb{F}_{2^m}^{*}}\sum_{z \in \mathbb{F}_{2^t}^{*}}(-1)^{Tr_{1}^{t}(z(Tr_{t}^{m}(F^{*}(\alpha+w\alpha')y+ay^{1-l})+i+wi'))}\\
& \ \ \ +\sum_{y \in \mathbb{F}_{2^m}^{*}}\sum_{w \in \mathbb{F}_{2^t}^{*}}(-1)^{Tr_{1}^{t}(w(Tr_{t}^{m}(F^{*}(\alpha')y+ay^{1-l})+i'))}+2^{m}\delta_{0}(a)-1\\
&=2^{t}\sum_{w \in \mathbb{F}_{2^t}}|\{y \in \mathbb{F}_{2^m}^{*}: Tr_{t}^{m}(F^{*}(\alpha+w\alpha')y+ay^{1-l})=i+wi'\}|\\
& \ \ +2^{t}|\{y \in \mathbb{F}_{2^m}^{*}: Tr_{t}^{m}(F^{*}(\alpha')y+ay^{1-l})=i'\}|-(2^t+1)(2^m-1)+2^{m}\delta_{0}(a)-1,
\end{split}
\end{equation*}
\end{small}where in the first equation we use $F^{*}(zx)=z^{d}F^{*}(x), z \in \mathbb{F}_{p^t}^{*}, x \in V_{n}^{(p)}$, and $F^{*}(0)=0$. $\hfill\blacksquare$

\subsection{The proof of Lemma 8}
(i) We only need to show that $\delta_{w^{i}H_{b}}(z^{2})=1$ if and only if $z \in w^{\frac{i}{2}}H_{b}$. Note that $-1 \in H_{b}$ since $b$ is odd. If $\delta_{w^{i}H_{b}}(z^{2})=1$, then $z^{2}w^{-i} \in H_{b}$. Since $i$ is even, $b$ is odd, we have $z^{2}w^{-i} \in H_{2b}$. Thus, there is an integer $k$ such that $z^{2}=w^{i+2kb}$, which implies that $z=\pm w^{\frac{i}{2}+kb} \in w^{\frac{i}{2}}H_{b}$. If $z \in w^{\frac{i}{2}}H_{b}$, then there is an integer $k$ such that $z=w^{\frac{i}{2}+kb}$ and $z^{2}=w^{i+2kb} \in w^{i}H_{b}$.

(ii) We only need to show that $\delta_{w^{i}H_{b}}(z^{2})=1$ if and only if $z \in w^{\frac{i+b}{2}}H_{b}$. If $\delta_{w^{i}H_{b}}(z^{2})=1$, since $i$ and $b$ are both odd, there is an odd integer $k$ such that $z^{2}=w^{i+kb}$. Then $z=\pm w^{\frac{i+b}{2}+\frac{k-1}{2}b} \in w^{\frac{i+b}{2}}H_{b}$. If $z \in w^{\frac{i+b}{2}}H_{b}$, then there is an integer $k$ such that $z=w^{\frac{b+i}{2}+kb}$ and $z^{2}=w^{i+(2k+1)b} \in w^{i}H_{b}$.

(iii) When $b$ is even, $H_{\frac{b}{2}}=H_{b} \cup w^{\frac{b}{2}}H_{b}$, and we only need to show that $\delta_{w^{i}H_{b}}(z^{2})=1$ if and only if $z \in w^{\frac{i}{2}}H_{\frac{b}{2}}$. Since $b \mid (p^m-1)$, we have $-1=w^{\frac{p^m-1}{2}} \in H_{\frac{b}{2}}$. If $\delta_{w^{i}H_{b}}(z^{2})=1$, then there is an integer $k$ such that $z^{2}=w^{i+kb}$ and $z=\pm w^{\frac{i}{2}+k\frac{b}{2}} \in w^{\frac{i}{2}}H_{\frac{b}{2}}$. If $z \in w^{\frac{i}{2}}H_{\frac{b}{2}}$, then there is an integer $k$ such that $z=w^{\frac{i}{2}+k\frac{b}{2}}$ and $z^{2}=w^{i+kb} \in w^{i}H_{b}$.

(iv) Since $i$ is odd and $b$ is even, $\delta_{w^{i}H_{b}}(z^{2})=0$ for any $z \in \mathbb{F}_{p^m}^{*}$.

(v) When $p=2$, we have
\begin{small}
\begin{equation*}
\begin{split}
X&=\sum_{z \in \mathbb{F}_{2^m}^{*}}(-1)^{Tr_{1}^{m}(z\beta)}\delta_{w^{i}H_{b}}(z^{2})=\sum_{z \in \mathbb{F}_{2^m}^{*}}(-1)^{Tr_{1}^{m}(z^{2^{m-1}}\beta)}\delta_{w^{i}H_{b}}(z)\\
&=\sum_{z \in \mathbb{F}_{2^m}^{*}}(-1)^{Tr_{1}^{m}(z\beta^{2})}\delta_{w^{i}H_{b}}(z)=\sum_{z \in H_{b}}(-1)^{Tr_{1}^{m}(zw^{i}\beta^{2})}.
\end{split}
\end{equation*}
\end{small}
$\hfill\blacksquare$

\subsection{The proof of Lemma 9}

Note that when $p=2$, then $b$ is odd, that is, when $b$ is even, then $p$ is odd.

(i) When $F^{*}(\alpha)=0, \beta=0$, obviously $T=0$. When $b$ is even, $F^{*}(\alpha) \neq 0, \beta=0, \gamma F^{*}(\alpha)^{-1} \in \mathcal{N}$, $T=\sum_{a \in \gamma H_{b}}|\{y \in \mathbb{F}_{p^m}^{*}: y^{2}=F^{*}(\alpha)^{-1}a\}|=2|\gamma F^{*}(\alpha)^{-1}H_{b}\cap \mathcal{S}|=0$.

(ii) When $F^{*}(\alpha)=0, \beta\neq 0$, obviously $T=\frac{p^m-1}{b}$. When $b, p$ are odd, $F^{*}(\alpha)\neq 0, \beta=0$, $T=\sum_{a \in \gamma H_{b}}|\{y \in \mathbb{F}_{p^m}^{*}: y^{2}=F^{*}(\alpha)^{-1}a\}|=2|\gamma F^{*}(\alpha)^{-1}H_{b}\cap \mathcal{S}|=2\cdot \frac{p^m-1}{2b}=\frac{p^m-1}{b}$.
If $p=2$, $F^{*}(\alpha)\neq 0, \beta=0$, since $y^{2}$ is a permutation over $\mathbb{F}_{2^m}^{*}$, $T=\frac{2^m-1}{b}$.

(iii) When $b$ is even, $F^{*}(\alpha) \neq 0, \beta=0, \gamma F^{*}(\alpha)^{-1} \in \mathcal{S}$, $T=\sum_{a \in \gamma H_{b}}|\{y \in \mathbb{F}_{p^m}^{*}: y^{2}=F^{*}(\alpha)^{-1}a\}|=2|\gamma F^{*}(\alpha)^{-1}H_{b}\cap \mathcal{S}|=\frac{2(p^m-1)}{b}$.

(iv) Note that $-1 \in H_{b}$ since $b$ is odd. When $b$ is odd, $F^{*}(\alpha) \neq 0, \beta\neq 0$,
\begin{small}
\begin{equation*}
\begin{split}
T&=p^{-m}\sum_{a \in \gamma H_{b}}\sum_{y \in \mathbb{F}_{p^m}^{*}}\sum_{z \in \mathbb{F}_{p^m}}\zeta_{p}^{Tr_{1}^{m}(z(F^{*}(\alpha)y^{2}-\beta y-a))}\\
&=p^{-m}\sum_{y, z \in \mathbb{F}_{p^m}^{*}}\zeta_{p}^{Tr_{1}^{m}(F^{*}(\alpha)y^{2}z-\beta yz)}\sum_{a \in H_{b}}\zeta_{p}^{Tr_{1}^{m}(z\gamma a)}+p^{-m}(p^m-1)\frac{p^m-1}{b}.
\end{split}
\end{equation*}
\end{small}Since $b$ is odd, $\frac{p^j+1}{b}$ is even when $p$ is odd. For any $a \in \mathbb{F}_{p^m}^{*}$, by Proposition 9,
\begin{equation}\label{17}
\sum_{x \in H_{b}}\zeta_{p}^{Tr_{1}^{m}(ax)}=\delta_{H_{b}}(a)(-1)^{j'+1}p^{\frac{m}{2}}+\frac{(-1)^{j'}p^{\frac{m}{2}}-1}{b}.
\end{equation}
By Eq. (17), we have
\begin{small}
\begin{equation*}
\begin{split}
T&=(-1)^{j'+1}p^{-\frac{m}{2}}\sum_{y, z \in \mathbb{F}_{p^m}^{*}}\delta_{H_{b}}(z\gamma)\zeta_{p}^{Tr_{1}^{m}(F^{*}(\alpha)y^{2}z-\beta yz)}+\frac{(-1)^{j'}p^{\frac{m}{2}}-1}{b}p^{-m}\sum_{y \in \mathbb{F}_{p^m}^{*}}\sum_{z \in \mathbb{F}_{p^m}^{*}}\zeta_{p}^{Tr_{1}^{m}(z(F^{*}(\alpha) y^{2}-\beta y))}\\
& \ \ \ +p^{-m}(p^m-1)\frac{p^m-1}{b}\\
\end{split}
\end{equation*}
\begin{equation}\label{18}
\begin{split}
&=(-1)^{j'+1}p^{-\frac{m}{2}}\sum_{y, z \in \mathbb{F}_{p^m}^{*}}\delta_{H_{b}}(y^{-1}z\gamma)\zeta_{p}^{Tr_{1}^{m}(F^{*}(\alpha)yz-\beta z)}+\frac{(-1)^{j'}p^{\frac{m}{2}}-1}{b}|\{y \in \mathbb{F}_{p^m}^{*}: F^{*}(\alpha) y^{2}-\beta y=0\}|\\
&\ \ \ -\frac{(-1)^{j'}p^{\frac{m}{2}}-1}{b}p^{-m}(p^m-1)+p^{-m}(p^m-1)\frac{p^m-1}{b}\\
&=(-1)^{j'+1}p^{-\frac{m}{2}}\sum_{z \in \mathbb{F}_{p^m}^{*}}\zeta_{p}^{Tr_{1}^{m}(-z\beta)}\sum_{y \in \mathbb{F}_{p^m}^{*}}\delta_{\gamma zH_{b}}(y)\zeta_{p}^{Tr_{1}^{m}(F^{*}(\alpha)yz)}+\frac{(-1)^{j'}p^{\frac{m}{2}}-1}{b}p^{-m}+p^{-m}(p^m-1)\frac{p^m-1}{b}\\
&=(-1)^{j'+1}p^{-\frac{m}{2}}\sum_{z \in \mathbb{F}_{p^m}^{*}}\zeta_{p}^{Tr_{1}^{m}(-z\beta)}\sum_{y \in H_{b}}\zeta_{p}^{Tr_{1}^{m}(\gamma z^{2}F^{*}(\alpha)y)}+\frac{(-1)^{j'}p^{\frac{m}{2}}-1}{b}p^{-m}+p^{-m}(p^m-1)\frac{p^m-1}{b}\\
&=(-1)^{j'+1}p^{-\frac{m}{2}}R+\frac{(-1)^{j'}p^{\frac{m}{2}}-1}{b}p^{-m}+p^{-m}(p^m-1)\frac{p^m-1}{b},
\end{split}
\end{equation}
\end{small}where $R=\sum_{z \in \mathbb{F}_{p^m}^{*}}\zeta_{p}^{Tr_{1}^{m}(-z\beta)}\sum_{y \in H_{b}}\zeta_{p}^{Tr_{1}^{m}(\gamma z^{2}F^{*}(\alpha)y)}$.

When $p$ is odd, by Eq. (17) and Lemma 8, we have
\begin{small}
\begin{equation}\label{19}
\begin{split}
R&=\sum_{z \in \mathbb{F}_{p^m}^{*}}\zeta_{p}^{Tr_{1}^{m}(-z\beta)}((-1)^{j'+1}\delta_{H_{b}}(\gamma z^{2}F^{*}(\alpha))p^{\frac{m}{2}}+\frac{(-1)^{j'}p^{\frac{m}{2}}-1}{b})\\
&=(-1)^{j'+1}p^{\frac{m}{2}}\sum_{z \in \mathbb{F}_{p^m}^{*}}\zeta_{p}^{Tr_{1}^{m}(-z\beta)}\delta_{\gamma^{-1}F^{*}(\alpha)^{-1}H_{b}}(z^{2})-\frac{(-1)^{j'}p^{\frac{m}{2}}-1}{b}\\
&=\left\{
\begin{split}
(-1)^{j'+1}p^{\frac{m}{2}}\sum_{z \in H_{b}}\zeta_{p}^{Tr_{1}^{m}(-z\sqrt{\gamma^{-1}F^{*}(\alpha)^{-1}}\beta)}-\frac{(-1)^{j'}p^{\frac{m}{2}}-1}{b}, & \text{ if } \gamma^{-1}F^{*}(\alpha)^{-1}\in \mathcal{S},\\
(-1)^{j'+1}p^{\frac{m}{2}}\sum_{z \in H_{b}}\zeta_{p}^{Tr_{1}^{m}(-z\sqrt{\gamma^{-1}F^{*}(\alpha)^{-1}w^{b}}\beta)}-\frac{(-1)^{j'}p^{\frac{m}{2}}-1}{b}, & \text{ if } \gamma^{-1}F^{*}(\alpha)^{-1}\in \mathcal{N},\\
\end{split}
\right.\\
&=\left\{
\begin{split}
&(-1)^{j'+1}p^{\frac{m}{2}}((-1)^{j'+1}p^{\frac{m}{2}}\delta_{H_{b}}(\beta \sqrt{\gamma^{-1}F^{*}(\alpha)^{-1}})+\frac{(-1)^{j'}p^{\frac{m}{2}}-1}{b})-\frac{(-1)^{j'}p^{\frac{m}{2}}-1}{b},\\
& \ \ \ \ \ \ \ \ \ \ \ \ \ \ \ \ \ \ \ \ \ \ \ \ \ \ \ \ \ \ \ \ \ \ \ \ \ \ \ \ \ \ \ \ \ \ \ \ \ \ \ \ \ \ \ \ \ \ \ \ \ \ \ \ \ \ \ \ \ \ \text{ if }\gamma^{-1}F^{*}(\alpha)^{-1}\in \mathcal{S},\\
&(-1)^{j'+1}p^{\frac{m}{2}}((-1)^{j'+1}p^{\frac{m}{2}}\delta_{H_{b}}(\beta \sqrt{\gamma^{-1}F^{*}(\alpha)^{-1}w^{b}})+\frac{(-1)^{j'}p^{\frac{m}{2}}-1}{b})-\frac{(-1)^{j'}p^{\frac{m}{2}}-1}{b},\\
& \ \ \ \ \ \ \ \ \ \ \ \ \ \ \ \ \ \ \ \ \ \ \ \ \ \ \ \ \ \ \ \ \ \ \ \ \ \ \ \ \ \ \ \ \ \ \ \ \ \ \ \ \ \ \ \ \ \ \ \ \ \ \ \ \ \ \ \ \ \ \text{ if }\gamma^{-1}F^{*}(\alpha)^{-1}\in \mathcal{N},\\
\end{split}
\right.\\
&=\left\{
\begin{split}
p^{m}-\frac{p^m-1}{b}, & \ \text{ if } \gamma^{-1}F^{*}(\alpha)^{-1}\in \mathcal{S}, \beta \sqrt{\gamma^{-1}F^{*}(\alpha)^{-1}} \in H_{b} \\
 & \  \text{ or }\gamma^{-1}F^{*}(\alpha)^{-1}\in \mathcal{N}, \beta \sqrt{\gamma^{-1}F^{*}(\alpha)^{-1}w^{b}} \in H_{b},\\
-\frac{p^m-1}{b}, & \ \text{ if } \gamma^{-1}F^{*}(\alpha)^{-1}\in \mathcal{S}, \beta \sqrt{\gamma^{-1}F^{*}(\alpha)^{-1}} \notin H_{b} \\
 & \  \text{ or }\gamma^{-1}F^{*}(\alpha)^{-1}\in \mathcal{N}, \beta \sqrt{\gamma^{-1}F^{*}(\alpha)^{-1}w^{b}} \notin H_{b}.
\end{split}
\right.
\end{split}
\end{equation}
\end{small}
Combine Eq. (18) and Eq. (19), the result holds.

When $p=2$, by Eq. (17) and Lemma 8, we have
\begin{small}
\begin{equation}\label{20}
\begin{split}
R&=\sum_{z \in \mathbb{F}_{2^m}^{*}}(-1)^{Tr_{1}^{m}(z\beta)}((-1)^{j'+1}\delta_{H_{b}}(\gamma z^{2}F^{*}(\alpha))2^{\frac{m}{2}}+\frac{(-1)^{j'}2^{\frac{m}{2}}-1}{b})\\
&=(-1)^{j'+1}2^{\frac{m}{2}}\sum_{z \in \mathbb{F}_{2^m}^{*}}(-1)^{Tr_{1}^{m}(z\beta)}\delta_{\gamma^{-1}F^{*}(\alpha)^{-1}H_{b}}(z^{2})-\frac{(-1)^{j'}2^{\frac{m}{2}}-1}{b}\\
&=(-1)^{j'+1}2^{\frac{m}{2}}\sum_{z \in H_{b}}(-1)^{Tr_{1}^{m}(\gamma^{-1}F^{*}(\alpha)^{-1} \beta^{2}z)}-\frac{(-1)^{j'}2^{\frac{m}{2}}-1}{b}\\
&=(-1)^{j'+1}2^{\frac{m}{2}}((-1)^{j'+1}2^{\frac{m}{2}}\delta_{H_{b}}(\gamma^{-1}F^{*}(\alpha)^{-1} \beta^{2})+\frac{(-1)^{j'}2^{\frac{m}{2}}-1}{b})-\frac{(-1)^{j'}2^{\frac{m}{2}}-1}{b}\\
&=\left\{
\begin{split}
2^{m}-\frac{2^m-1}{b}, & \ \text{ if } \gamma^{-1}F^{*}(\alpha)^{-1} \beta^{2} \in H_{b},\\
-\frac{2^m-1}{b}, & \ \text{ if } \gamma^{-1}F^{*}(\alpha)^{-1} \beta^{2} \notin H_{b}.
\end{split}
\right.\\
\end{split}
\end{equation}
\end{small}Combine Eq. (18) and Eq. (20), the result holds.

(v) When $b$ is even, $p$ is odd. If $j'$ is odd, $b \mid (p^{jj'}+1)$ and $2b \mid (p^m-1)$; if $j'$ is even, $b \mid (p^{jj'}-1)$ and $2b \mid (p^m-1)$. Hence, $-1 \in H_{b}$ and then with the same computation as in (iv),
\begin{equation*}
\begin{split}
T&=p^{-m}\sum_{y, z \in \mathbb{F}_{p^m}^{*}}\zeta_{p}^{Tr_{1}^{m}(F^{*}(\alpha)y^{2}z-\beta yz)}\sum_{a \in H_{b}}\zeta_{p}^{Tr_{1}^{m}(z\gamma a)}+p^{-m}(p^m-1)\frac{p^m-1}{b}.
\end{split}
\end{equation*}
We only prove the case that $j'$ and $\frac{p^j+1}{b}$ are both odd, since the other case is similar. When $j'$ and $\frac{p^j+1}{b}$ are both odd, for any $a \in \mathbb{F}_{p^m}^{*}$, by Proposition 9,
\begin{equation}\label{21}
\sum_{x \in H_{b}}\zeta_{p}^{Tr_{1}^{m}(ax)}=\delta_{w^{\frac{b}{2}}H_{b}}(a)p^{\frac{m}{2}}-\frac{p^{\frac{m}{2}}+1}{b}.
\end{equation}
With Eq. (21) and the similar computation as in (iv),
\begin{equation}\label{22}
\begin{split}
T&=p^{-\frac{m}{2}}R-\frac{p^{\frac{m}{2}}+1}{b}p^{-m}+p^{-m}(p^m-1)\frac{p^m-1}{b},
\end{split}
\end{equation}
where $R=\sum_{z \in \mathbb{F}_{p^m}^{*}}\zeta_{p}^{Tr_{1}^{m}(-z\beta)}\sum_{y \in H_{b}}\zeta_{p}^{Tr_{1}^{m}(\gamma z^{2}w^{\frac{b}{2}}F^{*}(\alpha)y)}$. By Eq. (21) and Lemma 8, we have
\begin{small}
\begin{equation*}
\begin{split}
R&=\sum_{z \in \mathbb{F}_{p^m}^{*}}\zeta_{p}^{Tr_{1}^{m}(-z\beta)}(\delta_{w^{\frac{b}{2}}H_{b}}(\gamma z^{2}w^{\frac{b}{2}}F^{*}(\alpha))p^{\frac{m}{2}}-\frac{p^{\frac{m}{2}}+1}{b})\\
&=p^{\frac{m}{2}}\sum_{z \in \mathbb{F}_{p^m}^{*}}\zeta_{p}^{Tr_{1}^{m}(-z\beta)}\delta_{\gamma^{-1}F^{*}(\alpha)^{-1}H_{b}}(z^{2})+\frac{p^{\frac{m}{2}}+1}{b}\\
&=\left\{
\begin{split}
p^{\frac{m}{2}}\sum_{z \in H_{b}}(\zeta_{p}^{Tr_{1}^{m}(-z\beta \sqrt{\gamma^{-1}F^{*}(\alpha)^{-1}})}+\zeta_{p}^{Tr_{1}^{m}(-z\beta \sqrt{\gamma^{-1}F^{*}(\alpha)^{-1}w^{b}})})+\frac{p^{\frac{m}{2}}+1}{b}, & \ \text{ if } \gamma^{-1}F^{*}(\alpha)^{-1} \in \mathcal{S},\\
\frac{p^{\frac{m}{2}}+1}{b}, & \ \text{ if } \gamma^{-1}F^{*}(\alpha)^{-1} \in \mathcal{N},\\
\end{split}
\right.\\
\end{split}
\end{equation*}
\begin{equation}\label{23}
\begin{split}
&=\left\{
\begin{split}
& p^{\frac{m}{2}}(p^{\frac{m}{2}}\delta_{w^{\frac{b}{2}}H_{b}}(\beta\sqrt{\gamma^{-1}F^{*}(\alpha)^{-1}})
+p^{\frac{m}{2}}\delta_{H_{b}}(\beta\sqrt{\gamma^{-1}F^{*}(\alpha)^{-1}})-\frac{2(p^{\frac{m}{2}}+1)}{b})
+\frac{p^{\frac{m}{2}}+1}{b}, \\
& \ \ \ \ \ \ \ \ \ \ \ \ \text{ if } \gamma^{-1}F^{*}(\alpha)^{-1} \in \mathcal{S},\\
& \frac{p^{\frac{m}{2}}+1}{b}, \ \text{ if } \gamma^{-1}F^{*}(\alpha)^{-1} \in \mathcal{N},\\
\end{split}
\right.\\
&=\left\{
\begin{split}
p^{m}-(2p^{\frac{m}{2}}-1)\frac{p^{\frac{m}{2}}+1}{b}, & \ \text{ if } \gamma^{-1}F^{*}(\alpha)^{-1} \in \mathcal{S}, \beta\sqrt{\gamma^{-1}F^{*}(\alpha)^{-1}} \in H_{\frac{b}{2}},\\
-(2p^{\frac{m}{2}}-1)\frac{p^{\frac{m}{2}}+1}{b}, & \ \text{ if } \gamma^{-1}F^{*}(\alpha)^{-1} \in \mathcal{S}, \beta\sqrt{\gamma^{-1}F^{*}(\alpha)^{-1}} \notin H_{\frac{b}{2}},\\
\frac{p^{\frac{m}{2}}+1}{b}, & \ \text{ if } \gamma^{-1}F^{*}(\alpha)^{-1} \in \mathcal{N}.
\end{split}
\right.\\
\end{split}
\end{equation}
\end{small} Combine Eq. (22) and Eq. (23), the result holds.
$\hfill\blacksquare$
\subsection{The proof of Lemma 12}

With the same computation as in the proof of Lemma 4, Eq. (16) holds. Since $F$ is a vectorial dual-bent function with Condition III, for $a \in \mathbb{F}_{p^m}^{*}, \alpha \in V_{n}^{(p)} \backslash \{0\}, \beta \in \mathbb{F}_{p^t}$, by Eq. (16) we have
\begin{small}
\begin{equation*}
\begin{split}
&N_{a, \alpha, \beta}\\
&=\vartheta p^{\frac{n}{2}-m-t}\sum_{y \in \mathbb{F}_{p^m}^{*}}\eta_{m}(y)\zeta_{p}^{Tr_{1}^{m}(-ay)}\sum_{z \in \mathbb{F}_{p^t}}\zeta_{p}^{Tr_{1}^{m}(y^{1-d}F^{*}(z\alpha))+Tr_{1}^{t}(-z\beta)}+p^{n-m-t}\\
&=\vartheta p^{\frac{n}{2}-m-t}\sum_{y \in \mathbb{F}_{p^m}^{*}}\eta_{m}(y)\zeta_{p}^{Tr_{1}^{m}(-ay)}\sum_{z \in \mathbb{F}_{p^t}^{*}}\zeta_{p}^{Tr_{1}^{m}((\frac{z}{y})^{d-1}zF^{*}(\alpha))+Tr_{1}^{t}(-z\beta)}+\vartheta (-1)^{m-1}\epsilon^{m} \eta_{m}(-a)p^{\frac{n-m}{2}-t}+p^{n-m-t}\\
&=\vartheta p^{\frac{n}{2}-m-t}\sum_{z \in \mathbb{F}_{p^t}^{*}}\zeta_{p}^{Tr_{1}^{t}(-z\beta)}\sum_{y \in \mathbb{F}_{p^m}^{*}}\eta_{m}(y^{1-l}z)\zeta_{p}^{Tr_{1}^{m}(yzF^{*}(\alpha))+Tr_{1}^{m}(-ay^{1-l}z)}+\vartheta (-1)^{m-1}\epsilon^{m} \eta_{m}(-a)p^{\frac{n-m}{2}-t}+p^{n-m-t}\\
&=\vartheta p^{\frac{n}{2}-m-t}\sum_{y \in \mathbb{F}_{p^m}^{*}}\eta_{m}(y)\sum_{z \in \mathbb{F}_{p^t}^{*}}\eta_{m}(z)\zeta_{p}^{Tr_{1}^{t}(z(Tr_{t}^{m}(F^{*}(\alpha)y-ay^{1-l})-\beta))}+\vartheta (-1)^{m-1}\epsilon^{m} \eta_{m}(-a)p^{\frac{n-m}{2}-t}+p^{n-m-t},
\end{split}
\end{equation*}
\end{small}where in the second equation we use Lemma 10 and Proposition 7.

When $\frac{m}{t}$ is even, we have $\eta_{m}(z)=1$ for all $z \in \mathbb{F}_{p^t}^{*}$, and
\begin{small}
\begin{equation*}
  \begin{split}
   N_{a, \alpha, \beta}&=\vartheta p^{\frac{n}{2}-m}\sum_{y \in \mathbb{F}_{p^m}^{*}}\eta_{m}(y)\delta_{0}(Tr_{t}^{m}(F^{*}(\alpha)y-ay^{1-l})-\beta)+\vartheta (-1)^{m-1}\epsilon^{m} \eta_{m}(-a)p^{\frac{n-m}{2}-t}+p^{n-m-t}.
  \end{split}
\end{equation*}
\end{small}When $\frac{m}{t}$ is odd, we have $\eta_{m}(z)=\eta_{t}(z)$ for all $z \in \mathbb{F}_{p^t}^{*}$, and
\begin{small}
\begin{equation*}
  \begin{split}
  N_{a, \alpha, \beta}&=\vartheta(-1)^{t-1}\epsilon^{t}p^{\frac{n-t}{2}-m}\sum_{y \in \mathbb{F}_{p^m}^{*}}\eta_{m}(y)\eta_{t}(Tr_{t}^{m}(F^{*}(\alpha)y-ay^{1-l})-\beta)+\vartheta (-1)^{m-1}\epsilon^{m} \eta_{m}(-a)p^{\frac{n-m}{2}-t}+p^{n-m-t}.
  \end{split}
\end{equation*}
\end{small}
For $a=0, \alpha \in V_{n}^{(p)} \backslash \{0\}, \beta \in \mathbb{F}_{p^t}$, by Eq. (16) we have
\begin{small}
\begin{equation*}
\begin{split}
N_{0, \alpha, \beta}&=\vartheta p^{\frac{n}{2}-m-t}\sum_{y \in \mathbb{F}_{p^m}^{*}}\eta_{m}(y)\sum_{z \in \mathbb{F}_{p^t}}\zeta_{p}^{Tr_{1}^{m}(y^{1-d}F^{*}(z\alpha))+Tr_{1}^{t}(-z\beta)}+p^{n-m-t}\\
&=\vartheta p^{\frac{n}{2}-m-t}\sum_{y \in \mathbb{F}_{p^m}^{*}}\eta_{m}(y)\sum_{z \in \mathbb{F}_{p^t}^{*}}\zeta_{p}^{Tr_{1}^{m}(y^{1-d}z^{d}F^{*}(\alpha))+Tr_{1}^{t}(-z\beta)}+p^{n-m-t}\\
&=\vartheta p^{\frac{n}{2}-m-t}\sum_{z \in \mathbb{F}_{p^t}^{*}}\zeta_{p}^{Tr_{1}^{t}(-z\beta)}\sum_{y \in \mathbb{F}_{p^m}^{*}}\eta_{m}(y^{1-l}z^{(l-1)d})\zeta_{p}^{Tr_{1}^{m}(yF^{*}(\alpha))}+p^{n-m-t}\\
\end{split}
\end{equation*}
\begin{equation*}
\begin{split}
&=\vartheta p^{\frac{n}{2}-m-t}\sum_{z \in \mathbb{F}_{p^t}^{*}}\zeta_{p}^{Tr_{1}^{t}(-z\beta)}\sum_{y \in \mathbb{F}_{p^m}^{*}}\eta_{m}(y)\zeta_{p}^{Tr_{1}^{m}(yF^{*}(\alpha))}+p^{n-m-t}\\
&=\vartheta (-1)^{m-1}\epsilon^{m}\eta_{m}(F^{*}(\alpha))p^{\frac{n-m}{2}-t}(p^{t}\delta_{0}(\beta)-1)+p^{n-m-t}.
\end{split}
\end{equation*}
\end{small}
$\hfill\blacksquare$
\end{document}